\documentclass[12pt]{article}

\usepackage{graphicx}
\usepackage{enumerate}
\usepackage{natbib}
\usepackage{url} 
\usepackage{amsfonts,amsmath,amsbsy,amsthm,amssymb,mathrsfs}
\usepackage{cases} 
\usepackage{anysize,indentfirst,geometry,xcolor}
\usepackage{authblk}

\usepackage{bm}
\usepackage{bigints} 
\usepackage{tikz}
\usepackage{caption}
\usepackage{subcaption}
\usepackage{multicol} 
\usepackage{multirow}
\usepackage{setspace}

\usepackage[figuresright]{rotating}

\usepackage{booktabs}

\usepackage{pdfpages}
\usepackage{lscape}

\def\spacingset#1{\renewcommand{\baselinestretch}%
{#1}\small\normalsize} \spacingset{1}

\usepackage[linesnumbered,ruled,vlined]{algorithm2e}
\SetKwInput{KwInput}{Input}                
\SetKwInput{KwOutput}{Output}  

\theoremstyle{plain}
\newtheorem{thm}{Theorem}
\newtheorem{lem}{Lemma}

\newtheorem{cond}{Condition}

\theoremstyle{remark}

\newcommand\blue[1]{{\color{black}#1}}

\def\tildeT{\widetilde{T}}
\def\tildedelta{\widetilde{\delta}}

\def\hatalpha{\widehat{\alpha}}
\def\Pr{\operatorname{Pr}}
\def\calP{\mathcal{P}}

\title{Learning association from multiple intermediate events for dynamic prediction of survival: an application to  cardiovascular disease prognosis 
}
\author{Tonghui Yu$^1$, Liming Xiang$^2$\thanks{Corresponding author. Email: LMXiang@ntu.edu.sg} \hspace{.2cm}\\ 
$^1$ International Center for Interdisciplinary Statistics, School of Mathematics\\ Harbin Institute of Technology, Harbin, China\\
$^2$ School of Physical and Mathematical Sciences\\ Nanyang Technological University, Singapore}
\date{}

\begin{document}
\maketitle

\begin{abstract}
Cardiovascular diseases are major causes of mortality globally. They often co-occur and are interrelated, leading to partial-order relationships among their onset times. However, these onset times are subject to informative censoring due to the occurrence of death, posing significant challenges for survival prediction. In this article, we propose a novel copula-based framework that learns dependence among multiple correlated marginal components through a pseudo-likelihood for estimation. We adopt nonparametric marginals, alleviating the reliance on marginal distribution assumptions typically required in conventional copula models, and estimate the association between the onsets of intermediate cardiovascular diseases and death by solving a concordance estimating equation. Under this framework, a renewable risk assessment method is developed for dynamic survival prediction, leveraging information on disease onset times and the maximum follow-up duration. Our proposed method yields estimators with well-established properties, and its flexibility and predictive effectiveness are demonstrated through extensive simulation studies. We apply the method to data from a heart disease study, demonstrating the benefits of incorporating the associations among various cardiovascular diseases and their synergistic effects on mortality for dynamic prediction of overall survival.
\end{abstract}


\doublespacing
\section{Introduction}
\label{sec:intro}

\subsection{Motivation}
Cardiovascular diseases including hypertension (HYP), angina pectoris (AP), coronary heart disease (CHD), hospitalized myocardial infarction (MI), fatal coronary heart disease (MIFC), fatal cerebrovascular disease (CVD) and stroke (STRK), are the leading causes of death worldwide as reported in the annual publications of the World Health Organization. 
These different forms of cardiovascular diseases 
often occur synergistically, exacerbate one another and are all correlated with death \citep{georges2021genetic}. 

Our work is motivated by data from the Framingham Heart Study, a longitudinal and community-based study of cardiovascular epidemiology \citep{andersson2019}, where 
patients may experience multiple intermediate cardiovascular diseases before 
death, and the times to these intermediate events and to death are possibly correlated one and another.
Analyzing and modeling the association patterns for these diseases, in addition to the their impacts on overall survival, is of primary 
interest in practice. Advances in cardiovascular disease research suggest the heterogeneous effects of these conditions on patient mortality. For example, modern diagnostic and classification techniques aided with artificial intelligence have been developed to better understand and differentiate these diseases \citep{wang2024screening}; therapeutic and genomic research have identified genetic risk factors and therapeutic targets for specific cardiovascular diseases \citep{
schiano2015epigenetic}. However, 
little is known about the pathogenic mechanisms, the relationships among different forms of cardiovascular diseases and the extent of their contributions to mortality \citep{
chia2024current}.  

In view that incorporating individual histories of cardiovascular diseases can markedly improve the accuracy of death time predictions, 
\cite{zhang2024ai} recently made a successful attempt by  
developing an AI-based online system for survival risk monitoring. However, their model only incorporated the occurrence indicators of multiple cardiovascular diseases, excluding their onset times and leading to information loss.

In this paper, the endpoints of our interest include the time to death (DTH, terminal event), and the times to various cardiovascular diseases 
(intermediate events). Such a data structure is encountered in many clinical studies particularly when patients are monitored over time for the development of multiple health outcomes \citep{
rueda2019dynamics}. The configurations are akin to the semi-competing risks setting \citep{fine2001semi}, wherein the occurrence of the terminal event censors the intermediate events but not vice versa. Unlike the classical semicompeting risks setting considering only one intermediate event, patients in many clinical studies may exhibit diverse pathogenetic pathways involving multiple intermediate diseases. These complexities pose significant theoretical and computational challenges, limiting the applicability of the standard semi-competing risks models. 

The focus of our work is on survival prediction with multiple dependent intermediate events, which are subject to random censoring or informative censoring due to the occurrence of the terminal event. To this end, we propose a novel two-stage procedure, where we first reveal the underlying complex association structure among the intermediate and terminal events and construct a joint distribution of these event times shown in the right panel of Figure \ref{fig:workflow} for individual-specific predictions in subsequent analyses, and then followed by dynamic prediction of overall survival as illustrated in the lower left panel of Figure \ref{fig:workflow}. 
\begin{figure}[t]
    \centering
    \includegraphics[width=0.85\textwidth,trim=55 50 30 10,clip]{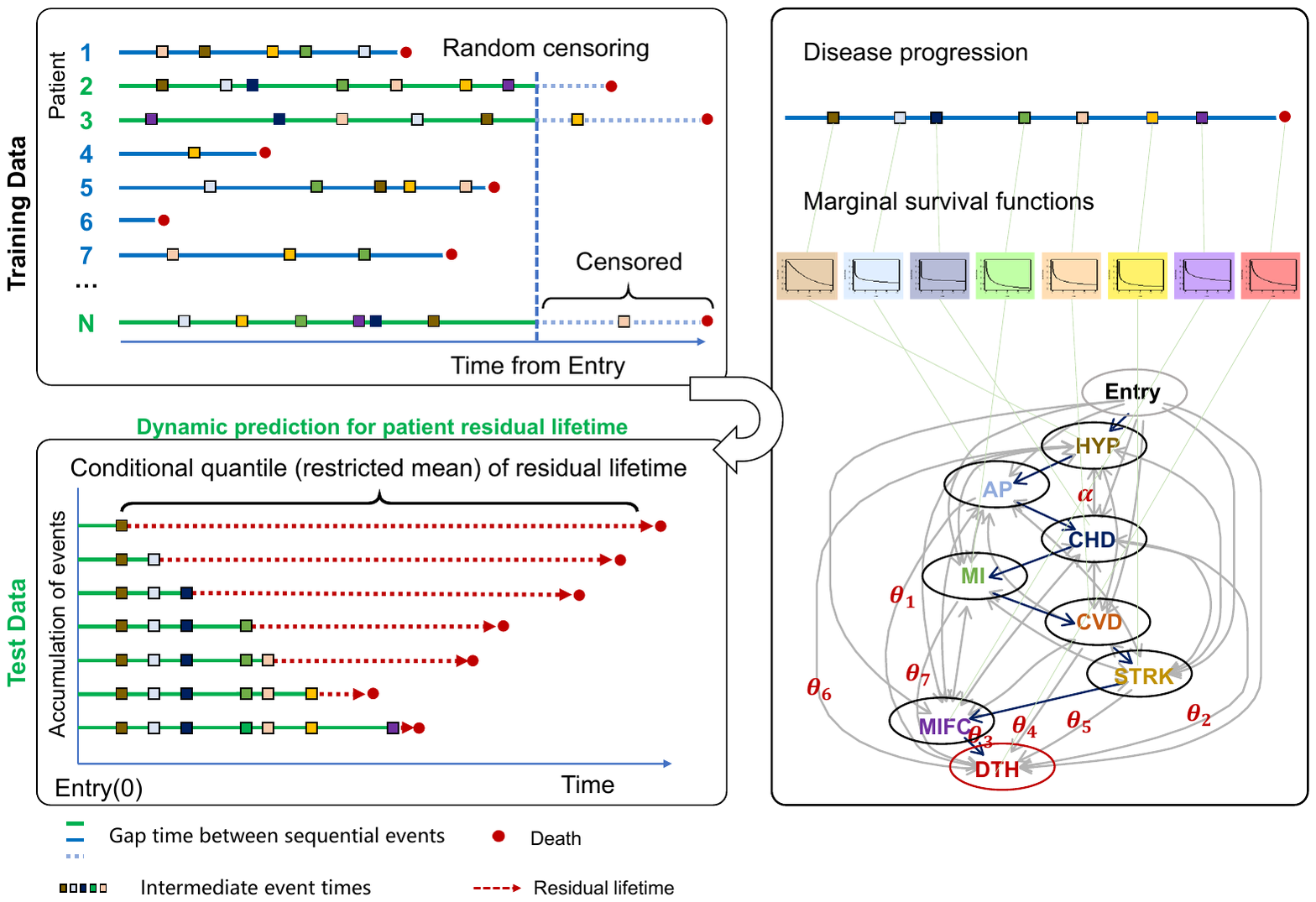}
    \caption{A flowchart of the study purpose. Event colors are aligned across the training data (upper left), test data (lower left), and marginal survival functions (right panel), with labels shown in the directed graph using the same color scheme. The association parameter $\theta_k$ ($k=1,\cdots,7$) quantifies relationship between each cardiovascular disease with death, and $\alpha$ will be defined in Section 2. }
    \label{fig:workflow}
\end{figure}

\subsection{Related works, limitations and challenges}

Although association analysis of bivariate survival data with a single intermediate event time subject to dependent censoring caused by the terminal event have been widely practiced in semicompeting risks settings \citep{fine2001semi}, 
its application to data with multiple intermediate events is rather limited so far. \cite{peng2019joint} proposed a time-dependent association measure for two intermediate events and assumed that they have different associations with the terminal event.  
\cite{li2023evaluating} proposed an iterative algorithm for estimating the association parameters between the two intermediate events, while intermediate events were assumed to share the same association with the terminal event.
It is worth to note that both works overlooked the issue that the association structure of the intermediate event times should be limited within the observable region, where the intermediate event times are always no greater than the terminal event time. In the presence of multiple intermediate events, \cite{li2020multiple} made some attempts based on the stringent assumption that all intermediate and terminal event times share the same association.

The dynamic survival prediction based on patients' past medical history is critical for guiding therapeutic interventions in clinical practice, and has attracted much attention from statisticians.  Recent literature and methodological advances include contributions by 
\cite{
rhodes2023dynamic} and \cite{zhang2024dynamic} 
with longitudinal/time-dependent covariates, as well as by 
\cite{liang2023tackling}, 
who addressed multiple intermediate events as recurrent events with emphasis on the relationship between the terminal event time and gap times between two consecutive events, without accounting for the stochastic nature of these events.
In practice, the correlation with multiple intermediate events often changes over their occurrence times and
individual's genetic or environmental features, 
and thus affects individual's overall survival  
\citep{burzykowski2005evaluation}.  
A dynamic terminal prediction framework integrating the onset information of patients' multiple intermediate events is essential. 
 
Unlike multi-state models focused on the transition from one state to another state \citep{
rueda2019dynamics}, 
providing limited information about the dependency and marginals of both events \citep{peng2008overview}, we aim to integrate the information from medical history into a joint model framework for prediction.
It is noted that inclusion of multiple intermediate events in the study makes learning of association particularly challenging due to complex dependence between intermediate events and the terminal event. The complexity may stem from the following characteristics of the data considered. 
Firstly, the times to intermediate events often exhibit skewness and are subject to both independent and dependent censoring. Secondly, the 
occurrences of intermediate events varies across patients as demonstrated in 
the Framingham heart study, 
with complex transitions between different 
intermediate states and underlying disease progression. 
Thirdly, 
patients may experience some of these intermediate events only rather than 
all of them before death, resulting in incomplete and imbalanced observations of the intermediate event times.
Ignoring the dependencies among intermediate event times and/or their synergistic effects on overall survival could lead to inaccurate overall survival prediction. 
This complexity may lead to biased estimation in multi-state models, especially when some transitions are rare in a limited sample.

To tackle these issues, we exploit an Archimedean copula-based model framework tailored for jointly analysis of multivariate intermediate event times 
subject to dependent censoring due to the terminal event death.   
Despite the wealthy literature of copula models for 
ordered bivariate survival data \citep[e.g.,][]
{peng2007semicompeting,lakhal2008,
yu2025exploring} and multivariate survival data \citep{
prenen2017extending,li2018varying}, which typically assume independent censoring, 
these methods are rarely applicable to settings involving 
multiple intermediate events subject to 
dependent censoring as considered in this study.  
Based on the joint modeling framework, we then develop a nonparametric method for estimating marginal components, along with a pseudo-likelihood estimator for the association parameter among intermediate events. Our proposed method enables to leverage the onset information of multiple intermediate events for dynamic prediction of overall survival subsequently.

The rest of the article is organized as follows.
Section \ref{sec:model} describes the proposed copula-based model framework.  
Section \ref{sec:estimation} develops the estimation method, and 
Section \ref{sec:prediction} details a procedure for dynamic prediction of overall survival.  
Section \ref{sec:simulation} reports simulation results and an analysis of the motivating data. 
The paper concludes with a discussion in Section \ref{sec:conclusion}. 
Due to the page limit, a thorough simulation study in comparison with the other competing algorithms, and rigorous theoretical proofs for the asymptotic properties 
of the proposed estimator, are included in the Supplementary Materials.

\vspace{-10pt}
\section{Model and notations}
\label{sec:model}

Let $\tildeT_1,\cdots,\tildeT_K$ be the times to diagnosis of $K$ cardiovascular diseases (intermediate events), $K\geq 2$, 
$D$ be death time, and $C$ be censoring time. 
It is assumed that $(\tildeT_1,\cdots,\tildeT_K, D)$ is independent of $C$. The sample data consists of $n$ individuals with potential event times and censoring time $\{\tildeT_{i1},\cdots,\tildeT_{iK},D_i,C_i:i=1,\cdots,n\}$, which are $n$ 
realizations of $(\tildeT_1,\cdots,\tildeT_K,D,C)$.
Due to the right censoring,  $\tildeT_k$ and $D$ are observed as $(T_k,Y,\delta_k,\tildedelta)$, $k=1,\cdots,K$. Here, $T_k =\min( \tildeT_k, D,C)$ and $Y=\min(D, C)$ are the observed $k$th intermediate event time and terminal event time, respectively, $\delta_k=I(\tildeT_k\leq \min(D, C))$ represents the indicator variable for the $k$th intermediate event, and $\tildedelta= I(Y\leq C)$ represents the death status. Thus, the observed data takes the form of $\mathcal{O}=\{T_{i1},\cdots,T_{iK}, \delta_{i1},\cdots,\delta_{iK}, Y_i,\tildedelta_i:i=1,\cdots,n\}$. 
\begin{figure}[htbp]
    \centering
    \includegraphics[width=0.3\textwidth,trim=55 20 50 20,clip]{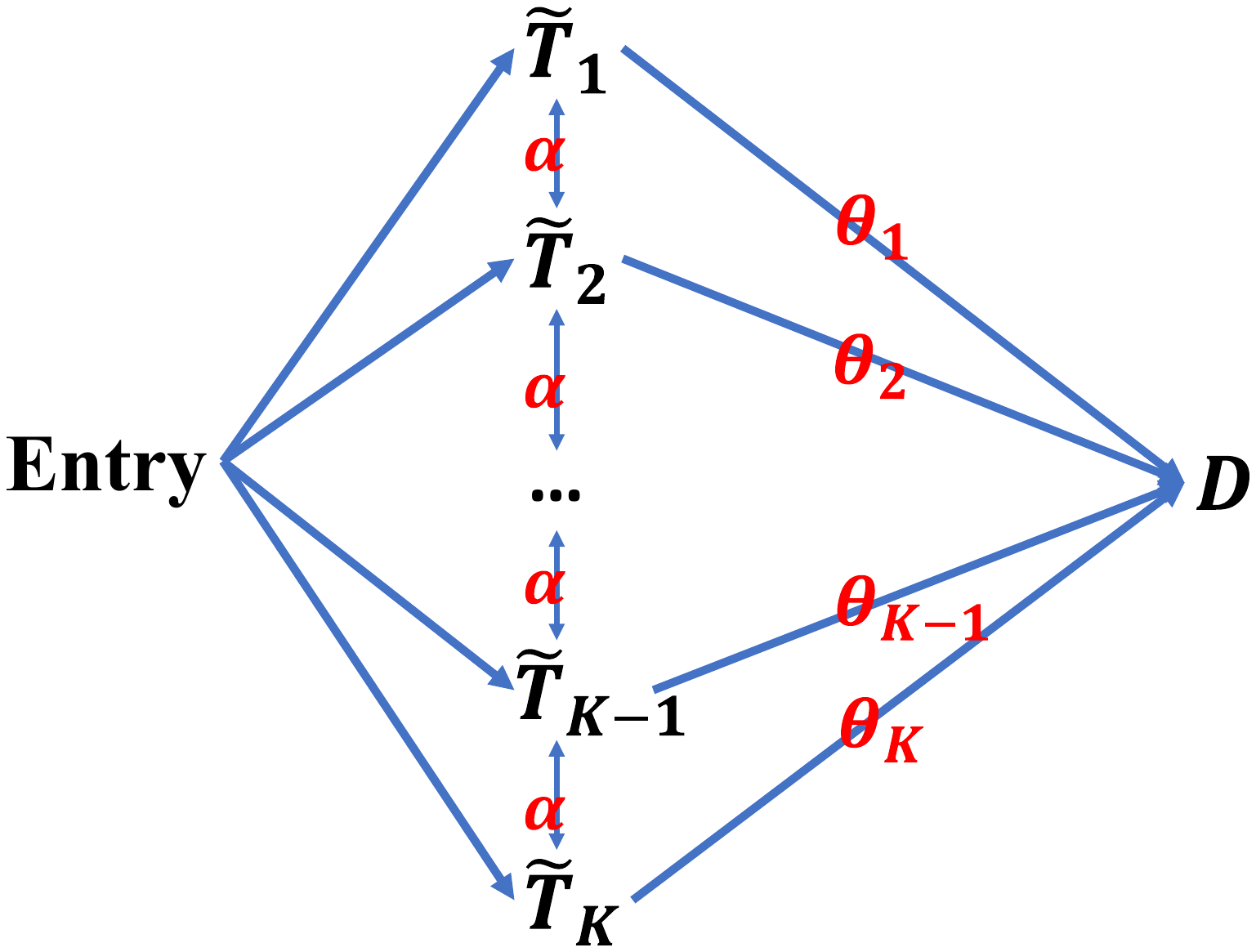}
    \caption{Data structure and association between any two of all event times.  $\tildeT_1,\cdots,\tildeT_K$ represent the potential times to intermediate events from entry, and $D$ represents the terminal event time. The association parameter $\theta_k$ ($k=1,\cdots,K$), defined in model (1), quantifies relationship between each intermediate event with terminal event. $\alpha$ defined in model (3) quantifies the conditional dependence among all intermediate events given death. }
    \label{fig:notation}
\end{figure}

To learn association among $(\tildeT_1,\cdots,\tildeT_K,D)$ from the observed data $\mathcal{O}$ and conduct survival prediction for a new subject, 
we consider a model framework comprising two layers 
of submodels: bivariate models that examine the relationship between each intermediate event time and death time, and a multivariate joint model that integrates all intermediate event times, as illustrated in Figure \ref{fig:notation}.  
In the first layer 
of the model, the dependence structure between bivariate event times is formulated with an Archimedean copula within the observable region. The Archimedean copula is used due to its flexibility in capturing variety of dependence structures and its simple form 
\citep{nelsen2006introduction}. 
Specifically, the joint survival function of the $k$th intermediate event time $\tildeT_k$ and $D$ satisfies 
\begin{equation}\label{model:copula_terminal}
\begin{split}
S_{k,D}(t_k,t) &= \Pr (\tildeT_k> t_k, D> t) = H\{S_k(t_k),S_{D}(t);\theta_k\},0\leq t_k\leq t,
\end{split}
\end{equation} 
and $ H$ is a continuous mapping defined on the unit plane $[0,1]^2$ 
possessing an Archimedean copula such that 
$ H\{S_k(t_k),S_{D}(t);\theta_k\}=\psi_{\theta_k}\left\{\phi_{\theta_k}(S_k(t_k))+\phi_{\theta_k}(S_{D}(t))\right\},$  
where $S_k(t)$ is the marginal survival function of $\tildeT_k$, $S_{D}$ is the marginal survival function of $D$, $\theta_k$ is an unknown association parameter, $k=1,\cdots,K$, and $\psi_{\theta_k} = \phi_{\theta_k}^{-1}$ is the inverse function of $\phi_{\theta_k}$.  In the Archimedean family, the generator function $\phi_{\theta_k}$ maps from $[0,1]$ to $[0,\infty)$, and is continuous and strictly decreasing from $\phi_{\theta_k}(0)>0$ to $\phi_{\theta_k}(1)=0$.

In the second layer, 
we establish the joint survival function of $\tildeT_1,\cdots,\tildeT_K$ and $D$ using  
a $K-$dimensional copula, denoted by $C\{S_{1,D},\cdots,S_{K,D}\}$, 
where $S_{k,D}=S_{k,D}(t_k,t) $ is a bivariate copula of $\tildeT_k$ and $D$ defined in \eqref{model:copula_terminal}. This construction method has been studied 
by \cite{joe1997multivariate} and 
\cite{nelsen2006introduction} in the context of a compatible joint distribution given univariate 
and bivariate marginal distributions. 
On the observable region of the data, the joint survival function of $\tildeT_1,\cdots,\tildeT_K$ and $D$ 
can then be written as
\begin{equation}\label{model:copula_nonterminal}
\begin{split}
&\Pr (\tildeT_1> t_1,\cdots,\tildeT_K> t_K, D>t) 
= -\int_t^{\infty}\Pr (\tildeT_1> t_1,\cdots,\tildeT_K> t_K|D=y) dS_D(y),
\end{split}
\end{equation} 
for $0\leq t_1,\cdots, t_K\leq t$, where the conditional dependence structure among multiple intermediate event times, namely, $\Pr (\tildeT_1> t_1,\cdots,\tildeT_K> t_K|D=y)$, is assumed to be exchangeable 
with a global association parameter $\alpha$. That is, 
\begin{equation}\label{model:copula_nonterminal2}
\begin{split}
&\Pr (\tildeT_1> t_1,\cdots,\tildeT_K> t_K|D=y) \\
&= C\{G_1(t_1;y),\cdots,G_K(t_K;y);\alpha\} 
=\psi_{\alpha}\left\{\sum\limits_{k=1}^K\phi_{\alpha}(G_k(t_k;y))\right\}, t_1,\cdots,t_K\leq y,
\end{split}
\end{equation} 
and $G_k(t_k;y) = \Pr(\tildeT_k> t_k|D=y)$. 
The conditional distribution in \eqref{model:copula_nonterminal2} has been similarly utilized by 
\cite{cook2007statistical} (Pages 220-221) for analyzing 
recurrent event data in the presence of death.  
Note that the association parameter $\alpha$ measures the 
conditional dependence given $D$. Through 
models \eqref{model:copula_terminal}-\eqref{model:copula_nonterminal2}, $\alpha$ and $\theta_k$ ($k=1,\cdots,K)$ jointly regulate 
the dependence between any pair of $K$ intermediate event times such that the unconditional association among $(\tildeT_1,\cdots,\tildeT_K)$ 
may not necessarily adhere 
to an exchangeable structure (Web Figure S.1). This potential complex association 
illustrated by numerical examples in our simulation study in Web Appendix D.

In addition, it is worth noting that the proposed
copula-based joint models 
in \eqref{model:copula_terminal}--\eqref{model:copula_nonterminal2} are defined within the upper wedge with $\tildeT_1,\cdots,\tildeT_K\leq D$ only,  well capturing the 
nature of the partially ordered 
observations in the study. 
The model in the lower wedge is unverifiable as pointed out by \cite{fine2001semi} and \cite{xu2010illnessdeath}.  
Specifically, let $\tildeT_k=\infty$ if the terminal event occurs before $k$th intermediate event. Then,
there is no probability mass in the lower wedge. To balance the probability function,  the joint density of $(\tildeT_k,D)$ is taken to be absolutely continuous defined in the upper wedge, and continuous along the line at $t_k = \infty$. In this case, for $t_k<\infty$, $G_k(t_k;t)=G_k(\min(t_k,t);t)$ and $\Pr (\tildeT_1> t_1,\cdots,\tildeT_K> t_K|D=t)=\Pr (\tildeT_1> \min(t_1,t),\cdots,\tildeT_K> \min(t_K,t)|D=t)$. In other words, among those who die at time $t$, no new intermediate 
events occur possibly later than $t$ \citep{nevo2022causal}. These definitions provide validity for the marginal, conditional and joint distributions of intermediate event times and death time.

\section{Marginal and association analyses}
\label{sec:estimation}
\subsection{Joint likelihood function}
\label{sec: jointlik}

To construct an estimation procedure for association parameters in models \eqref{model:copula_terminal}-\eqref{model:copula_nonterminal2}, we consider a joint likelihood function based on 
observations $\{T_{i1},\cdots,T_{iK}, \delta_{i1},\cdots,\delta_{iK}, Y_i,\tildedelta_i\}$ for $i=1,\cdots,n$. Due to the complexity of the data structure and association between different types of events considered in this paper, 
the joint likelihood is formulated as follows. 
For each individual 
patient $i$ whose death time is exactly observed with $\tildedelta_i=1$, 
if all $K$ intermediate events 
occur before death, the corresponding likelihood 
can be derived from taking the $(K+1)^{th}$ order derivative of \eqref{model:copula_nonterminal}. If some intermediate events of patients $i$ do not occur before death, they are treated 
as events that never happened. 
Without loss of generality, suppose 
that the first $r_i$ intermediate events of individual $i$ never occur, $r_i\in\{1, 2, \cdots, K\}$. 
From the facts $G_k(t_k;y)=G_k(\min(t_k,y);y)$ and $\Pr (\tildeT_1> t_1,\cdots,\tildeT_K> t_K|D=y)=\Pr (\tildeT_1> \min(t_1,y),\cdots,\tildeT_K> \min(t_K,y)|D=y)$, 
model \eqref{model:copula_nonterminal2} can be rewritten as
\begin{equation}\label{model:copula_nonterminal3}
\begin{split}
&\Pr (\tildeT_1> t_1,\cdots,\tildeT_K> t_K|D=y) \\
&=\psi_{\alpha}\left\{\sum\limits_{k:t_k\leq y}\phi_{\alpha}(G_k(t_k;y))+
\sum\limits_{k:t_k> y}\phi_{\alpha}(G_k(y;y))\right\}, t_1,\cdots,t_K,y\geq 0.
\end{split}
\end{equation} 
We first consider the case where $r_i=1$, that is, the 1st intermediate event never happens $\tildeT_1> D$.  
Plugging \eqref{model:copula_nonterminal3} into model \eqref{model:copula_nonterminal}, when $D>t$, $t_1>t$ and $t_2,\cdots,t_K\leq t$,  
we have 
\begin{equation*}
\begin{split}
&\Pr (\tildeT_1> t_1,\cdots,\tildeT_K> t_K, D>t) \\
&= -\int_t^{t_1}\psi_{\alpha}\left\{\sum\limits_{k=2}^K\phi_{\alpha}(G_k(t_k;y))+
\phi_{\alpha}(G_1(y;y))\right\} dS_D(y)
 -\int_{t_1}^{\infty}\psi_{\alpha}\left\{\sum\limits_{k=1}^K\phi_{\alpha}(G_k(t_k;y))\right\} dS_D(y).
\end{split}
\end{equation*} 
In the case that $r_i=2$, when $t_1>t_2>t$ and $t_3,\cdots,t_K\leq t$, similarly we have
\begin{equation*}
\begin{split}
&\Pr (\tildeT_1> t_1,\cdots,\tildeT_K> t_K, D>t) \\
&= -\int_t^{t_2}\psi_{\alpha}\left\{\sum\limits_{k=3}^K\phi_{\alpha}(G_k(t_k;y))+\sum\limits_{k=1}^2
\phi_{\alpha}(G_k(y;y))\right\} dS_D(y)\\
& \quad -\int_{t_2}^{t_1}\psi_{\alpha}\left\{\sum\limits_{k=2}^K\phi_{\alpha}(G_k(t_k;y))+
\phi_{\alpha}(G_1(y;y))\right\} dS_D(y)
 -\int_{t_1}^{\infty}\psi_{\alpha}\left\{\sum\limits_{k=1}^K\phi_{\alpha}(G_k(t_k;y))\right\} dS_D(y),
\end{split}
\end{equation*} 
where the first term represents the scenario that neither of the first two intermediate events occurs, and the second term implies that only the second intermediate event occurs but the first one does not. The contributions of the above two cases to the joint likelihood can be obtained by taking the derivative of  $\Pr (\tildeT_1> t_1,\cdots,\tildeT_K> t_K, D>t)$ with respective to  $t_k$ and $t$.   
Note that for individuals with $\tildedelta=1$, the dependency in the joint survival function does not change when dropping some intermediate event times later than death, and the derivative of $\Pr (\tildeT_1> t_1,\cdots,\tildeT_K> t_K, D>t)$ depends on the number of uncensored intermediate events only.   
Let $d_i=\sum_k \delta_{ik}$ and $\mathbf{G}=\{G_1,\cdots,G_K\}$.  
For individual $i$ with $\tildedelta_i=1$ and 
$r_i$ ($=K-d_i$) intermediate events that do not occur before death, the corresponding contribution 
to the joint likelihood function can be derived using induction as follows: 
\begin{equation}
\begin{split}
L_{i1} (\alpha,\mathbf{G},S_D)  
& = (-1)^{d_i+1} S'_D(Y_i)\psi^{(d_i)}_{\alpha}\Big\{ \sum\limits_{k=1}^K
\phi_{\alpha}(G_k(T_{ik};Y_i))\Big\}\\
&\times\prod_{k=1}^K \left[\phi'_{\alpha}(G_k(T_{ik};Y_i))G'_k(T_{ik};Y_i)\right]^{\delta_{ik}},
\label{eq:likelihood1}
\end{split}
\end{equation}
where $G'_k(t_k;t)$ is the first-order derivative of $G_k(t_k;t)$ with respect to $t_k$, and $\psi^{(d_i)}(\cdot)$ is the $d_i$\emph{th}-order derivative of $\psi(\cdot)$ .

Next, we consider the scenario where 
individual $i$ is still alive with $\tildedelta_i=0$. 
For ease of presentation, we denote 
the joint density function of $(\tildeT_{m_1},\cdots,\tildeT_{m_d},D)$ by
$f_{m_1,\cdots,m_d,K+1}(t_{m_1},\cdots,t_{m_d},t)$, 
which can be derived from models \eqref{model:copula_terminal}-\eqref{model:copula_nonterminal3} accounting for the probability balance discussed at the end of Section 2, and takes a form 
similar to the right-hand side of Eq.\eqref{eq:likelihood1}. 
In the case that 
$\delta_{i1}=0$ and $\delta_{i2} =\cdots=\delta_{iK}=1$, that is, the first intermediate event time $\tildeT_{i1}$ is censored and the rest  $\tildeT_{i2},\cdots,\tildeT_{iK}$ are exactly observed, we have  
\begin{equation}\label{eq:J2}
\begin{split}
&\Pr (\tildeT_1> Y_i,\tildeT_2=T_{i2},\cdots,\tildeT_K= T_{iK}, D> Y_i)\\
&=\int_{Y_i}^{\infty} \int_{Y_i}^t f_{1,2,\cdots,K,K+1}(t_1,T_{i2},\cdots,T_{iK},t)dt_1dt
+ \int_{Y_i}^{\infty}  f_{2,\cdots,K,K+1}(T_{i2},\cdots,T_{iK},t)dt\\
&: = J_{2,\cdots,K,K+1}^1(T_{i1},\cdots,T_{iK},Y_i)+ J_{2,\cdots,K,K+1}^{\emptyset}(T_{i1},\cdots,T_{iK},Y_i).
\end{split}
\end{equation}
When the first two intermediate event times  $\tildeT_{i1}$ and $\tildeT_{i2}$ are censored with 
$\delta_{i1}=\delta_{i2}=0$, and the rest $\tildeT_{ik}$ are exactly observed with $\delta_{ik}=1$ for $k=3,\cdots, K$, 
we have

\begin{equation}\label{eq:J3}
\begin{split}
&\Pr (\tildeT_1> Y_i,\tildeT_2>T_{i2},\tildeT_3=T_{i3},\cdots,\tildeT_K= T_{iK}, D> Y_i)\\
&=\int_{Y_i}^{\infty} \int_{Y_i}^t\int_{Y_i}^t f_{1,2,\cdots,K,K+1}(t_1,t_2,T_{i3},\cdots,T_{iK},t)dt_1dt_2dt\\
&+ \int_{Y_i}^{\infty}\int_{Y_i}^t  f_{2,\cdots,K,K+1}(t_2,T_{i3},\cdots,T_{iK},t)dt_2dt\\
&+ \int_{Y_i}^{\infty}\int_{Y_i}^t  f_{1,3\cdots,K,K+1}(t_1,T_{i3},\cdots,T_{iK},t)dt_1dt
+ \int_{Y_i}^{\infty} f_{3\cdots,K,K+1}(T_{i3},\cdots,T_{iK},t)dt.
\end{split}
\end{equation}
Analogous to \eqref{eq:J2}, the right-hand side of 
\eqref{eq:J3} can be rewritten as 
\begin{equation}
\begin{split}
& 
J^{1,2}_{3,\cdots,K,K+1}(T_{i1},\cdots,T_{iK},Y_i)+J^2_{3,\cdots,K,K+1}(T_{i1},\cdots,T_{iK},Y_i)\\
& +J^1_{3,\cdots,K,K+1}(T_{i1},\cdots,T_{iK},Y_i)+ J_{3,\cdots,K,K+1}^{\emptyset}(T_{i1},\cdots,T_{iK},Y_i).
\end{split}
\end{equation}
In general, let $s$ be any subset of 
$\{k:\delta_{ik}=0\}$, and $|s|$ be the size of the set $s$. By induction, the contribution of individual $i$ whose 
$\tildedelta_i=0$ 
to the joint likelihood function is given by
\begin{equation}
L_{i2} (\alpha,\mathbf{G},S_D) = \sum\limits_{s\subset\{k:\delta_{ik}=0\}}J^{s}_{\{k:\delta_{ik}=1\},K+1}(T_{i1},\cdots,T_{iK},Y_i),
\label{eq:likelihood2}
\end{equation}
where 
\begin{equation}\label{eq:Js}
\begin{split}
&J^{s}_{\{k:\delta_{ik}=1\},K+1}(T_{i1},\cdots,T_{iK},Y_i)\\
&= (-1)^{d_i+1+|s|}\int_{Y_i}^{\infty}\left\{\prod_{k:\delta_{ik}=1} \phi'_{\alpha}(G_k(T_{ik};t))G'_k(T_{ik};t)\right\}
\int_{Y_i}^t \psi^{(d_i+|s|)}_{\alpha}\Bigg\{ 
\sum\limits_{k:\delta_{ik}=1}\phi_{\alpha}(G_k(T_{ik};t))\\
&+ \sum\limits_{k\in s}\phi_{\alpha}(G_k(t_k;t))+\sum\limits_{k\notin s,\delta_{ik}=0}\phi_{\alpha}(G_k(t;t))\Bigg\}
d\prod\limits_{k\in s}\phi_{\alpha}(G_k(t_k;t))dS_D(t).
\end{split}
\end{equation}
Combining the likelihood contributions $L_{i1}$ and $L_{i2}$ in \eqref{eq:likelihood1} and \eqref{eq:likelihood2} over all individuals for $i=1, \cdots, n$,  the resulting joint likelihood function is 
\begin{equation}
L_n(\alpha,\mathbf{G},S_D) = \prod_{i=1}^n L_{i1}(\alpha,\mathbf{G},S_D)^{\tildedelta_i}L_{i2}(\alpha,\mathbf{G},S_D)^{1-\tildedelta_i}.
\label{eq:likelihood}
\end{equation}

It is noted that the likelihood above includes 
three- 
or higher-
dimensional integrals due to Eq.\eqref{eq:Js} when $|s|\geq 2$. This computational challenge can be addressed numerically through the application of a Laplace transformation \citep{joe1997multivariate}. 
Since the generator $\psi_{\alpha}$ is a Laplace transformation of a positive distribution function $F_{\alpha}(x)$ with $F_{\alpha}(0)=0$ 
such that 
$\psi_{\alpha}(t) = \int_0^{\infty}\exp(-tx)dF_{\alpha}(x), t\geq 0,$ the $(d_i+|s|)^{th}$ order derivative of $\psi_{\alpha}$ is in the form of $\psi^{(d_i+|s|)}_{\alpha}(t) = \int_0^{\infty} (-x)^{d_i+|s|}\exp(-tx)dF_{\alpha}(x).$ As such, the integral in \eqref{eq:Js} reduces to a two-dimensional one:
\begin{equation}\label{eq:Js2}
\begin{split}
&J^{s}_{\{k:\delta_{ik}=1\},K+1}(T_{i1},\cdots,T_{iK},Y_i)\\
&=(-1)^{d_i+1+|s|}\int_{Y_i}^{\infty}\prod_{k:\delta_{ik}=1} \phi'_{\alpha}(G_k(T_{ik};t))G'_k(T_{ik};t)\\
&\times\int_0^{\infty}(-x)^{d_i}\Bigg[\exp\left\{-x\sum\limits_{k:\delta_{ik}=1}\phi_{\alpha}(G_k(T_{ik};t))\right\}\exp\left\{-x\sum\limits_{k\notin s,\delta_{ik}=0}\phi_{\alpha}(G_k(t;t))\right\}\\
&\times \prod\limits_{k\in s}\Big\{\exp\left\{-x \phi_{\alpha}(G_k(t;t))\right\}-\exp\left\{-x \phi_{\alpha}(G_k(Y_i;t))\right\}\Big\}\Bigg]dF_{\alpha}(x)dS_D(t),
\end{split}
\end{equation}
where the inner integral can be numerically evaluated 
via  Monte Carlo integration, and the outer one is a Stieltjes integral and can be approximated by the Riemann sum.
The use of the Laplace transformation in  
\eqref{eq:Js2} significantly improves computational efficiency 
especially when the size of the set $s$ is large.    

\subsection{Estimation procedure}
\label{sec: two-stage estimation}

The joint likelihood function $L_n(\alpha,\mathbf{G},S_D)$ 
includes unknown functions $G_k$ and $S_D$. We propose an estimation procedure basically starting from estimation of $G_k$ and $S_D$, and then maximizing the pseudo-likelihood obtained by plugging their estimators   
for estimation of $\alpha$. Herein, the survival function of 
$D$, $S_{D}$, can be estimated by the Kaplan-Meier (KM) method since $D$ is subject to independent censoring. We denote the KM estimator of $S_D$ by $\widehat{S}_D$.  The nonparametric estimation of $G_k$ is not trivial due to the intermediate events being subject to informative censoring. 
Let $H_1(u,v;\theta_k) =\partial H(u,v,\theta_k)/\partial u$,   $H_2(u,v;\theta_k) =\partial H(u,v,\theta_k)/\partial v$ and   $H_{12}(u,v;\theta_k) =\partial^2 H(u,v,\theta_k)/\partial v\partial u$. 
From model \eqref{model:copula_terminal}, $G_k(t_k;t)$ 
can be written as
\begin{equation}\label{eq:tildeS}
G_k(t_k;t) = H_2\{S_k(t_k),S_{D}(t);\theta_k\}, \quad t_k\leq t,
\end{equation} 
and thus its derivative $G'_k(t_k;t)$ is 
$G'_k(t_k;t) = H_{12}\{S_k(t_k),S_{D}(t);\theta_k\}S'_k(t_k).$
Clearly, substituting consistent estimates of $\theta_k$, $S_k$ and $S_{D}$ yields estimates of $G_k$, denoted by $\widehat{G}_k$. Unlike $S_D$, the KM method is invalid for estimating the marginal survival functions $S_k$ due to informative censoring. 
In the following, 
we present the estimation methods for $\theta_k$ and $S_k$ first 
and then 
the maximum pseudo-likelihood estimation for $\alpha$. 

Under an arbitrary Archimedean copula, a unified estimator, denoted by $\widehat{\theta}_k$, can be developed following the framework 
of \cite{lakhal2008}, which results in a consistent and asymptotically normal estimator of $\theta_k$. To estimate the marginal survival function $S_k$ in model \eqref{model:copula_terminal}, we adopt pseudo self-consistent estimator denoted by $\widehat{S}_k(t)$, which is similar to that of \cite{jiang2005pseudo}. Further details are outlined in Web Appendix B. 
Plugging  $\widehat{S}_k$, $\widehat{S}_D$ and  $\widehat{\theta}_k$ into \eqref{eq:tildeS} yields  the estimator for $G_k$, denoted by $\widehat{G}_k$. Subsequently,  the association parameter $\alpha$ is estimated by $\hatalpha=\arg\max\limits_{\alpha} \log L(\alpha,\widehat{\mathbf{G}},\widehat{S}_D)$.  The steps to facilitate the proposed estimation procedure are outlined 
in Algorithm 1 in Web Appendix A. 

This two-stage approach decomposes the estimation problem into lower dimensional components, estimating only one-dimensional parameters in the second stage. 
It substantially reduces computational burden and improves numerical stability compared with full joint likelihood estimation, while retaining consistency and correct coverage.
The robustness of this algorithm has been assessed via simulation studies (Web Appendix D) across a range of copula structures and perturbations in the unobservable region of the potential data. 
%
%
%
Since the asymptotics results of $\widehat{S}_D$, $\widehat{\theta}_k$ and $\widehat{S}_k$ follow from the arguments along the lines of \cite{fleming2013counting}, \cite{lakhal2008} and \cite{jiang2005pseudo}, respectively, we now provide the properties of the estimator $\widehat{\alpha}$ in the following theorems. 
\begin{thm}
Under Conditions 1-6 provided in Web Appendix C, 

(i) $\hatalpha$ is a consistent estimator for the true parameters $\alpha_0$; 

(ii) $\sqrt{n}(\hatalpha-\alpha_0)$ converges to a zero-mean normal distribution.
\end{thm}

The detailed proof of this Theorem 
is provided in Web Appendix C. The asymptotic variance is complicated and lacks an explicit form. As such, 
we suggest the percentile Bootstrap method, which accounts for the uncertainty associated with the first-stage estimates, to construct a confidence interval for $\alpha$ in practice. 
Specifically, we sample with replacement from these $n$ subjects and apply 
Algorithm 1 in Web Appendix A to each of the $B$ bootstrapped data sets to acquire $\{\widehat{\alpha}_b\}_{b=1}^B$. The confidence interval for $\alpha$ can then be constructed using the empirical percentiles of these bootstrap estimates.

\section{Dynamic prediction of overall survival}
\label{sec:prediction}
\subsection{Overall survival rate 
}

In the 
estimation procedure, we have inferred 
the marginal survival functions of intermediate and terminal event times and their association parameters. 
We now construct a prediction for overall survival for a patient who has experienced some intermediate events and is still alive at the last follow-up. Since the exact death time is unknown for this patient and the order of future events is unclear, the prediction is based solely on the observable intermediate events, without considering those yet to occur. 

Particularly, given a new subject with the exactly observed intermediate event times $\tildeT_{l_1} = t_{l_1}, \cdots, \tildeT_{l_m} = t_{l_m}$, where $l_1,\cdots, l_m$ are distinct elements belonging to $\{1,\cdots,K\}$, and assuming 
$t_{l_1}\leq \cdots\leq t_{l_m}$ without loss of generality, 
we predict 
the overall survival rate 
for a patient 
who is still alive at the landmark time $t_m^{*} = \max(t_{l_1},\cdots,t_{{l_m}})$. That is, $S^{*}(t| \tildeT_{l_1} = t_{l_1}, \cdots, \tildeT_{{l_m}} = t_{{l_m}})=\Pr(D>t|\tildeT_{l_1} = t_{l_1}, \cdots, \tildeT_{{l_m}} = t_{{l_m}}, D> t_m^{*} )$. Here we allow the 
landmark time  to vary 
up to the maximum observation of intermediate event times, enabling dynamic prediction. 
When no intermediate event occurs in the observable region, we assign zero to each of $m, l_0, t_0, t_m^{*}$ and let $\tildeT_{0} =0.$  
In this case, $S^{*}(t| \tildeT_{0} = 0)=\Pr(D>t)= S_D(t)$ can be estimated by the KM method. Next, we
consider $m=1,2,\cdots, K$. When $m=1$, there is only one intermediate event occurring before the landmark time, say, event $l_1$, occurring at $t_{l_1}$. 
For $t>t_{l_1}$,
$S^{*}(t| \tildeT_{l_1} = t_{l_1})
= H_1\{S_{l_1}(t_{l_1}),S_{D}(t);\theta_k\}/H_1\{S_{l_1}(t_{l_1}),S_{D}(t_{l_1});\theta_k\},$
which 
can be estimated by plugging the estimators of the estimated marginal survival functions and the association between $l_1$-th intermediate event and terminal event times.
When $m\geq 2$, to obtain the survival probability $S^{*}$ requires the joint distribution of multiple intermediate event times, and thus the association parameter $\alpha$ plays a critical role in the prediction. 
The computation of the survival rate relies on the permutation-symmetric property of the Archimedean copula. 
To ease presentation, we define 
\begin{equation*}
\begin{split}
&Q_m(t_{l_1},\cdots,t_{l_m},t;\alpha,G_{l_1},\cdots,G_{l_m})\\
&= (-1)^m\psi^{(m)}_{\alpha}\left\{\sum\limits_{k=l_1,\cdots,l_m}\phi_{\alpha}(G_k(t_k;t))\right\} \prod\limits_{k=l_1,\cdots,l_m}\phi'_{\alpha}(G_k(t_k;t))G'_k(t_k;t).
\end{split}
\end{equation*}
For $m=2$, 
the conditional survival probability at $t$, $t_U>t>\max(t_{l_1},t_{{l_2}})$,  to be predicted is
\begin{equation}
\begin{split}
&S^{*}(t| \tildeT_{l_1} = t_{l_1}, \tildeT_{{l_2}} = t_{{l_2}})
=\frac{\Pr(\tildeT_{l_1} = t_{l_1}, \tildeT_{{l_2}} = t_{{l_2}},D>t)}{\Pr(\tildeT_{l_1} = t_{l_1}, \tildeT_{{l_2}} = t_{{l_2}}, D> \max(t_{l_1},t_{{l_2}}))}\\
&=\frac{\int_t^{\infty}Q_2(t_{l_1},t_{l_2},y;\alpha,G_{l_1},G_{l_2}) dS_D(y)}{\int_{\max(t_{l_1},t_{{l_2}})}^{\infty}Q_2(t_{l_1},t_{l_2},y;\alpha,G_{l_1},G_{l_2}) dS_D(y)}.
\end{split}
\end{equation}
More generally, for $m>2$ and $t_U>t>t_m^{*} =\max(t_{l_1},\cdots,t_{l_m})$, 
it becomes
\begin{equation}
\begin{split}
&S^{*}(t| \tildeT_{l_1} = t_{l_1}, \cdots, \tildeT_{{l_m}} = t_{{l_m}})
=\frac{\int_t^{\infty}Q_m(t_{l_1},\cdots,t_{l_m},y;\alpha,G_{l_1},\cdots,G_{l_m}) dS_{D}(y)}{\int_{t_m^{*} }^{\infty}Q_m(t_{l_1},\cdots,t_{l_m},y;\alpha,G_{l_1},\cdots,G_{l_m}) dS_{D}(y)}.
\end{split}
\end{equation}
The implementation of 
this survival prediction procedure is summarized in Algorithm 2 in Web Appendix A. We assess the in-sample and out-of-sample performances of this algorithm and other competing approaches via simulations (see details in Web Appendix D).

\subsection{
Residual lifetime}
In addition to the 
conditional survival probabilities, the prediction of the remaining lifetime based on past medical history is often of interest in practice. We now introduce the methods for prediction of restricted mean residual life and quantile residual lifetime.  
Given the occurrence of intermediate events, the conditional restricted mean of residual lifetime \citep{
cortese2017regression} 
has the form of

$E\left[\min(D,t_U^{*})-t_m^{*}|\tildeT_{l_1} = t_{l_1}, \cdots, \tildeT_{{l_m}} = t_{{l_m}}, D> t_m^{*}\right]
=\int _{t_m^{*}}^{t_U^{*}} S^{*}(t| \tildeT_{l_1} = t_{l_1}, \cdots, \tildeT_{{l_m}} = t_{{l_m}})dt,$
where $t_m^{*}\leq t_U^{*}\leq t_U$. 
To compare with the actually observed survival time, we consequently compute a conditional restricted mean survival time (CMST) given the landmark time $t_m^{*}$: 
$$E\left[\min(D,t_U^{*})|\tildeT_{l_1} = t_{l_1}, \cdots, \tildeT_{{l_m}} = t_{{l_m}}, D> t_m^{*}\right]
=t_m^{*}+\int _{t_m^{*}}^{t_U^{*}} S^{*}(t| \tildeT_{l_1} = t_{l_1}, \cdots, \tildeT_{{l_m}} = t_{{l_m}})dt.$$
To evaluate the distribution of residual lifetimes for survival, the quantile residual lifetime is also of interest. Let $\eta_{\tau,t_m^{*}}$ 
be the $\tau$-th conditional quantile of residual lifetimes such that  
$S^{*}(t_m^{*}+ \eta_{\tau,t_m^{*}}| \tildeT_{l_1} = t_{l_1}, \cdots, \tildeT_{{l_m}} = t_{{l_m}})
= 1-\tau,$
where $0<\tau< 1-S^{*}(t_U| \tildeT_{l_1} = t_{l_1}, \cdots, \tildeT_{{l_m}} = t_{{l_m}})$ for 
identifiability. 
Similar to CMST, we compute $t_m^{*}+ \eta_{\tau,t_m^{*}}$ in our numerical studies for ease of comparison with the actual survival times, referring it as the conditional quantile survival time (CQST) given the landmark time $t_m^{*}$.  
The $95\%$ prediction interval for each in-sample/out-of-sample individual can be derived from CQST at 0.025 and 0.975 quantile levels. 

To facilitate these predictions 
using estimators $(\hatalpha, \widehat{\theta}_1,\cdots,\widehat{\theta}_K,\widehat{S}_1,\cdots, \widehat{S}_K,\widehat{S}_D)$ we obtain a plug-in estimator for the joint survival function $S_{T_1,\cdots,T_K,D}(t_1,\cdots,t_K,t)=\Pr (\tildeT_1> t_1,\cdots,\tildeT_K> t_K, D>t) $ in model \eqref{model:copula_nonterminal} defined on the upper wedge. By Theorem 1 and mimicking the proof of Theorem 3 in \cite{zhu2012analysing}, it is straightforward to show that 
this plug-in joint survival estimator, denoted by $\widehat{S}_{T_1,\cdots,T_K,D}(t_1,\cdots,t_K,t)$ and subsequently $\widehat{S}^{*}$ are consistent estimators of $\Pr (\tildeT_1> t_1,\cdots,\tildeT_K> t_K, D>t) $ and $S^{*}$, respectively. 
The asymptotic normality of both estimators can also be justified using Theorem 1(ii), the functional delta method and the reasoning provided in the proof of Theorem 3 in \cite{zhu2012analysing}. Plugging in $\widehat{S}^{*}$ for $S^{*}$, the resulting estimator of the conditional restricted mean or quantile of residual lifetimes 
is also asymptotically normal.

\section{Numerical examples}
\label{sec:simulation}
\subsection{Simulation studies}

We conducted extensive simulation studies to evaluate the performance of the proposed methods in association estimation and dynamic survival prediction. Two simulation settings, Ex1 and Ex2 
described in Supplementary Section D.1, were considered to examine model performance under varying numbers of intermediate events and various dependence structures between intermediate and terminal events, respectively.  
Results reported in Web Appendix 
D.3 demonstrate that the proposed method yields consistent estimation of the association parameters, as reflected by small bias and standard deviation, and reasonable coverage probability. To assess survival prediction performance, we considered overall accuracy metrics: 
mean squared prediction error (MSPE), quantile prediction error (QPE), Brier score (BS), integrated Brier score (IBS), and area under the receiver operating characteristic curve (AUC), as detailed in Web Appendix D.1. 
Empirical coverage probability and median interval width are used to measure the reliability of individual prediction intervals.
Results in 
Web Appendix D.5 show that the proposed dynamic prediction (DP) algorithm provides accurate prediction estimates and reliable individual-level prediction intervals, outperforming competing approaches (listed in 
Web Appendix D.4). Web Appendix D.6 in the Supplementary Materials shows that our algorithm exhibits robustness to model misspecification. Furthermore, additional simulation results in 
 Web Appendix D.7 demonstrate that the estimated joint model on the upper wedge remains robust to variations in the joint density within the unobservable region.

\subsection{Application to the  Framingham heart data}
\label{sec:realdata}

In this section, we apply the proposed methodology to analyze the data obtained in the Framingham heart study. 
The Framingham Heart Study, initiated in 1948 by the National Heart, Lung, and Blood Institute of the National Institutes, is a long-term prospective study of cardiovascular epidemiology in the community of Framingham, Massachusetts. Participants have been examined biennially since the study entry and monitored throughout their lifetimes. Disease progression and survival status have been completely recorded through regular surveillance of area hospitals, participant contact, and death certificates. Clinical data were collected during three examination cycles approximately six years apart. In our study, subjects were excluded if (1) they were not enrolled at the first examination cycle, (2) they had a prior history of diseases, (3) they had missing data or irregular observations. This study included 2833 patients and considered $K=7$ intermediate cardiovascular diseases: AP, CHD, MIFC, CVD, STRK, HYP and MI. 
Among them, 696 patients (25$\%$) had recorded death times ranging from 58 to 8766 days, 
330 patients (12$\%$) had exact observations for the times to diagnosis of AP (ranged $231-8764$ days), 562 (20$\%$) for  CHD (ranged $80-8758$ days), 303 (11$\%$) for MIFC (ranged $169-8758 $ days),  479 (17$\%$) for CVD (ranged $101-8758$ days),  140 (5$\%$) for STRK (ranged $ 101-8696$ days),  1714 (61$\%$) for HYP (ranged $0-8764$ days), and  198 (7$\%$) for MI (ranged $169-8758$ days). 
Observations obtained from the first 2500 patients 
were used to form a training dataset for parameter estimation, while the rest were served for testing.

\begin{table}[htbp]
\scriptsize
\tabcolsep=4pt
\begin{center}
\caption{Estimation results of association parameters $\tau_{\alpha}$ and $\tau_{\theta_k}$ for the Framingham heart data under different copula structures (Frank, Clayton, and Gumbel copula functions), where the subscript $k$ corresponds to one of seven intermediate events (AP,CHD,MIFC,CVD,STRK,HYP,MI). Values in each bracket are 95$\%$ confidence intervals based on 200 bootstrap samples.}
\label{table:realdata alpha}
\begin{tabular}{ccccccccccccccc}
\toprule
\multicolumn{1}{l}{Copula}&\multicolumn{1}{c}{$\tau_{\alpha}$}& \multicolumn{7}{c}{$\tau_{\theta_k}$}\tabularnewline
\cmidrule(lr){3-9}
&&\multicolumn{1}{c}{AP}&\multicolumn{1}{c}{CHD}&\multicolumn{1}{c}{MIFC}&\multicolumn{1}{c}{CVD}&\multicolumn{1}{c}{STRK}&\multicolumn{1}{c}{HYP}&\multicolumn{1}{c}{MI}\tabularnewline
\hline
Frank   &0.371  &0.300& 0.482& 0.694& 0.608& 0.680& 0.113& 0.584\tabularnewline
& [0.314,0.426]  & [0.238,0.36] & [0.434,0.522] & [0.641,0.746] & [0.568,0.652] & [0.615,0.75] & [0.07,0.15] & [0.514,0.651]\tabularnewline
Clayton &0.443  &0.449& 0.631& 0.806& 0.740& 0.798& 0.166& 0.728\tabularnewline
& [0.364,0.495] &[0.371,0.516] &[0.585,0.668] &[0.768,0.842] &[0.708,0.773] &[0.748,0.844] &[0.105,0.218] &[0.666,0.778] \tabularnewline
Gumbel  &0.222  &0.224& 0.401& 0.584& 0.512& 0.560& 0.100& 0.459\tabularnewline
& [0.18,0.279] &[0.172,0.273]& [0.36,0.439]& [0.534,0.638]& [0.476,0.558]& [0.488,0.634]& [0.06,0.137]& [0.393,0.526]\tabularnewline
\bottomrule
\end{tabular}\end{center}
\end{table}

Table \ref{table:realdata alpha} summarizes the estimates of 
the association parameters $\theta_k$ and $\alpha$ under Frank, Clayton, and Gumbel copula structures, 
and  
95$\%$ confidence intervals (CI) using 50 Bootstrap replicates. The results indicate that all intermediate events have moderately to highly positive associations with death. The ranking of these associations from highest to lowest is 
MIFC $>$ STRK $>$ CVD $>$ MI $>$ CHD $>$ AP $>$ HYP and concordant across different copula types,  suggesting the inappropriateness of assuming the same association between the intermediate events and death. 
We identified fatal coronary heart disease, fatal cerebrovascular disease, and stroke as strongly significant contributors to mortality, which aligns 
with clinical findings 
\citep{wang2006multiple}. 
We further assessed the impacts of each intermediate disease; for instance, the estimated concordance parameter of Clayton’s copula was $8.31$, corresponding to Kendall’s tau of $0.806$ between the onset times of MIFC  
disease and death. This can be interpreted as a predictive hazard ratio: the hazard of death for those who had experienced a MIFC 
disease is $8.31$ times bigger than the hazard of death for those who had not experienced this disease. 
Additionally, the conditional global associations (Kendall's tau corresponding to $\alpha$) among these seven intermediate events were found using our methods to be $0.37$ (95$\%$ CI $[0.310,0.428]$), $0.450$ (CI $[0.361,0.497]$), 0.217 (CI $[0.162,0.294]$) under the Frank, Clayton, and Gumbel copula structures, respectively, highlighting significant associations among these events. Coupled with estimates of $\theta_k$, the estimate of $\alpha$ depicts the underlying dependence structure of these events. 

As indicated in the simulation study (see Web Appendix D), the performance of the 
dynamic prediction (DP) algorithm is slightly sensitive to model misspecification of the copula structure. Model selection can be carried out practically by maximizing the joint likelihood function \eqref{eq:likelihood} or minimizing Akaike information criterion (AIC). 
In this real data example, the AIC values  in the Frank, Clayton, and Gumbel models are 
110053, 110312 and 113114,
respectively, indicating the Frank model is preferable, followed by the Clayton model.

To illustrate the capabilities of our proposal in dynamics and personalized prediction, 
referring to the workflow in Figure \ref{fig:workflow}, we present in  Table \ref{table:realdata} (upper panel) the predicted outcomes for several patients randomly selected from the test dataset.
The individualized prediction intervals can be constructed by $[\operatorname{CQST}_{i,0.025},\operatorname{CQST}_{i,0.975}]$ in Table \ref{table:realdata}. Aggregate summaries presented in 
Web Figure S.3  demonstrate the concordance between the actual death time and the predicted values.
In addition, 
we generated 14 artificial subjects whose intermediate events occurred in sequence.  
Predicted overall survival reported in the lower panel (for synthetic data) of Table  \ref{table:realdata} indicate clear dynamic changes in patients'  predicted survival rates, CMST and CQST with different intermediate disease progressions.

We also evaluate the predictive accuracy and robustness against the different copula types 
using the out-of-sample Brier score curves in Web Figure S.4 in the supplementary material.
The proposed DP method always offered 
a lower Brier prediction error in the identifiable region $[0,t_U]$ with $t_U=24.02 $ years (8766 days), 
showing the advantages of the proposed dynamic prediction algorithm over the other competitors. 
Furthermore, the DP methods with the Frank, Clayton, and Gumbel copula structures yield the integrated Brier scores of $0.0863$,  $0.0866$, and $0.0908$, respectively, for the test dataset, suggesting the Frank and Clayton models are comparable in terms of the survival prediction accuracy in this example. Calibration analysis provided in Web Appendix 
E further demonstrates that the proposed model is well-calibrated and yields reliable predicted survival probabilities.

\begin{table}[t]
\tiny
\begin{center}
\caption{Personalized and dynamic prediction for the Framingham heart data under Frank copula assumption. $t_U^{*}=24.02$ years (8766 days) was chosen for computing conditional restricted mean survival time (CMST) and conditional quantile survival time (CQST). The landmark time $t_{m}^{*}$ is random and corresponds to the maximal intermediate event time of each patient. The individual prediction intervals can be derived from CQST at 0.025 and 0.975 quantile levels. }
\label{table:realdata}
\begin{tabular}{ccccccccccccccc}
\toprule
&\multicolumn{7}{c}{intermediate event times (days)}& \multicolumn{1}{c}{OS }&\multicolumn{5}{c}{prediction of overall survival}\tabularnewline
\cmidrule(lr){2-8}\cmidrule(lr){10-15}
\multicolumn{1}{c}{ID}&\multicolumn{1}{c}{AP}&\multicolumn{1}{c}{CHD}&\multicolumn{1}{c}{MIFC}&\multicolumn{1}{c}{CVD}&\multicolumn{1}{c}{STRK}&\multicolumn{1}{c}{HYP}&\multicolumn{1}{c}{MI}&\multicolumn{1}{c}{DTH}&\multicolumn{1}{c}{CMST}&\multicolumn{3}{c}{CQST (years)}&\multicolumn{2}{c}{$S(t+t_{m}^{*}| \tildeT_{l_1} = t_{l_1}, \cdots, \tildeT_{l_m}=t_{l_m})$}\tabularnewline
\cmidrule(lr){11-13}\cmidrule(lr){14-15}
&&&&&&&&(years)&(years)&(0.025)&(0.5)&(0.975)&$t=1$ year& $t=2$ years\tabularnewline
\hline
\multicolumn{10}{l}{Patients in the Framingham heart data }\tabularnewline
457803&6945&6945&7260&6945&-&6595&7260&20.84&21.6&19.95&21.3&23.84&0.64&0.38\tabularnewline
843035&1770&1770&-&1770&3975&2767&-&13.73&15.08&11.1&14.44&22.09&0.87&0.72\tabularnewline
905672&3948&3948&6307&3948&-&2948&-&21.64&20.27&17.52&20.02&23.75&0.83&0.64\tabularnewline
909239&5216&3241&-&3241&-&3663&-&17.97&18.21&14.44&17.7&23.6&0.87&0.71\tabularnewline
954821&-&7293&7293&6598&6598&4206&7293&21.62&21.63&20.02&21.37&23.84&0.59&0.35\tabularnewline
1401968&6895&6289&6289&6289&-&1470&6289&21.36&21.32&19.02&21.05&23.9&0.79&0.57\tabularnewline
1641185&-&7267&-&3863&3863&3598&-&20.74&21.65&19.95&21.4&23.9&0.65&0.41\tabularnewline
1757873&3941&3941&5118&-&-&4326&5118&19.27&17.5&14.1&16.93&23.08&0.84&0.67\tabularnewline
2511641&5247&2231&2231&2231&-&3634&2231&16.85&16.97&14.44&16.46&22.28&0.75&0.53\tabularnewline
2754801&1882&1882&6753&-&-&3626&6753&21.05&20.91&18.64&20.75&23.81&0.77&0.55\tabularnewline
2964471&3973&3973&7916&3973&-&5380&-&23.6&22.86&21.75&22.8&23.95&0.55&0.14\tabularnewline
3053543&-&2493&2493&2493&6130&5614&2493&16.86&18.55&16.83&18.13&22.58&0.6&0.36\tabularnewline
4812868&-&2243&-&2243&3058&1511&-&10.52&13.84&8.58&13.36&21.96&0.9&0.78\tabularnewline
5650728&5268&5268&6622&5268&-&3643&6622&20.85&20.61&18.27&20.38&23.76&0.79&0.55\tabularnewline
6317475&-&8256&8256&8256&-&2890&-&23.84&23.35&22.66&23.34&23.96&0.35&0\tabularnewline
6318099&497&406&406&406&-&2919&406&16.07&11.73&8.16&10.81&19.11&0.81&0.66\tabularnewline
7048683&1688&1596&1596&1596&-&1486&1596&12.36&10.07&4.81&9.26&19.22&0.86&0.76\tabularnewline
7135704&-&5602&5602&5602&-&826&-&21.28&19.6&15.6&19.49&23.76&0.9&0.76\tabularnewline
7138958&4594&3839&-&3839&-&2044&-&19.1&17.91&12.95&17.55&23.64&0.92&0.81\tabularnewline
7222607&-&6232&6232&6232&-&651&-&20.56&20.86&17.41&20.86&23.9&0.91&0.79\tabularnewline
7943235&3935&3935&6452&3935&-&4186&-&19.02&20.32&17.86&20.02&23.75&0.79&0.58\tabularnewline
8265511&5774&4397&4397&4397&7484&3632&4397&21.79&21.83&20.59&21.55&23.84&0.51&0.28\tabularnewline
8659129&3494&3494&3584&3494&-&2205&3584&19.17&14.57&9.93&13.98&21.96&0.83&0.74\tabularnewline
8755719&4109&4109&4201&2752&2752&-&4201&17.89&16.16&11.66&15.49&22.73&0.86&0.75\tabularnewline
9030568&1096&1096&5094&1096&-&2150&-&17.23&17.92&14.08&17.38&23.34&0.87&0.73\tabularnewline
9661647&1117&1117&6233&1117&-&2952&6233&17.15&19.99&17.22&19.75&23.67&0.83&0.64\tabularnewline
9674054&4368&4276&4276&4276&-&3647&4276&15.83&15.97&12.24&15.4&22.43&0.84&0.69\tabularnewline
9967157&3273&3273&6662&3273&-&-&6662&20.17&21.04&18.42&20.94&23.9&0.85&0.66\tabularnewline
\hline
\multicolumn{10}{l}{Synthetic data}\tabularnewline
A1&500&-&-&-&-&-&-&-&18.97&3.95&22.98&24.02&0.99&0.99\tabularnewline
A2&500&500&-&-&-&-&-&-&12.7&2.27&12.83&23.21&0.97&0.96\tabularnewline
A3&500&500&900&-&-&-&-&-&10.43&3.22&9.78&20.91&0.98&0.91\tabularnewline
A4&500&500&900&1000&-&-&-&-&9.53&3.75&8.52&19.59&0.98&0.87\tabularnewline
A5&500&500&900&1000&1100&-&-&-&8.34&3.66&7.42&16.83&0.89&0.78\tabularnewline
A6&500&500&900&1000&1100&1200&-&-&8.39&3.75&7.33&16.95&0.9&0.75\tabularnewline
A7&500&500&900&1000&1100&1200&1600&-&8.59&4.5&8.11&16.42&0.8&0.67\tabularnewline
A8&-&-&200&-&-&-&-&-&11.95&1.06&11.24&24.02&0.97&0.94\tabularnewline
A9&-&-&200&-&300&-&-&-&8.5&1.06&7.42&19.43&0.95&0.89\tabularnewline
A10&-&-&200&400&300&-&-&-&7.65&1.63&6.78&17.18&0.94&0.89\tabularnewline
A11&-&-&200&400&300&-&500&-&6.86&1.63&5.91&15.49&0.91&0.86\tabularnewline
A12&-&600&200&400&300&-&500&-&6.89&1.96&5.87&14.75&0.91&0.87\tabularnewline
A13&700&600&200&400&300&-&500&-&6.6&1.96&5.87&14.69&0.89&0.79\tabularnewline
A14&700&600&200&400&300&800&500&-&6.89&2.27&5.89&14.75&0.93&0.79\tabularnewline
\bottomrule
\end{tabular}\end{center}
\end{table}

\section{Discussion}
\label{sec:conclusion}

We propose a new modeling approach to learn the association between multiple dependent intermediate events and the terminal event. 
A copula-based framework is particularly introduced for jointly modeling these intermediate and terminal events. The proposed model and association learning algorithm capture the stochastic nature of multiple intermediate events, their partial-ordering in relation to the terminal event, as well as informative censoring due to the occurrence of the terminal event.

The application to a heart disease dataset reveals that 
the premature or delayed onset of intermediate diseases MIFC/ STRK/ CVD/ MI/ CHD/ AP 
are significantly associated with the mortality risks. 
On the other hand, different diseases have different impacts on death, and the onset of hypertension  has   
the least impact on mortality. 
This finding indicates that aggregated recurrent-event models fail to capture such asymmetric dependencies between intermediate and terminal outcomes. Furthermore, it 
emphasizes the importance of leveraging onset information from various intermediate diseases to improve the accuracy of overall survival predictions.  
Our DP approach enhances the accuracy of predicting the terminal event on a personalized level by integrating information on intermediate clinical events.

We model each intermediate event based on its first occurrence, though extensions to different types of recurrent events are possible. In practice, selecting intermediate events shall balance clinical relevance, statistical identifiability, and predictive utility. Events should be chosen based on clinical or biological factors, with enough cases and follow-up time for reliable analysis. Similar events can be grouped together if they share hazard patterns and outcomes, but distinct ones with unique characteristics should be modeled separately.

It is noteworthy that our proposed prediction algorithm differs significantly from the available  
landmarking methods \citep{van2007dynamic,parast2012landmark}, 
which postulated working time-specific models up to the last observed values directly or relied on some partition rules, lacking information regarding the dependency and marginal distributions of event times.   
Our DP algorithm is constructed upon conditional probabilities for death time given some subject-specific 
observable short-term events, which are derived from their joint distribution, 
adapting to dynamics caused by evolving information from multiple intermediate events. 

We acknowledge the proposed method may become increasingly computationally intensive as the sample size grows, the number of events increases, or the censoring rate becomes heavier. 
Computational efficiency could be improved in future work through parallel computing, algorithmic approximations, or optimized and compiled code implementations.

It would be valuable to extend the proposed method by 
incorporating some predictive biomarkers alongside multiple intermediate event times to reduce estimation bias caused by patient heterogeneity and further improve prediction accuracy. 
Methods based upon pseudo-observations of quantities
of interest \citep{andersen2003generalised} 
would be adopted for this purpose, with further details warranting future investigation.

\section*{Acknowledgements}
The authors thank the Co-editor, the associate editor and two anonymous reviewers for their constructive feedback and suggestion that significantly improved the quality of this paper. 
This work was supported by the HIT Research Start-up Fund ( Grant No. AUGA5710010426) , the Singapore Ministry of Education Academic Research Fund Tier 1 Grant( RG105/24) and Tier 2 Grant( MOE-T2EP20121-0004) .

\section*{Data Availability}
The Framingham heart study data used in this paper is accessible through the R package riskCommunicator 
at \url{https://search.r-project.org/CRAN/refmans/riskCommunicator/html/framingham.html}.

\bibliography{ref}
\bibliographystyle{apalike}

\section*{Supporting Information}
The supplementary material contains our algorithms for implementation, estimation for $\theta_k$ \blue{and $S_k$}, technical proof of Theorem 1 and regularity conditions referenced in Section 3.2, 
additional simulation studies referenced in Section 5.1, 
as well as additional numerical results referenced in Section 5.2. 

\clearpage

\appendix
\renewcommand{\theequation}{A.\arabic{equation}} 
\renewcommand{\thetable}{S.\arabic{table}}
\renewcommand{\thefigure}{S.\arabic{figure}}
\setcounter{table}{0}
\setcounter{figure}{0}
\setcounter{equation}{0}
\begin{center}
\Large{\textbf{Supplementary Material for 
``Learning association from multiple intermediate events for dynamic prediction of survival: an application to  cardiovascular disease prognosis"}}
\end{center}

\section{The proposed algorithms}
\quad

{
\linespread{1}
\begin{algorithm}[h]
\DontPrintSemicolon
\footnotesize
\caption{Marginal and association analyses}
\label{alg:asso}
\KwInput{$\{T_{i1},\cdots,T_{iK}, \delta_{i1},\cdots,\delta_{iK}, Y_i,\tildedelta_i:i=1,\cdots,n\}$}
\KwOutput{$\widehat{S}_D$, $\widehat{\theta}_k, \widehat{S}_k$ ($k=1,\cdots,K$) and $\widehat{\alpha}$}
Compute Kaplan-Meier estimate $\widehat{S}_D$ based on data $\{Y_i,\tildedelta_i:i=1,\cdots,n\}$.\\
\For{$k =1$  \KwTo $K$}{
     Solve $U_k(\theta_k)=0$, get $\widehat{\theta}_k$. 
     \tcp*{see Appendix B}
\tcc{In $U_k(\cdot)$, the copula family shall be pre-specified.}
Solve \eqref{eq:JFKC1} 
with plugging $\widehat{\theta}_k$, get $\widehat{S}_k$.\\
Compute $\widehat{G}_k(t_k;t) \gets H_2\{\widehat{S}_k(t_k),\widehat{S}_{D}(t);\widehat{\theta}_k\}$.
     }
    Compute $\hatalpha\leftarrow\arg\max\limits_{\alpha} \log L(\alpha,\widehat{\mathbf{G}},\widehat{S}_D)$.  \tcp*{see Section 3 in our main manuscript} 
\end{algorithm}
}

{
\linespread{1}
\begin{algorithm}[!ht]\label{alg:pred surv}
\DontPrintSemicolon
\footnotesize
\caption{Survival prediction algorithm}
\KwInput{$\{T_{i1},\cdots,T_{iK}, \delta_{i1},\cdots,\delta_{iK}, Y_i,\tildedelta_i:i=1,\cdots,n\}$}
\KwData{A test patient with exact intermediate event times $(t_{l_1},\cdots,t_{l_m})$}
\KwOutput{$\widehat{S}^{*}(t| \tildeT_{l_1} = t_{l_1}, \cdots, \tildeT_{{l_m}} = t_{{l_m}})$ for $t_m^{*}<t<t_U$}
Run Algorithm \ref{alg:asso} on $\{T_{i1},\cdots,T_{iK}, \delta_{i1},\cdots,\delta_{iK}, Y_i,\tildedelta_i:i=1,\cdots,n\}$ to get $\widehat{S}_D$, $\widehat{\theta}_k, \widehat{S}_k$ ($k=1,\cdots,K$) and $\widehat{\alpha}$.\\
\eIf{$m=0$}{
    $\widehat{S}^{*} \gets \widehat{S}_D(t)$;
}{
    \eIf{$m=1$}
    {
        $\widehat{S}^{*} \gets H_1\{\widehat{S}_{l_1}(t_{l_1}),\widehat{S}_{D}(t);\widehat{\theta}_k\}/H_1\{\widehat{S}_{l_1}(t_{l_1}),\widehat{S}_{D}(t_{l_1});\widehat{\theta}_k\}$;
    }{
    \For{$y>t_m^{*}$}{\For{$k=l_1,\cdots,l_m$}{$\widehat{G}_k(t_k;y) \gets H_2\{\widehat{S}_k(t_k),\widehat{S}_{D}(y);\widehat{\theta}_k\}$;
    \\$\widehat{G}'_k(t_k;y) \gets H_{12}\{\widehat{S}_k(t_k),\widehat{S}_{D}(y);\widehat{\theta}_k\}\widehat{S}'_k(t_k)$;\\
    \tcc{$\widehat{S}'_k$ is the numerical differentiation of $\widehat{S}_k$.}
}}
     \For{$y>t_m^{*}$}{
$ Q_m(y)\gets (-1)^m\psi^{(m)}_{\widehat{\alpha}}\left\{\sum\limits_{k=l_1,\cdots,l_m}\phi_{\widehat{\alpha}}(\widehat{G}_k(t_k;y))\right\} \prod\limits_{k=l_1,\cdots,l_m}\phi'_{\alpha}(\widehat{G}_k(t_k;y))\widehat{G}'_k(t_k;y)$;
    }
     $\widehat{S}^{*}\gets [\int_t^{\infty}Q_m(y) d\widehat{S}_{D}(y)]/[\int_{t_m^{*} }^{\infty}Q_m(y) d\widehat{S}_{D}(y)]$.
    }

}
Output $\widehat{S}^{*}$.
\end{algorithm}
}

\clearpage
\section{Estimators of $\theta_{k0}$ and $S_{k0}$}
If a Clayton copula structure is assumed for the joint model (1) 
between $\tilde T_k$ and $D$, the association parameter $\theta_k$ can be estimated using the concordance estimator proposed by \cite{fine2001semi}. 
For a more general copula, a unified estimator, denoted by $\widehat{\theta}_k$, can be developed following the framework 
of \cite{lakhal2008}, 
providing a consistent and asymptotically normal estimator of $\theta_k$. Specifically, let $\widetilde{X}_{ij}^{(k)} = \min(\tildeT_i^{(k)},\tildeT_j^{(k)})$, $\widetilde{Y}_{ij} = \min (D_i,D_j)$, $\widetilde{C}_{ij} = \min(C_i,C_j)$, $X_{ij}^{(k)} = \min(T_i^{(k)},T_j^{(k)})$, and $Y_{ij} = \min (Y_i, Y_j)$.  
$I\{ (\tildeT_i^{(k)}-\tildeT_j^{(k)})(D_i-D_j)>0\}$ is computable only if $\widetilde{X}_{ij}^{(k)}\leq \widetilde{Y}_{ij}\leq \widetilde{C}_{ij}$. 
Besides, $X_{ij}^{(k)}=\widetilde{X}_{ij}^{(k)}$ and $Y_{ij} = \widetilde{Y}_{ij}$ when $\widetilde{X}_{ij}^{(k)}\leq \widetilde{Y}_{ij}\leq \widetilde{C}_{ij}$ is true. Accordingly, the estimating equation for $\theta_k$ can be constructed as
\begin{equation}\label{eq:theta_nonparametric_estimating}
\begin{split}
U_k(\theta_k)&=\left(\begin{array}{c}n  \\ 2 \end{array}\right)^{-1}
\sum\limits_{i<j} w(X_{ij}^{(k)}, Y_{ij}) I(\widetilde{X}_{ij}^{(k)}\leq \widetilde{Y}_{ij}\leq \widetilde{C}_{ij}) \\
&\times \left[
I\{(T_i^{(k)}-T_j^{(k)})(Y_i-Y_j)>0 \} - \frac{\gamma_{\theta_k}\{s(X_{ij}^{(k)}, Y_{ij}) \}}{\gamma_{\theta_k}\{s(X_{ij}^{(k)}, Y_{ij}) \}+1}\right] = 0,
\end{split}
\end{equation}
where 
\begin{equation*}
s(x,y) = \frac{\sum\limits_{i=1}^n I\{T_i^{(k)}>x,Y_i>y\}}{n\widehat{S}_C(y)}, 0\leq x\leq y,
\end{equation*}
$\widehat{S}_C$ is the Kaplan-Meier (KM) estimator of the survival function of censoring variable $C$ based on observable data $\{(Y_i, 1-\delta_i^{(2)}):i=1,\cdots,n\}$, $\gamma_{\theta}\{s(t_1,t_2)\} = -s(t_1,t_2)\frac{\phi^{''}_{\theta}(s(t_1,t_2))}{\phi^{'}_{\theta}(S(t_1,t_2))}$, and $ \phi_{\theta}$ is the Archimedean generator function. 
As suggested by \cite{fine2001semi}, the weight function $w$ in \eqref{eq:theta_nonparametric_estimating} can be $w=1$ or
\begin{equation}\label{eq:theta_nonparametric_estimating_weight}
 w^{-1}(x, y;a,b) = n^{-1}\sum\limits_{i=1}^n I\{ T_i^{(k)}\geq \min(a,x), Y_i\geq \min(b,y)\}
\end{equation}
with positive constants $a$ and $b$ chosen to dampen $w(x, y) $ for large $x$ and $y$. The solution to \eqref{eq:theta_nonparametric_estimating}, denoted by $\hat\theta$ was shown by \cite{lakhal2008} to have consistency and asymptotical normality.

To estimate the marginal survival function $S_k$ in model (1) 
under an arbitrary Archimedean copula, we adopt a self-consistency estimation approach similar to that of \cite{jiang2005pseudo} developed for semi-competing risks. 
An alternative closed-form estimator for the marginal survival function was proposed by \cite{lakhal2008} based on the copula-graphic method. 
As demonstrated in their study, both estimators have comparable performance and outperform the simple plug-in estimator of  \cite{fine2001semi}.
Jiang et al's method builds on 
\cite{efron1967two}'s self-consistency equation used in deriving the Kaplan-Meier (KM) estimator, and possesses desirable 
asymptotic properties. 
In particular, the 
estimators of $S_k$ are obtained by solving the following self-consistency estimating equations: 
\begin{equation}\label{eq:JFKC1}
\begin{split}
S_k(t) &= \frac{1}{n}\sum\limits_{i=1}^n I(T_{ik}> t) +\frac{1}{n}\sum\limits_{T_{ik}\leq t} (1-\delta_i^{(1)})(1-\delta_i^{(2)})\frac{H(S_k(t),S_D(Y_i);\theta_k)}{H(S_k(T_{ik}),S_D(Y_i);\theta_k)}\\
&+\frac{1}{n}\sum\limits_{T_{ik}\leq t} (1-\delta_i^{(1)})\delta_i^{(2)}\frac{H_2(S_k(t),S_D(Y_i);\theta_k)}{H_2(S_k(T_{ik}),S_D(Y_i);\theta_k)}
\end{split}   
\end{equation}
for $k=1,\cdots,K$.
Specifically, we substitute $\widehat{S}_D$, $\widehat{\theta}_k$ and the initial guesses of $S_k(t)$ into the right-hand side of \eqref{eq:JFKC1} to produce updated estimates, and then 
iterate the process until convergence is achieved.  
We denote the resulting pseudo self-consistent estimators by $\widehat{S}_k(t) $.

\section{Technical proofs}
The proofs of theorems in the main text are based on the empirical processes, Glivenko-Cantelli theorem \citep{vandervaart1996}, the functional delta method and calculus of differentiable statistical functions in \cite{mises1947asymptotic}. 
Before proving asymptotic properties of $\widehat{\alpha}$, we first introduce some notations. Let $\calP_n$ and $\calP_0$ be the empirical measure and the true underlying measure, respectively. Write
\begin{equation*}
\begin{split}
&l_n(\alpha,\mathbf{G}, S_D) = n^{-1}\log L_n(\alpha,\mathbf{G}, S_D)\\
& = \calP_n\left[\widetilde{\delta}\log \psi_{\alpha}^{(d)}\left\{ \sum\limits_{k=1}^K\phi_{\alpha}(G_k(T_k; Y))\right\}\right]
 + \calP_n \left[ \sum\limits_{k=1}^K \widetilde{\delta}\delta_k\log \phi'_{\alpha}(G_k(T_k; Y))\right]\\
&+ \calP_n \Bigg[ (1-\widetilde{\delta})\log\Bigg\{\sum\limits_{s\in\{k:\delta_k=0\}} (-1)^{1+|s|}\int_Y^{\infty} \prod\limits_{k=1}^K \left[\phi'_{\alpha}(G_k(T_k; t))H_{12}\{S_k(T_k),S_D(t);\theta_k\}\right]^{\delta_k}dS_D(t) \\
&\times \int_0^{\infty}(-x)^{d}\bigg\{\exp\bigg[-x\sum\limits_{k=1}^K\left(\phi_{\alpha}(G_k(T_k;t)) \right)^{\delta_k}\left(\phi_{\alpha}(G_k(t;t)) \right)^{I\{k\notin s,\delta_k=0\}}\bigg]\\
&\times \prod\limits_{k\in s}\Big(\exp\left[-x \phi_{\alpha}(G_k(t;t))\right]-\exp\left[-x \phi_{\alpha}(G_k(Y_i;t))\right]\Big)\bigg\}dF_{\alpha}(x)\Bigg\}\Bigg]\\
 &: = \calP_n \{\pi^{(1)}+\pi^{(2)}+\pi^{(3)})(T_1,\cdots,T_K,\delta_1,\cdots,\delta_K,Y,\widetilde{\delta},d;\alpha,\mathbf{G}, S_D)\},
\end{split}
\end{equation*}
in which some terms irrelevant to $\alpha$ are dropped and $(T_1,\cdots,T_K,\delta_1,\cdots,\delta_K,Y,\widetilde{\delta},d)$ are random counterparts of  $(T_{i1},\cdots,T_{iK},\delta_{i1},\cdots,\delta_{iK},Y_i,\widetilde{\delta}_i,d_i)$. We write $\pi = \pi^{(1)}+\pi^{(2)}+\pi^{(3)}$ and $\ell (\alpha,\mathbf{G}, S_D)= \calP_0 \{\pi(T_1,\cdots,T_K,\delta_1,\cdots,\delta_K,Y,\widetilde{\delta},d;\alpha,\mathbf{G}, S_D)\}$. Let $u_{\alpha}(\alpha,\mathbf{G},S_D) = \partial \pi(\alpha,\mathbf{G},S_D)/\partial\alpha$, $v_{\alpha}(\alpha,\mathbf{G},S_D) = \partial^2 \pi(\alpha,\mathbf{G},S_D)/\partial\alpha^2$, where $\pi(\alpha,\mathbf{G},S_D)$, $u_{\alpha}(\alpha,\mathbf{G},S_D)$ and $v_{\alpha}(\alpha,\mathbf{G},S_D)$ are shorthand versions of 
$\pi(T_1,\cdots,T_K,$ $\delta_1,\cdots,\delta_K,Y,\widetilde{\delta},d;\alpha,\mathbf{G},S_D)$, $u_{\alpha}(T_1,\cdots,T_K,\delta_1,\cdots,\delta_K,Y,\widetilde{\delta},d;\alpha,\mathbf{G},S_D)$ and $v_{\alpha}(T_1,\cdots,T_K,$ $\delta_1,\cdots,\delta_K,Y,\widetilde{\delta},d;\alpha,\mathbf{G},S_D)$.  The score function of $\alpha$ is $U_n(\alpha,\mathbf{G},S_D) = \calP_nu_{\alpha}(\alpha,\mathbf{G},S_D) =0$ and the pseudo score function is $U_n^P(\alpha,\widehat{\mathbf{G}},\widehat{S}_D)=\calP_nu_{\alpha}(\alpha,\widehat{\mathbf{G}},\widehat{S}_D)$.

We then present some regularity conditions for preparation of theoretical justifications.
\begin{cond}\label{cond:Theta}
$K$ is finite. All $\theta_k$ and $\alpha$ belong to a bounded parameter space $\Theta\subset \mathbb{R}$. The generator functions $\phi_{\theta}(x)$ is Lipschitz continuous for $\theta \in \Theta$ with any fixed $0\leq x\leq 1$. 
\end{cond}
\begin{cond}\label{cond:likelihood}
$E\left[\log \frac{L_{i1}(\alpha_1,\mathbf{G}_0, S_{D0})^{\tildedelta_i}L_{i2}(\alpha_1,\mathbf{G}_0, S_{D0})^{1-\tildedelta_i}}{L_{i1}(\alpha_2,\mathbf{G}_0, S_{D0})^{\tildedelta_i}L_{i2}(\alpha_2,\mathbf{G}_0, S_{D0})^{1-\tildedelta_i}}\right]$ exists for all $\alpha_1,\alpha_2\in\Theta$, $i=1,\cdots,n$.
\end{cond}
\begin{cond}\label{cond:phidev}
$\psi_{\alpha}^{(1+K)}(t)$ is bounded for $t\in [0,M_{\psi}]$ for $M_{\psi}<\infty$. And the second derivative of $\phi_{\alpha}$ is bounded for t in an interval $(\epsilon_t, 1)$ for any $\epsilon_t > 0$.
\end{cond}
\begin{cond}\label{cond:tu}
Denote the survival function of the independent censoring variable $C$ by $S_C(\cdot)$. There exist a maximum follow-up $t_U<\infty$ such that $S_{k0}(t_U)S_{D0}(t_U)S_C(t_U)>\epsilon_S>0$ with fixed $\epsilon_S$.
\end{cond}
\begin{cond}\label{cond:IC}
1) For any $k=1,\cdots,K$, the deterministic function $\operatorname{IC}_k(t_k)$ satisfying 
\begin{equation}
\begin{split}
&\frac{\partial}{\partial\epsilon}\calP_0\{u_{\alpha}(\alpha_0,G_{10},\cdots,H_2\{S_{k0}+\epsilon(S_k-S_{k0}),S_{D0};\theta_{k0}\},\cdots,G_
{K0},S_{D0})\}\Big|_{\epsilon=0}\\
&= \int \operatorname{IC}_k (t_k)d(S_k-S_{k0})(t_k)
\end{split}
\end{equation}is uniformly bounded in $t_k\in [0,t_U]$; \\
2) $V_{\theta_k} = \frac{\partial\calP_0\{u_{\alpha}(\alpha_0,G_{10},\cdots,H_2\{S_{k0},S_{D0};\theta_k\},\cdots,G_
{K0},S_{D0})\}}{\partial\theta_k}\Big|_{\theta_k=\theta_{k0}}$ is bounded.\\
3) A function $r(\alpha,\mathbf{G})$ is defined such that $\pi^{(3)}(\alpha,\mathbf{G}, S_D)=(1-\widetilde{\delta})\log\bigg\{\int_Y^{t_U} r(\alpha,\mathbf{G})(t)dS_D(t)\bigg\}$. 
$R(\alpha,\mathbf{G}, S_{D})(t) = \frac{r(\alpha,\mathbf{G})(t)}{\int_Y^{t_U} r(\alpha,\mathbf{G})(x)dS_{D}(x)}$ and its derivative with respect to $\alpha$ are uniformly bounded in $\alpha\in \Theta$ and $S_k\in\mathcal{S}_k$ and $S_D\in\mathcal{S}_D$, where $\mathcal{S}_k$ and $\mathcal{S}_D$ are nonparametric family of variable $\tildeT_k$ and $D$.
\end{cond}
\begin{cond}\label{cond:p0}
$-\calP_0 v_{\alpha}(\alpha_0,\mathbf{G}_0,S_{D0})>\epsilon_v>0$ with fixed $\epsilon_v$.
\end{cond}
Condition \ref{cond:Theta} is necessary for proving Lemma 1 and establishing the consistency of the second-stage estimator. Condition \ref{cond:likelihood} is required for proving the uniqueness of maxima in Lemma 4. Condition \ref{cond:phidev} is a precondition for applying the permanence of the Donsker properties to the closure of the convex hull, similar conditions are imposed in \cite{lakhal2008}. Condition \ref{cond:tu} is also imposed in \cite{lakhal2008} and \cite{jiang2005pseudo} to ensure the asymptotic properties of $\widehat{\theta}$ and weak convergence of $\widehat{S}_k(t)$ in $t\in [0,t_U]$. Conditions \ref{cond:IC} and \ref{cond:p0} are necessary for applying the functional and finite-dimensional delta methods, and subsequently for proving the asymptotic normality of the maximum pseudo-likelihood estimator.

\subsection{Proof of Theorem 1 (i)}

\begin{lem}\label{lem:donsker_Sk}
Under Condition \ref{cond:Theta}, the class $\mathcal{F} = \{H_2(S_k(t_k),S_{D}(t);\theta_k): \theta_k\in\Theta, 0\leq t_k\leq t<\infty, S_k\in\mathcal{S}_k, S_D\in \mathcal{S}_D\}\}$ is $\mathcal{P}_0$-Glivenko-Cantelli, where $\mathcal{S}_k$ and $\mathcal{S}_D$ are nonparametric families of $\tildeT_k$ and $D$, respectively.
\end{lem}

\begin{proof}
According to Theorem 2.4.1 in \cite{vandervaart1996}, to prove $\mathcal{F}$ to be $\mathcal{P}_0$ -Glivenko-Cantelli, it shall be shown that the bracketing number $N_{[~]}(\epsilon,\mathcal{F}, L_1(\mathcal{P}_0))$, which is the minimum number of $\epsilon$-brackets needed to cover $\mathcal{F}$, is finite for any nontrivial $\epsilon>0$. 

Let $\mathcal{F}_1 = \{S_k(t_k)\}$ and $\mathcal{F}_2 = \{S_D(t)\}$. From Theorem 2.7.5 in \cite{vandervaart1996}, the number of brackets $[L_i^{(k)},U_i^{(k)}]$ such that $L_i^{(k)}(t_k)\leq S_k(t_k)\leq U_i^{(k)}(t_k)$ and $\int  |U_i^{(k)}(t_k)- L_i^{(k)}(t_k)|dS_{k0}(t_k)\leq \epsilon$ satisfies $\log N_{[~]}(\epsilon,\mathcal{F}_1, L_1(\mathcal{P}_0))\leq B_1/\epsilon$  for a constant $0<B_1<\infty$. Similarly, the number of brackets $[L_j^{(D)},U_j^{(D)}]$ such that $L_j^{(D)}(t)\leq S_D(t)\leq U_j^{(D)}(t)$ and $\int  |U_j^{(D)}(t)- L_j^{(D)}(t)|dS_{D0}(t)\leq \epsilon$ satisfies $\log N_{[~]}(\epsilon,\mathcal{F}_2, L_1(\mathcal{P}_0))\leq B_2/\epsilon$ for a constant $0<B_2<\infty$. Because $\Theta$ is bounded, we partition $\Theta$ by a set of intervals $[l_m,u_m)$ such that $|u_m-l_m|<\epsilon$, and consequently the number of such intervals is bounded by $B_3/\epsilon$ with a constant $0<B_3<\infty$. 

Now we construct brackets for $\mathcal{F}$. Note that under the Archimedean copula structure, $H_2(S_k(t_k),S_{D}(t);\theta_k)=\psi'_{\theta_k}\left\{\phi_{\theta_k}(S_k(t_k))+\phi_{\theta_k}(S_{D}(t))\right\}\phi'_{\theta_k}(S_{D}(t))$. By construction, the generator functions $\phi_{\theta_k}(\cdot)$ and $\psi_{\theta_k}$ are continuous and strictly decreasing, $\phi'_{\theta_k}(\cdot)$ and $\psi'_{\theta_k}$ are continuous and strictly increasing. Thus, 
\begin{equation*}
\begin{split}
0>&\psi'_{\theta_k}\left\{\phi_{\theta_k}(L_i^{(k)}(t_k))+\phi_{\theta_k}(L_j^{(D)}(t))\right\}
\geq \psi'_{\theta_k}\left\{\phi_{\theta_k}(S_k(t_k))+\phi_{\theta_k}(S_D(t))\right\}\\
&\geq \psi'_{\theta_k}\left\{\phi_{\theta_k}(U_i^{(k)}(t_k))+ \phi_{\theta_k}(U_j^{(D)}(t))\right\}.
\end{split}
\end{equation*} 
and $ \phi'_{\theta_k}(L_j^{(D)}(t))\leq \phi'_{\theta_k}(S_D(t))\leq \phi'_{\theta_k}(U_j^{(D)}(t))<0$. 
It follows that
\begin{equation}\label{eq:LUbound}
\begin{split}
0<&\psi'_{\theta_k}\left\{\phi_{\theta_k}(L_i^{(k)}(t_k))+\phi_{\theta_k}(L_j^{(D)}(t))\right\}\phi'_{\theta_k}(U_j^{(D)}(t))\\
&\leq \psi'_{\theta_k}\left\{\phi_{\theta_k}(S_k(t_k))+\phi_{\theta_k}(S_D(t))\right\}\phi'_{\theta_k}(S_D(t))\\
&\leq \psi'_{\theta_k}\left\{\phi_{\theta_k}(U_i^{(k)}(t_k))+ \phi_{\theta_k}(U_j^{(D)}(t))\right\}\phi'_{\theta_k}(L_j^{(D)}(t)).
\end{split}
\end{equation} 
Define 
\begin{equation*}
\begin{split}
\widetilde{L}_{ijm}(t_k,t) &= \min\bigg\{\psi'_{l_m}\left\{\phi_{l_m}(L_i^{(k)}(t_k))+\phi_{l_m}(L_j^{(D)}(t))\right\}\phi'_{l_m}(U_j^{(D)}(t)),\\
&\quad\quad\quad\quad  \psi'_{u_m}\left\{\phi_{u_m}(L_i^{(k)}(t_k))+\phi_{u_m}(L_j^{(D)}(t))\right\}\phi'_{u_m}(U_j^{(D)}(t))\bigg\},
\end{split}
\end{equation*} 
\begin{equation*}
\begin{split}
\widetilde{U}_{ijm}(t_k,t) &=  \min\bigg\{\psi'_{l_m}\left\{\phi_{l_m}(U_i^{(k)}(t_k))+ \phi_{l_m}(U_j^{(D)}(t))\right\}\phi'_{l_m}(L_j^{(D)}(t))\\
&\quad\quad\quad\quad  \psi'_{u_m}\left\{\phi_{u_m}(U_i^{(k)}(t_k))+ \phi_{u_m}(U_j^{(D)}(t))\right\}\phi'_{u_m}(L_j^{(D)}(t))
\bigg\}.
\end{split}
\end{equation*} 
Clearly, $\widetilde{L}_{ijm}(t_k,t) \leq H_2(S_k(t_k),S_{D}(t);\theta_k)\leq \widetilde{U}_{ijm}(t_k,t)$. And we see that
\begin{align}
&\mathcal{P}_0|\widetilde{U}_{ijm}-\widetilde{L}_{ijm}|\nonumber\\
&\leq \int_0^{\infty}\int_{0}^t \Bigg| \psi'_{l_m}\left\{\phi_{l_m}(U_i^{(k)}(t_k))+ \phi_{l_m}(U_j^{(D)}(t))\right\}\phi'_{l_m}(L_j^{(D)}(t))-\nonumber\\
& \psi'_{l_m}\left\{\phi_{l_m}(L_i^{(k)}(t_k))+\phi_{l_m}(L_j^{(D)}(t))\right\}\phi'_{l_m}(U_j^{(D)}(t))\Bigg|dH(S_{k0}(t_k),S_{D0}(t);\theta_{k0})\label{eq:diff_UL1}\\
&+ \int_0^{\infty}\int_{0}^t \Bigg|\psi'_{u_m}\left\{\phi_{u_m}(U_i^{(k)}(t_k))+ \phi_{u_m}(U_j^{(D)}(t))\right\}\phi'_{u_m}(L_j^{(D)}(t))-\nonumber\\
&\psi'_{u_m}\left\{\phi_{u_m}(L_i^{(k)}(t_k))+\phi_{u_m}(L_j^{(D)}(t))\right\}\phi'_{u_m}(U_j^{(D)}(t))
\Bigg|dH(S_{k0}(t_k),S_{D0}(t);\theta_{k0})\label{eq:diff_UL2}\\
&+ \int_0^{\infty}\int_{0}^t \Bigg|
\psi'_{u_m}\left\{\phi_{u_m}(U_i^{(k)}(t_k))+ \phi_{u_m}(U_j^{(D)}(t))\right\}\phi'_{u_m}(L_j^{(D)}(t))-\nonumber\\
&\psi'_{l_m}\left\{\phi_{l_m}(L_i^{(k)}(t_k))+\phi_{l_m}(L_j^{(D)}(t))\right\}\phi'_{l_m}(U_j^{(D)}(t))
\Bigg|dH(S_{k0}(t_k),S_{D0}(t);\theta_{k0})\label{eq:diff_UL3}\\
&+ \int_0^{\infty}\int_{0}^t \Bigg|
\psi'_{l_m}\left\{\phi_{l_m}(U_i^{(k)}(t_k))+ \phi_{l_m}(U_j^{(D)}(t))\right\}\phi'_{l_m}(L_j^{(D)}(t))-\nonumber\\
&\psi'_{u_m}\left\{\phi_{u_m}(L_i^{(k)}(t_k))+\phi_{u_m}(L_j^{(D)}(t))\right\}\phi'_{u_m}(U_j^{(D)}(t))
\Bigg|dH(S_{k0}(t_k),S_{D0}(t);\theta_{k0}).
\label{eq:diff_UL4}
\end{align}
Since $[L_i^{(k)},U_i^{(k)}]$ and $[L_j^{(D)},U_j^{(D)}]$ are brackets for $\mathcal{F}_1$ and $\mathcal{F}_2$, respectively, from the continuity of generator functions and their derivatives, there exist a constant $0<c_0<\infty$ such that $\eqref{eq:diff_UL1}\leq c_0\epsilon$ and $\eqref{eq:diff_UL2}\leq c_0\epsilon$. Furthermore, from \eqref{eq:LUbound}, we obtain
\begin{equation*}
\begin{split}
\eqref{eq:diff_UL3}
&\leq 2c_o\epsilon+ \int_0^{\infty}\int_{0}^t \Bigg|H_2(S_k(t_k),S_D(t);u_m)-H_2(S_k(t_k),S_D(t);l_m)
\Bigg|dH(S_{k0}(t_k),S_{D0}(t);\theta_{k0})\\
&\leq (2c_0+c_1)\epsilon,
\end{split}
\end{equation*}
for some constant $c_1>0$ derived from Condition \ref{cond:Theta}. Similarly, we have $\eqref{eq:diff_UL4}\leq (2c_0+c_1)\epsilon$. 
Hence we have $N_{[~]}((6c_0+2c_1)\epsilon,\mathcal{F}, L_1(\mathcal{P}_0))\leq \exp(B_1/\epsilon+B_2/\epsilon)B_3/\epsilon$, i.e.  $N_{[~]}(\epsilon,\mathcal{F}, L_1(\mathcal{P}_0))\leq \exp([(B_1+B_2)(6c_0+2c_1)]/\epsilon)B_3(6c_0+2c_1)/\epsilon\leq \exp([(B_1+B_2+B_3)(6c_0+2c_1)]/\epsilon)$. Therefore, $\mathcal{F}$ is $\mathcal{P}_0$ -Glivenko-Cantelli.
\end{proof}

\begin{lem}\label{lem: law of large numbers ln}
Under Condition \ref{cond:Theta}, $\sup\limits_{\alpha,\mathbf{G},S_D}|l_n(\alpha,\mathbf{G}, S_D)-\ell(\alpha,\mathbf{G}, S_D)|\xrightarrow{p}0$.
\end{lem}

\begin{proof}
Let $\mathcal{F}_j = \{\pi^{(j)}(T_1,\cdots,T_K,\delta_1,\cdots,\delta_K,Y,\widetilde{\delta},d;\alpha,\mathbf{G}, S_D)\}$, $j=1,2,3$. Analogue to the proof of Lemma \ref{lem:donsker_Sk}, we can show that  the 
$\int_0^1\sqrt{\log N_{[~]}(\epsilon,\mathcal{F}, L_2(\mathcal{P}_0))}d\epsilon$ is also finite, i.e. $\mathcal{F}$ is $\mathcal{P}_0$-Donsker according to Theorem 2.5.6 in \cite{vandervaart1996}.
Note that $G_k(t_k;t) = H_2\{S_k(t_k),S_{D}(t);\theta_k\}$. By Theorem 2.10.6 and Example 2.10.11 in \cite{vandervaart1996} and the continuity of generator functions,  $\mathcal{F}_1$ and $\mathcal{F}_2$ are $\mathcal{P}_0$-Donsker and thus $\mathcal{P}_0$ -Glivenko-Cantelli. According to the permanence of the Donsker property for the closure of the convex hull in \cite{vandervaart1996}'s Theorem 2.10.3, the class consisting of the inner integrals in $\pi^{(3)}$ is $\mathcal{P}_0$-Donsker since $(-x)^d$ is Lipschitz. Again by this the permanence property,  the class consisting of the double integrals in  $\pi^{(3)}$ is also  $\mathcal{P}_0$-Donsker. Similar to $\mathcal{F}_1$ and $\mathcal{F}_2$, we can show that  $\mathcal{F}_3$ is a Donsker class. By Example  2.10.7 in \cite{vandervaart1996}, $\pi^{(1)}+\pi^{(2)}+\pi^{(3)}$ belongs to a Donsker class, and thus a Glivenko-Cantelli class. Therefore, the result in Lemma \ref{lem: law of large numbers ln} can be obtained using the definition of Glivenko-Cantelli classes.
\end{proof}

\begin{lem}\label{lem:l0}
$\sup\limits_{\alpha\in\Theta}|\ell(\alpha,\widehat{\mathbf{G}}, \widehat{S}_D)-\ell(\alpha,\mathbf{G}_0, S_{D0})|\xrightarrow{p}0$.
\end{lem}

\begin{proof}
We first prove the consistency of $\widehat{G}_k(t_k;t)$, $0\leq t_k,t\leq t_U$. Note that
\begin{align}
&\widehat{G}_k(t_k;t) = H_2\{\widehat{S}_k(t_k),\widehat{S}_D(t);\widehat{\theta}_k\}\nonumber\\
& = H_2\{\widehat{S}_k(t_k),\widehat{S}_D(t);\widehat{\theta}_k\}-H_2\{S_{k0}(t_k),\widehat{S}_D(t);\widehat{\theta}_k\}\label{eq:tildeS1}\\
&+H_2\{S_{k0}(t_k),\widehat{S}_D(t);\widehat{\theta}_k\}-H_2\{S_{k0}(t_k),S_{D0}(t);\widehat{\theta}_k\}\label{eq:tildeS2}\\
&+H_2\{S_{k0}(t_k),S_{D0}(t);\widehat{\theta}_k\}-H_2\{S_{k0}(t_k),S_{D0}(t);\theta_{k0}\}\label{eq:tildeS3}\\
&+H_2\{S_{k0}(t_k),S_{D0}(t);\theta_{k0}\}.\nonumber
\end{align}
By the consistency of $\widehat{S}_k$ \citep{jiang2005pseudo} and the mean value theorem, we have
\begin{equation*}
\begin{split}
&\eqref{eq:tildeS1}
=\left[\psi'_{\widehat{\theta}_k}\left\{\phi_{\widehat{\theta}_k}(\widehat{S}_k(t_k))+\phi_{\widehat{\theta}_k}(\widehat{S}_{D}(t))\right\}-
\psi'_{\widehat{\theta}_k}\left\{\phi_{\widehat{\theta}_k}(S_{k0}(t_k))+\phi_{\widehat{\theta}_k}(\widehat{S}_{D}(t))\right\}\right]\phi'_{\widehat{\theta}_k}(\widehat{S}_{D}(t))\\
& = o_p(1)+\psi^{(2)}_{\widehat{\theta}_k}\left\{\phi_{\widehat{\theta}_k}(S_{k0}(t_k))+\phi_{\widehat{\theta}_k}(\widehat{S}_{D}(t))\right\}\phi'_{\widehat{\theta}_k}(S_{k0}(t_k))[\widehat{S}_k(t_k)-S_{k0}(t_k)]\phi'_{\widehat{\theta}_k}(\widehat{S}_{D}(t))\\
&=o_p(1).
\end{split}
\end{equation*}
By the consistency of the Kaplan-Meier estimator $\widehat{S}_D$ \citep{csorgHo1983rate}, similarly we have $\eqref{eq:tildeS2}=o_p(1)$. Condition \ref{cond:Theta} and the consistency of $\widehat{\theta}_k$ \citep{lakhal2008} imply $\eqref{eq:tildeS3}=o_p(1)$. It then follows that $\widehat{G}_k(t_k;t) = G_{k0}(t_k;t) + o_p(1)$, i.e., $\sup\limits_{t_k,t}|\widehat{G}_k(t_k;t) - G_{k0}(t_k;t)|=o_p(1)$, where $G_{k0}(t_k;t) =H_2\{S_{k0}(t_k),S_{D0}(t);\theta_{k0}\}$. We apply the continuous mapping theorem for statistical functionals, and get $\sup\limits_{t_k,t\leq t_U,\alpha\in\Theta}|\phi_{\alpha}(\widehat{G}_k(t_k; t))-\phi_{\alpha}(G_{k0}(t_k; t))|=o_p(1)$. By the triangle inequality, $$\sup\limits_{\alpha\in\Theta}|\sum\limits_{k=1}^K\phi_{\alpha}(\widehat{G}_k(T_k; Y))-\phi_{\alpha}(G_{k0}(T_k; Y))|=o_p(1).$$ 
Again apply the continuous mapping theorem, we have \begin{equation}\label{eq:pi1}
\sup\limits_{\alpha\in\Theta}|\pi^{(1)}(\alpha,\widehat{\mathbf{G}})-\pi^{(1)}(\alpha,\mathbf{G}_0)|=o_p(1),
\end{equation} 
and similarly 
\begin{equation}\label{eq:pi2}
\sup\limits_{\alpha\in\Theta}|\pi^{(2)}(\alpha,\widehat{\mathbf{G}})-\pi^{(2)}(\alpha,\mathbf{G}_0)|=o_p(1).
\end{equation} 
As for $\pi^{(3)}$, we partition into $\pi^{(3)}(\alpha,\widehat{\mathbf{G}}, \widehat{S}_D)-\pi^{(3)}(\alpha,\mathbf{G}_0, S_{D0}) =\pi^{(3)}(\alpha,\widehat{\mathbf{G}}, \widehat{S}_D)-\pi^{(3)}(\alpha,\widehat{\mathbf{G}}, S_{D0})+ \pi^{(3)}(\alpha,\widehat{\mathbf{G}}, S_{D0})-\pi^{(3)}(\alpha,\mathbf{G}_0, S_{D0}) .$ 
By Hadamard derivative, we see that
\begin{equation*}
\begin{split}
&\frac{\partial}{\partial\epsilon} \pi^{(3)}(\alpha,\mathbf{G}, S_{D0}+\epsilon(\widehat{S}_D-S_{D0}))\Bigg|_{\epsilon=0}\\
& = (1-\widetilde{\delta})\int_{Y}^{t_U}R(\alpha,\mathbf{G}, S_{D0})(t)d(\widehat{S}_D-S_{D0})(t) = o_p(1),
\end{split}
\end{equation*}
where 
\begin{equation*}
\begin{split}
R(\alpha,\mathbf{G}, S_{D})(t) = \frac{r(\alpha,\mathbf{G})(t)}{\int_Y^{t_U} r(\alpha,\mathbf{G})(x)dS_{D}(x)}
<\infty.
\end{split}
\end{equation*}
Thus, $\pi^{(3)}(\alpha,\widehat{\mathbf{G}}, \widehat{S}_D)-\pi^{(3)}(\alpha,\widehat{\mathbf{G}}, S_{D0})=o_p(1)$. Analogue to $\pi^{(1)}$ and $\pi^{(2)}$, we can show that $\sup\limits_{\alpha\in\Theta}|r(\alpha,\widehat{\mathbf{G}})-r(\alpha,\mathbf{G}_0)|=o_p(1),$
which follows by 
\begin{equation}\label{eq:pi3}
\sup\limits_{\alpha\in\Theta}|\pi^{(3)}(\alpha,\widehat{\mathbf{G}}, \widehat{S}_D)-\pi^{(3)}(\alpha,\mathbf{G}_0, S_{D0})|=o_p(1).
\end{equation} 
Coupled with \eqref{eq:pi1} and \eqref{eq:pi2}, we get $\sup\limits_{\alpha\in\Theta}||\ell(\alpha,\widehat{\mathbf{G}}, \widehat{S}_D)-\ell(\alpha,\mathbf{G}_0, S_{D0})||\xrightarrow{p}0$, which completes the proof of Lemma \ref{lem:l0}.
\end{proof}

\begin{lem}\label{lem:maximum}
Under Conditions \ref{cond:Theta}, \ref{cond:likelihood},  for any $\epsilon>0$, $\sup\limits_{\alpha: |\alpha- \alpha_0|>\epsilon}\ell(\alpha,\mathbf{G}_0, S_{D0})< \ell(\alpha_0,\mathbf{G}_0, S_{D0})$.  
\end{lem}

\begin{proof}
For any $\alpha$ with $ |\alpha- \alpha_0|>\epsilon$, Jensen inequality and Condition \ref{cond:likelihood} imply that
\begin{equation*}
\begin{split}
&\ell(\alpha,\mathbf{G}_0, S_{D0})- \ell(\alpha_0,\mathbf{G}_0, S_{D0})\\
& =\lim\limits_{n\rightarrow\infty} n^{-1}\log L_n(\alpha,\mathbf{G}_0, S_{D0})-n^{-1}\log L_n(\alpha_0,\mathbf{G}_0, S_{D0})\\
& =E\left[\log \frac{L_{i1}(\alpha,\mathbf{G}_0, S_{D0})^{\tildedelta_i}L_{i2}(\alpha,\mathbf{G}_0, S_{D0})^{1-\tildedelta_i}}{L_{i1}(\alpha_0,\mathbf{G}_0, S_{D0})^{\tildedelta_i}L_{i2}(\alpha_0,\mathbf{G}_0, S_{D0})^{1-\tildedelta_i}}\right]\\
&\leq \log \left\{E\left[\frac{L_{i1}(\alpha,\mathbf{G}_0, S_{D0})^{\tildedelta_i}L_{i2}(\alpha,\mathbf{G}_0, S_{D0})^{1-\tildedelta_i}}{L_{i1}(\alpha_0,\mathbf{G}_0, S_{D0})^{\tildedelta_i}L_{i2}(\alpha_0,\mathbf{G}_0, S_{D0})^{1-\tildedelta_i}}\right]\right\}\\
&< \log(1)=0.
\end{split}
\end{equation*}
\end{proof}

\begin{proof}[Proof of Theorem 1(i)]
By construction, $\widehat{\alpha} = \arg\max_{\alpha}l_n(\alpha,\widehat{\mathbf{G}}, \widehat{S}_D)$. It follows that
\begin{equation*}
\begin{split}
0< &l_n(\widehat{\alpha},\widehat{\mathbf{G}}, \widehat{S}_D)-l_n(\alpha_0,\widehat{\mathbf{G}}, \widehat{S}_D)\\
&=l_n(\widehat{\alpha},\widehat{\mathbf{G}}, \widehat{S}_D)-\ell(\alpha_0,\widehat{\mathbf{G}}, \widehat{S}_D)+o_p(1)\\
& = l_n(\widehat{\alpha},\widehat{\mathbf{G}}, \widehat{S}_D)-\ell(\alpha_0,\mathbf{G}_0, S_{D0})+o_p(1),\\
\end{split}
\end{equation*}
where the second line is derived from Lemma \ref{lem: law of large numbers ln}, the third line comes from Lemma \ref{lem:l0}. This implies $\ell(\alpha_0,\mathbf{G}_0, S_{D0})<l_n(\widehat{\alpha},\widehat{\mathbf{G}}, \widehat{S}_D)+o_p(1)$ and consequently,
\begin{equation}
\begin{split}
&\ell(\alpha_0,\mathbf{G}_0, S_{D0})-\ell(\widehat{\alpha},\mathbf{G}_0, S_{D0})\\
&<l_n(\widehat{\alpha},\widehat{\mathbf{G}}, \widehat{S}_D)-\ell(\widehat{\alpha},\mathbf{G}_0, S_{D0})+o_p(1)\\
&\leq \sup\limits_{\alpha\in\Theta}|l_n(\alpha,\widehat{\mathbf{G}}, \widehat{S}_D)-\ell(\alpha,\widehat{\mathbf{G}}, \widehat{S}_D)+\ell(\alpha,\widehat{\mathbf{G}}, \widehat{S}_D)-\ell(\alpha,\mathbf{G}_0, S_{D0})|+o_p(1)\\
&\leq \sup\limits_{\alpha,\mathbf{G},S_D}|l_n(\alpha,\mathbf{G}, S_D)-\ell(\alpha,\mathbf{G}, S_D)|+\sup\limits_{\alpha\in\Theta}|\ell(\alpha,\widehat{\mathbf{G}}, \widehat{S}_D)-\ell(\alpha,\mathbf{G}_0, S_{D0})|+o_p(1)=o_p(1).
\label{eq:ell_alpha0}
\end{split}
\end{equation}
Take $\alpha$ such that $|\alpha-\alpha_0|\geq \epsilon$ for any fixed $\epsilon>0$. By Lemma \ref{lem:maximum}, there must exist some $\gamma_{\epsilon}>0$ such that $\ell(\alpha_0,\mathbf{G}_0, S_{D0})>\ell(\widehat{\alpha},\mathbf{G}_0, S_{D0})+\gamma_{\epsilon}$. It follows that 
$$\Pr(|\widehat{\alpha}-\alpha_0|\geq \epsilon)\leq \Pr\{\ell(\alpha_0,\mathbf{G}_0, S_{D0})>\ell(\widehat{\alpha},\mathbf{G}_0, S_{D0})+\gamma_{\epsilon} \}.$$
Equation \eqref{eq:ell_alpha0} implies that $ \Pr\{\ell(\alpha_0,\mathbf{G}_0, S_{D0})>\ell(\widehat{\alpha},\mathbf{G}_0, S_{D0})+\gamma_{\epsilon} \}\rightarrow 0$ as $n\rightarrow \infty$. Therefore, $\Pr(|\widehat{\alpha}-\alpha_0|\geq \epsilon)\rightarrow 0$ as $n\rightarrow \infty$, which completes the proof of Theorem 1(i).
\end{proof}

\subsection{Proof of Theorem 1(ii)}
\begin{proof}
Using the Taylor expansion on the pseudo score function $U_n^P(\alpha,\widehat{\mathbf{G}},\widehat{S}_D)$ around $\alpha_0$, rearranging and evaluating it at $\alpha = \widehat{\alpha}$, we have 
\begin{equation}
n^{1/2}(\widehat{\alpha}-\alpha_0)  = \frac{\sqrt{n}U_n^P(\alpha_0,\widehat{\mathbf{G}},\widehat{S}_D)}{-\calP_n v_{\alpha}(\alpha_0,\widehat{\mathbf{G}},\widehat{S}_D)}
\label{eq:hatalpha0}
\end{equation}
Analogue to the proof of Lemma \ref{lem:l0}, we have $\calP_n v_{\alpha}(\alpha_0,\widehat{\mathbf{G}},\widehat{S}_D)-\calP_n v_{\alpha}(\alpha_0,\mathbf{G},S_D)\xrightarrow{p} 0$. And analogue to the proof of Lemma \ref{lem: law of large numbers ln}, we have $\calP_n v_{\alpha}(\alpha_0,\widehat{\mathbf{G}},\widehat{S}_D)-\calP_0 v_{\alpha}(\alpha_0,\mathbf{G},S_D)\xrightarrow{p} 0$. Next, the numerator in \eqref{eq:hatalpha0} can be written as
\begin{equation}
\begin{split}\label{eq:Up123}
&\sqrt{n}U_n^P(\alpha_0,\widehat{\mathbf{G}},\widehat{S}_D)\\
&=\sqrt{n}(\calP_n-\calP_0)\{(u_{\alpha}(\alpha_0,\widehat{\mathbf{G}},\widehat{S}_D)-
u_{\alpha}(\alpha_0,\mathbf{G}_0,S_{D0}))\}\\
&+ \sqrt{n}(\calP_n-\calP_0)\{u_{\alpha}(\alpha_0,\mathbf{G}_0,S_{D0}))\}
+\sqrt{n}\calP_0\{u_{\alpha}(\alpha_0,\widehat{\mathbf{G}},\widehat{S}_D)\}
\end{split}
\end{equation}
where the first term converges to 0 by the proof of Lemma \ref{lem:l0}, the Donsker property and the dominated convergence theorem, the second term is a sum of $n$ independent and identically distributed zero-mean random variables by the Donsker property, namely
\begin{equation}\label{eq:calPn-calP0}
\sqrt{n}(\calP_n-\calP_0)\{u_{\alpha}(\alpha_0,\mathbf{G}_0,S_{D0}))\} = n^{-1/2}\sum_{i}I_{i0}.
\end{equation} 
The third term in \eqref{eq:Up123} can be decomposed using the Von Mises expansion \citep{mises1947asymptotic,serfling2009approximation} into
\begin{equation}
\begin{split}\label{eq:calP0_ua}
&\sqrt{n}\calP_0\{u_{\alpha}(\alpha_0,\widehat{\mathbf{G}},\widehat{S}_D)\}
=
\sum\limits_{k=1}^K\iint \operatorname{IC}_k^0(t_k,t)d\sqrt{n}(\widehat{G}_k(t_k;t)-G_{k0}(t_k;t))\\
&+\int \operatorname{IC}_D(t) d\sqrt{n}(\widehat{S}_D(t)-S_{D0}(t)) + o(||\sqrt{n}(\widehat{\mathbf{G}}-\mathbf{G}_0)||_{\infty}+||\sqrt{n}(\widehat{S}_D-S_{D0})||_{\infty}),
\end{split}
\end{equation}
where the influence curves $ \operatorname{IC}_k^0$ ($k=1,\cdots,K$) and $ \operatorname{IC}_D$ are obtained by differentiating $\calP_0\{u_{\alpha}(\alpha_0,G_1+\epsilon_1(\widehat{G}_1-G_1),\cdots,G_K+\epsilon_K(\widehat{G}_K-G_K),S_{D0}+\epsilon_{K+1}(\widehat{S}_D-S_{D0}))\}$ with respect to $\epsilon_k$ ($k=1,\cdots,K+1)$ and evaluating at $\epsilon_1=\cdots=\epsilon_{K+1} = 0$. 
We see from \cite{fleming2013counting} that
$\sqrt{n}(\widehat{S}_D(t)-S_{D0}(t))= n^{-1/2}\sum_{i}I_{i1}(t)+o(1),$ 
where $I_{i1}(t) = -S_{D0}(t)\int_{-\infty}^{t}dM_{i}(u)/y(u)$, $M_{i}(u)=\sum_{i}\widetilde{\delta}_i I(Y_i\leq u) - \int_{0}^{u}I(Y_i\geq t)d\wedge_D(t)$, $y(u)=\lim\limits_{n \rightarrow \infty} n^{-1}\sum_iI(Y_i\ge u)$, and $\bm\wedge_{D}(u)$ is the cumulative hazard function of $D$. From \cite{lakhal2008}, we see that $\sqrt{n}(\widehat{\theta}_k-\theta_k)$  is asymptotically equivalent to zero-mean U-statistic so that it can be reduced to sums of independent terms using the Hoeffding decomposition as $\sqrt{n}(\widehat{\theta}_k-\theta_k)=n^{-1/2}\sum_{i}I_{i2}+o(1)$. Furthermore, we see from \cite{jiang2005pseudo} that $\sqrt{n}(\widehat{S}_k(t)-S_{k0}(t))=n^{-1/2}\sum_{i}I_{i3}(t)+o(1).$ The mathematical forms of $I_{i2}$ and $I_{i3}$ can be found in \cite{lakhal2008} and \cite{jiang2005pseudo}, respectively. 
Thus, by the functional and finite-dimensional delta methods,
\begin{equation}\label{eq:calP0_ua2}
\begin{split}
\eqref{eq:calP0_ua}& = \sum\limits_{k=1}^K\iint \operatorname{IC}_k(t_k)d\sqrt{n}(\widehat{S}_k(t_k)-S_{k0}(t_k))\\
&+ \sum\limits_{k=1}^K V_{\theta_k}\sqrt{n}(\widehat{\theta}_k-\theta_{k0})
+n^{-1/2}\sum_{i} \int \operatorname{IC}_D(t) d I_{i1}(t) \\
&+ o(1+|\sqrt{n}(\widehat{\theta}_k-\theta_{k0})|+||\sqrt{n}(\widehat{S}_k-S_k)||_{\infty}+||\sqrt{n}(\widehat{S}_D-S_{D0})||_{\infty})\\
& =n^{-1/2}\sum_{i}\sum\limits_{k=1}^K\iint \operatorname{IC}_k(t_k)dI_{i3}(t_k)
+ n^{-1/2}\sum_{i}\sum\limits_{k=1}^K V_{\theta_k}I_{i2}\\
&+n^{-1/2}\sum_{i} \int \operatorname{IC}_D(t) d I_{i1}(t) 
+ o_p(1),
\end{split}
\end{equation}
where $\operatorname{IC}_k(t_k)$ and $V_{\theta_k}$ are defined in Condition \ref{cond:IC}, and
\begin{equation*}
\begin{split}
\operatorname{IC}_D(t)
= \calP_0\left\{  (1-\widetilde{\delta}) \frac{\partial}{\partial\alpha} R(\alpha,\mathbf{G}_0,S_{D0})(t)  \right\}.
\end{split}
\end{equation*}
Combining \eqref{eq:Up123}-\eqref{eq:calP0_ua2}, the numerator in \eqref{eq:hatalpha0}, $\sqrt{n}U_n^P(\alpha_0,\widehat{\mathbf{G}},\widehat{S}_D)$, can be asymptotically represented by a sum of $n$ independent and identically distributed zero-mean random variables, denoted as $n^{-1/2}\sum_{i}I_{i}$. Therefore, $$n^{1/2}(\widehat{\alpha}-\alpha_0)  = -[\calP_0 v_{\alpha}(\alpha_0,\mathbf{G}_0,S_{D0})]^{-1}n^{-1/2}\sum_{i}I_{i}+o_p(1)$$ is asymptotically normal.
\end{proof}

\section{Simulation studies}
\label{sec:simulation}

\subsection{Predictive accuracy measures}

To investigate the incremental value of incorporating intermediate event information in predicting  $I(D>t)$ or residual lifetime, we compare the prediction algorithms using various accuracy measures. The first measure is the mean square prediction error (MSPE) given by $\operatorname{MSPE} = n^{-1}\sum\limits_{i=1}^n (\min(D_i,t_U^{*}) - \operatorname{CMST}_i)^2$. The second measure is the quantile prediction error (QPE) given by $\operatorname{QPE} = n^{-1}\sum\limits_{i=1}^n \rho_{\tau}(\min(D_i,t_U^{*}) - \operatorname{CQST}_i)$, where $\rho_{\tau}(x)= x[\tau-I(x<0)]$ is the quantile check function. 
Both MSPE and QPE 
metrics are based on true values $D_i$ which are known in simulation studies but unknown in practice. As such, these two metrics are adjusted with inverse probability of censoring weights (IPCW) to evaluate agreement between predicted and observed survival times. 
Additionally, we consider the Brier score (BS) defined as
$BS(t) 
=n^{-1}\sum\limits_{i=1}^n w_i(t)I(t> t_{mi}^{*} )\left[ I(Y_i>t)- \widehat{S}^{*}(t|\{T_{ik}:\delta_{ik}=1,k=1,\cdots,K\})\right]^2
$
and its integrated version $\operatorname{IBS} = \int_0^{t_U^{*}}BS(t)dt$, where
$t_{mi}^{*} = \max\{T_{ik}:\delta_{ik}=1,k=1,\cdots,K\}$, 
$w_i(t)=\tildedelta_iI(Y_i\leq t)/\widehat{S}_c(Y_i-)+
I(Y_i> t)/\widehat{S}_c(t)$
and $\widehat{S}_c(t)$ is the KM estimate of the survival function of censoring times. We set $t_U^{*}=12$ in simulation studies 
to allow for a standardized comparison across different scenarios, and in the real data analysis, we set the maximum follow-up time as $t_U^{*}$. 
The three metrics above quantify overall prediction error. 
And as suggested by an anonymous reviewer,  we included the time-dependent AUC to assess discrimination, as defined by 
\citep{uno2007evaluating} 
\[ AUC(t) =\frac{\sum\limits_{i,j}w_i(t)w_j(t)I(t>t_{mi}^{*})I(t>t_{mj}^{*})I(Y_i\leq t)I(Y_j>t)I(\widehat{S}^{*}_i(t)<\widehat{S}^{*}_j(t)}{\sum\limits_{i,j}w_i(t)w_j(t)I(t>t_{mi}^{*})I(t>t_{mj}^{*})I(Y_i\leq t)I(Y_j>t)}.\]
The consistency and normality of these accuracy measures can be justified using the uniform consistency
of the survival prediction estimates and functional central limit theorem along with modifications of the arguments given by \cite{zheng2008time} and \cite{tian2007model}.

To further assess performance of individual prediction,
we construct $95\%$ prediction intervals for each test individual $i$ as follows. 
Let $\operatorname{CQST}_{i,0.025}$ and $\operatorname{CQST}_{i,0.975}$ denote the estimated 2.5\% and 97.5\%  conditional quantile survival times. The prediction interval for restricted death time $\min(D_i,t_U^{*})$ is construct as $[\operatorname{CQST}_{i,0.025},\operatorname{CQST}_{i,0.975}]$. The accuracy and reliability of individual prediction intervals are assessed in terms of empirical coverage probability (CP) $$n_{\operatorname{test}}^{-1}\sum\limits_{i=1}^{n_{\operatorname{test}}} I\left\{ \min(D_i,t_U^{*}) \in [\operatorname{CQST}_{i,0.025},\operatorname{CQST}_{i,0.975}]\right\}$$ and median interval width $$\operatorname{Median}\left\{ \operatorname{CQST}_{i,0.975}-\operatorname{CQST}_{i,0.025}]\right\}.$$

\subsection{Similation setups (Ex1 and Ex2)}
We conducted extensive simulation studies to evaluate the performance of our proposal for association estimation and predictions for overall survival.  
Before describing the simulation setups, we define Kendall's tau, a common metric for assessing the overall association between two survival times. It is defined as the probability of concordance minus the probability of discordance. Under Archimedean copula structure, Kendall’s $\tau$ can be expressed as \citep{lakhal2008}
\begin{equation}\label{eq:Kendall tau}
\operatorname{Kendall's~ tau} = 4\int_0^1 \phi_{\theta}(u)/\phi'_{\theta}(u) du+1,
\end{equation}
which is unrelated to the marginal distributions. 

We considered two practical settings in the following examples. The first example (Ex1) includes 
$K=3$ or $7$ intermediate events with different associations with death time. The true Kendall's taus corresponding to $(\theta_1,\cdots,\theta_K)$ were set to be $\tau_{\theta}$ $=(0.8, \cdots,0.8-0.6(k-1)/(K-1),\cdots, 0.2)$. In the second example (Ex2),  all $K$ intermediate events were assumed to equally contribute to overall survival with $\tau_{\theta} = (0.5,\cdots,0.5)$.
In both examples, we take Kendall's tau associated with $\alpha$ as $\tau_{\alpha} = 0.2$ or $0.5$, providing a weak or strong association between intermediate events.  
The marginal components $S_k$ and $S_D$ are survival functions of the exponential distributions with rates of 1 and 0.6, respectively.

Based on two samples 
generated under the Frank copula with setups Ex1 and Ex2, respectively, we presented 
the sample Kendall correlation coefficients of the true survival times $(\tildeT_1,\cdots,\tildeT_K,D)$ in Figure \ref{fig:simdata}.  
Figure \ref{fig:simdata}(a) presents different (unconditional) association structures among intermediate events in Ex1, while Figure \ref{fig:simdata}(b) depicts the same association structure in Ex2. 
This illustrates the flexibility of the proposed algorithm and its capability to accommodate various (unconditional) dependence structures among the intermediate and terminal event times.

We generated $n=n_{\operatorname{train}} +n_{\operatorname{test}}$ sets of $(\tildeT_1,\cdots,\tildeT_K,D)$ under models (1)-(3) 
on both the upper and lower wedges, where the number of training datasets  $n_{\operatorname{train}}= 100$ or $200$ for parameter estimation, and $n_{\operatorname{test}} =50$ test datasets were used for model evaluation. 
The independent censoring variable $C$ followed a $\operatorname{Uniform}[0,20]$ or $\operatorname{Uniform}[0,5]$ distribution, resulting in $10\%$ or $35\%$ censoring for $D$. 
For each simulation setting, 200 replications were generated to examine the performance of the proposed method.

\begin{figure}[h]
\centering
\includegraphics[width=0.6\textwidth,trim=100 210 100 130,clip]{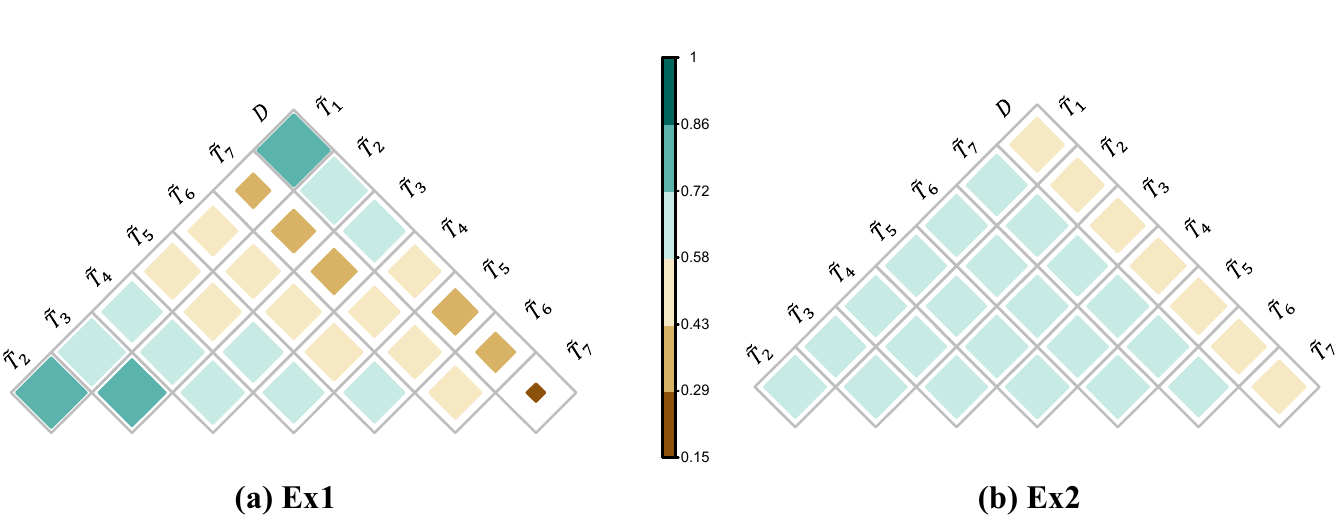}
\caption{Pairwise comparisons and 
sample Kendall correlation coefficients for 200 sets of $(\tildeT_1,\cdots,\tildeT_7,D)$ generated under models (1)-(3) on both wedge. The setup in Figure (a) is same as Ex1 with true Kendall's taus corresponding to $(\alpha,\theta_1,\cdots,\theta_7)$ equal to $(0.5,0.8, 0.7, 0.6, 0.5, 0.4, 0.3, 0.2)$. The setup in Figure (b) is same as Ex2 with true Kendall's taus corresponding to $(\alpha,\theta_1,\cdots,\theta_7)$ equal to $(0.5,0.5, \cdots,0.5)$.} 
\label{fig:simdata}
\end{figure}


\subsection{Simulation results for association analysis (Ex1 and Ex2)}
We considered the Frank copula function in both models (1)-(3) 
first. 
Our estimation results for $\theta_k$ and $S_k$ exhibited similar trends to those reported in \cite{lakhal2008} and \cite{jiang2005pseudo} and were thus omitted. Table \ref{table:alpha} summarizes estimates of $\tau_{\alpha}$ 
in terms of relative mean bias (RBias, namely, the mean bias divided by the truth), empirical standard deviation (SD), and coverage probability (CP) in Ex1 and Ex2 with different numbers $K$ of intermediate events. 
The proposed maximum pseudo-likelihood estimator $\widehat{\alpha}$ performed well across all setups. The relative mean bias was minimal and generally decreased as 
the sample size increased, with the coverage probability remaining close to the nominal level of 0.95.  The empirical SD was notably small and decreased as $n_{\operatorname{train}}$ or $K$ increased. 
Moreover, the performance of $\widehat{\alpha}$ 
remained consistent across different rates of independent censoring. Similar trends can be found in the results for the Gumbel and Clayton copula structures (see Tables \ref{table:alpha gumber}-\ref{table:alpha clayton}). 

\subsection{Competing survival prediction algorithms}
For comparison and to showcase the dynamic prediction proposed, we summarize existing survival prediction methods: landmark prediction using the Kaplan-Meier (KM) estimator of $S_D$ (P0), survival prediction incorporating the $k$-th intermediate outcome (P$k$), and landmark prediction incorporating the $k$-th intermediate outcome and the maximum observed time (P$k$m). 
In particular, the method
P0 estimates the probability $S^{*}_{(P0)}(t|t_m^{*}) = \Pr(D>t| D> t_m^{*} )  =\Pr(D>t)/\Pr(D> t_m^{*} )$. By plugging in the KM estimator $\widehat{S}_D$, we obtain the estimator $S^{*}_{(P0)}$, denoted by $\widehat{S}^{*}_{(P0)}$. Consequently, P0-derived CMST is given by $t_m^{*} + \int_{t_m^{*}}^{t_U^{*}}\widehat{S}^{*}_{(P0)}(t|t_{m}^{*})dt$. Methods P$k$ and P$k$m focus on estimating the probability $S^{*}_{(Pk)}(t|t_k) = S^{*}(t| \tildeT_k = t_k)$ and $S^{*}_{(Pkm)}(t|t_k,t_m^{*}) = \Pr(D>t|\tildeT_k=t_k, D> t_m^{*} )=H_1\{S_k(t_k),S_{D}(t);\theta_k\}/H_1\{S_k(t_k),S_{D}(t_m^{*});\theta_k\}$. Thus with the plug-in estimators $\widehat{S}^{*}_{(Pk)}(t|t_k)$ and $\widehat{S}^{*}_{(Pkm)}(t|t_k,t_{m}^{*})$, the P$k$-derived and P$k$m-derived CMST are given by $t_k + \int_{t_k}^{t_U^{*}}\widehat{S}^{*}_{(Pk)}(t|t_k)dt$, $t_m^{*} + \int_{t_m^{*}}^{t_U^{*}}\widehat{S}^{*}_{(Pkm)}(t|t_k,t_{m}^{*})dt$, respectively. Their CQST can be calculated in a similar manner.

\subsection{Simulation results of survival prediction algorithms (Ex1-2)}

Tables \ref{table:ex1_n100_cr5}-\ref{table:ex2_n200_cr5} 
summarize predictive accuracy measures for Ex1 and Ex2 with the Frank copula structure using the proposed dynamic prediction (DP) algorithm in comparison with the other competing algorithms (P0, P1, P1m, PK, and PKm). 
These tables include averaged mean squared prediction error (MSPE), quantile prediction error (QPE) with $\tau=0.5$ and integrated Brier Score (IBS), 
along with their empirical standard deviations and relative prediction accuracy. The relative prediction accuracy was computed using the proposed DP algorithm as a benchmark, with a higher value indicating better performance. 
Figure 
\ref{fig:EX2 bs} 
presents averaged Brier score curves across 200 training and test samples. 

These tables and the figure demonstrate the following: 
1) The proposed DP algorithm consistently outperformed the other algorithms 
across different scenarios for both training and testing datasets. 
As $n_{\operatorname{train}}$ increased, the prediction error of the DP method decreased.
2) The prediction error of the P0 method was $42\%\sim 150\%$ higher than that of the DP method, highlighting the advantages 
of incorporating intermediate event information in prediction. 
3) Given the same value of Kendall's tau 
$\tau_{\alpha}$, the prediction errors and their empirical standard deviations of the DP algorithm decreased as the number of intermediate events increased 
in most cases.
4) The DP algorithm's outperformance 
over those competing algorithms was more pronounced when all intermediate event times had moderate to low 
association with death.
Specifically, in the first example (EX1), where the $K$th intermediate event had the lowest effect on death, PK and P0 performed worst, followed by PKm, P1, P1m, and DP. While methods P1, P1m, and DP were comparable in some specific cases regarding certain predictive measures. Overall, the DP method performed better and was more stable. 
In the second example (Ex2), where the 1st and $K$th intermediate events 
were equally associated with death, a clear difference 
can be found between single and simultaneous analysis of multiple intermediate events. The ranking in terms of prediction accuracy was DP $>$ P1m $=$ PKm $>$ PK $=$ P1. As $K$ increased, the difference between DP and P1m became more evident. These results underscore the importance of integrating information on intermediate events and the maximum observed follow-up time into prediction of survival. 
5) The DP method exhibited a degree of robustness to 
variations  in $\tau_{\alpha}$, $K$ and the independent censoring rate. 

We also increased the training sample size to 2000. Table \ref{table:ex12_n2000_cr20} shows a trend similar to that for sample sizes 100 and 200. As suggested by an anonymous reviewer, we added a multi-event comparison (a grouped subset of intermediates) in simulation studies to demonstrate how predictive accuracy improves with the number of events included. The simulation results reported in Table \ref{table:ex12_group} indicate that prediction results improve when multiple intermediate events with strong associations to death time are incorporated into the model.

Two anonymous reviewers suggested evaluating robustness across different training–test configurations.
We further complemented the single train/test split with a repeated K-fold cross-validation procedure as well as a general random cross-validation procedure. Splits were stratified so that the event/censoring distribution is balanced across folds.

\textbf{The $M$-fold cross-validation} randomly splits the data into $M$ disjoint sets of about equal size and labels them as $\mathcal{I}_j$, $j=1,\cdots,M$. An estimate $(\widehat{\alpha}_{-j},\widehat{\bm{\theta}}_{-j})$ is obtained based on all observations that are not in $\mathcal{I}_j$. We then compute the predicted error estimate based on observations in $\mathcal{I}_j$ and get averaged prediction error. 

\textbf{The general random cross-validation scheme} randomly splits the data into training and testing sets in some ratio. Let $n$ be the sizes of the whole sample, and $n_{train}$ and $n_{test}$ be the sizes of the training and test subsamples, where $n/n_{test}$ is roughly a fixed positive integer. For each specified ratio, we repeatedly draw random training subset with remaining subjects as test subset, estimate model parameters using the training set, and compute prediction errors on the corresponding test set. This procedure is repeated with fresh random training-testing splits,  and the prediction errors are averaged over all repetitions. 

For Examples Ex1-2, we adopted these two repeated random splitting strategies to reduce variability arising from a single data partition. Specifically, we implemented 3-fold cross-validation repeated 20 times (yielding 60 splits), and a general random cross-validation scheme with a testing proportion $n_{test}/n=1/3$ or $1/5$ as well as 20 repetitions. Table \ref{table:ex12_split} summarizes predictive performance of different methods in terms of three metrics (MSPE, IBS and AUC at $t=3$). Across all evaluation matrics, the relative predictive performance of different methods keeps consistent, and the two random splitting strategies lead to comparable results.

Furthermore, we conducted additional simulations to evaluate the performance of individual prediction intervals for the proposed method, using both the training samples and independent test sets of 50 new individuals. 
As shown in Table \ref{tab:predictionInterval}, the in-sample individual prediction intervals achieve empirical coverage close to the nominal 0.95 level under the training dataset and generally become narrower with higher observation rates in Example Ex1. Moreover, as the training sample size increases, the out-of-sample prediction intervals evaluated on the testing data become more reasonable. From a overall averaged perspective (MSPE/QPE/IBS/AUC) for prediction estimates, the proposed method yields the smallest prediction error compared with existing approaches, and this error further decreases as training sample size increases. Taken together, these results demonstrate that the proposed method provides accurate prediction estimates and reliable individual-level prediction intervals, a key advantage of our proposed joint modeling approach over existing methods such as landmarking models.  
A correctly specified joint model in our method can lead to 
consistent predictions, wheras landmarking methods fail to capture 
the joint distribution of multiple event times. 

\subsection{Robustness against model misspecification}
We examined the performances of the proposed 
prediction algorithms (Algorithm 2 in Appendix A) 
under model misspecification. 
Results were summarized in Tables \ref{table:mis gumbel} - \ref{table:misspecification comparison} 
where datasets were generated using the Frank copula models but fitted with the Gumbel or Clayton models. 
As evident from Tables \ref{table:ex1_n100_cr5} -\ref{table:ex2_n100_cr5} and S.9- S.11,  when the true model is the Frank copula, the proposed prediction algorithms using the misspecified 
Gumbel copula model yield results comparable to those of the true Frank model in terms of average ranking across all predictive accuracy measures, while the prediction based on the misspecified Clayton model performs 
relatively poorly in Ex1. On the other hand, under the Ex2 setup, predictions based on both the Gumbel and Clayton models appear to be similarly poor, indicating a certain degree of their sensitivity to model misspecification. 

\subsection{Robustness against variations in unobservable regions (Ex3)}
It is important to note that the proposed method is insensitive  
to data in the unobservable region, where $\tildeT_k > D$. To demonstrate it, we further simulated data with different joint densities for $(\tildeT_k, D)$ across upper and lower wedges. In this simulation example (Ex3), we used Frank copula models and considered model (1) 
with $S_{k,D}^{(u)}=(S_k= S_D=S, \tau_{\theta_k}=\tau^{(u)})$ for $\tildeT_k\leq D$, and  $S_{k,D}^{(l)}=(S_k= S_D=S, \tau_{\theta_k}=\tau^{(l)})$ 
for $\tildeT_k> D$, $k=1,\cdots,7$. $S$ is the survival function of the $Exp(1)$ 
distribution, $\tau^{(u)}=0.5$ and $\tau^{(l)}$ is  set to either 0.3, 0.5 or 0.7.  The independent censoring variable $C$ was generated from a $\operatorname{Uniform}[0,10]$ distribution, yielding approximately $50\%$ informative censoring and $6\%$ independent censoring for $\tildeT_k$, and $10\%$ independent censoring for $D$. 
Table \ref{table:ex3mix_theta}
presents the estimation results for $\alpha$ and $\tau^{(u)}$ under different combinations of the true values for $\tau_{\alpha}$ and $\tau^{(l)}$. 
Table \ref{table:ex3mix_DP} reports the performance results of the proposed dynamic prediction algorithm. Results in both tables 
show that $\widehat{\alpha}$, $\widehat{\theta}_k$ and the proposed 
survival prediction are all insensitive to the choice of $S_{k,D}^{(l)}$, supporting the applicability and effectiveness of the proposed method on the lower wedge.

\subsection{Computational times}

Using a standard computer (Intel(R) Core(TM) i9-14900KF), the average runtime per replicate across simulation scenarios is reported in Table \ref{tab:runtime}.
Since the current implementation is entirely written in R, computational speed is limited. 
Computational efficiency could be 
improved in future work by implementing performance-critical components in C++.

\begin{table}[ht]
\scriptsize
\tabcolsep=5pt
\caption{Estimation results of $\tau_{\alpha}$ under train datasets of Ex1 and Ex2.}
\label{table:alpha}
\begin{center}
\begin{tabular}{lllllllllll}
\toprule
&&&&\multicolumn{3}{c}{$10\%$ censoring for $D$}&&\multicolumn{3}{c}{$35\%$ censoring for $D$}\tabularnewline
\cmidrule(lr){5-7}\cmidrule(lr){9-11}
\multicolumn{1}{c}{Example}&\multicolumn{1}{c}{$n_{\operatorname{train}}$}&\multicolumn{1}{c}{K}&\multicolumn{1}{c}{truth}&\multicolumn{1}{c}{RBias}&\multicolumn{1}{c}{SD}&\multicolumn{1}{c}{CP}&\multicolumn{1}{c}{}&\multicolumn{1}{c}{RBias}&\multicolumn{1}{c}{SD}&\multicolumn{1}{c}{CP}\tabularnewline
\hline
Ex1&100&3&0.2&-0.069&0.050&0.930&&-0.065&0.059&0.945\tabularnewline
&&&0.5&-0.039&0.079&0.970&&-0.047&0.089&0.970\tabularnewline
\cline{3-11}
&&7&0.2&-0.054&0.032&0.950&&-0.048&0.035&0.940\tabularnewline
&&&0.5&-0.056&0.043&0.930&&-0.068&0.038&0.855\tabularnewline
\cline{2-11}
&200&3&0.2&-0.039&0.038&0.945&&-0.043&0.043&0.955\tabularnewline
&&&0.5&-0.031&0.039&0.955&&-0.053&0.044&0.935\tabularnewline
\cline{3-11}
&&7&0.2&-0.027&0.026&0.945&&-0.037&0.025&0.945\tabularnewline
&&&0.5&-0.024&0.025&0.930&&-0.035&0.024&0.900\tabularnewline
\hline
Ex2&100&3&0.2&-0.006&0.050&0.950&&0.049&0.060&0.945\tabularnewline
&&&0.5&-0.021&0.051&0.950&&-0.025&0.052&0.955\tabularnewline
\cline{3-11}
&&7&0.2&-0.034&0.033&0.940&&0.001&0.035&0.945\tabularnewline
&&&0.5&-0.036&0.04&0.940&&-0.049&0.044&0.940\tabularnewline
\cline{2-11}
&200&3&0.2&-0.009&0.037&0.950&&-0.010&0.040&0.935\tabularnewline
&&&0.5&-0.008&0.037&0.905&&-0.001&0.039&0.960\tabularnewline
\cline{3-11}
&&7&0.2&-0.005&0.026&0.950&&0.001&0.028&0.935\tabularnewline
&&&0.5&-0.016&0.028&0.930&&-0.005&0.028&0.945\tabularnewline
\bottomrule
\end{tabular}\end{center}
\end{table}

\begin{table}[ht]
\scriptsize
\tabcolsep=5pt
\caption{Estimation results of $\tau_{\alpha}$ under train datasets of Ex1 and Ex2 with Gumbel copula structure.}
\label{table:alpha gumber}
\begin{center}
\begin{tabular}{lllllllllll}
\toprule
&&&&\multicolumn{3}{c}{$10\%$ censoring for $D$}&&\multicolumn{3}{c}{$35\%$ censoring for $D$}\tabularnewline
\cmidrule(lr){5-7}\cmidrule(lr){9-11}
\multicolumn{1}{c}{Example}&\multicolumn{1}{c}{$n_{\operatorname{train}}$}&\multicolumn{1}{c}{K}&\multicolumn{1}{c}{truth}&\multicolumn{1}{c}{RBias}&\multicolumn{1}{c}{SD}&\multicolumn{1}{c}{CP}&\multicolumn{1}{c}{}&\multicolumn{1}{c}{RBias}&\multicolumn{1}{c}{SD}&\multicolumn{1}{c}{CP}\tabularnewline
\hline
Ex1&100&3&0.2&-0.051&0.061&0.965&&-0.024&0.062&0.975\tabularnewline
&&&0.5&-0.095&0.061&0.885&&-0.112&0.064&0.885\tabularnewline
\cline{3-11}
&&7&0.2&-0.033&0.047&0.935&&-0.027&0.046&0.94\tabularnewline
&&&0.5&-0.085&0.052&0.88&&-0.081&0.053&0.87\tabularnewline
\cline{2-11}
&200&3&0.2&-0.057&0.045&0.95&&-0.022&0.049&0.93\tabularnewline
&&&0.5&-0.053&0.046&0.925&&-0.069&0.054&0.915\tabularnewline
\cline{3-11}
&&7&0.2&0.007&0.035&0.95&&0.017&0.035&0.945\tabularnewline
&&&0.5&-0.048&0.038&0.925&&-0.057&0.038&0.91\tabularnewline
\hline
Ex2&100&3&0.2&0.058&0.065&0.95&&0.072&0.067&0.945\tabularnewline
&&&0.5&-0.035&0.061&0.935&&-0.009&0.063&0.945\tabularnewline
\cline{3-11}
&&7&0.2&0.024&0.045&0.955&&0.051&0.051&0.92\tabularnewline
&&&0.5&-0.045&0.055&0.92&&-0.046&0.046&0.885\tabularnewline
\cline{2-11}&200&3&0.2&0.016&0.047&0.95&&0.041&0.048&0.95\tabularnewline
&&&0.5&-0.013&0.046&0.945&&-0.012&0.051&0.95\tabularnewline
\cline{3-11}
&&7&0.2&0.022&0.039&0.94&&0.035&0.037&0.955\tabularnewline
&&&0.5&-0.026&0.04&0.915&&-0.011&0.04&0.97\tabularnewline
\bottomrule
\end{tabular}\end{center}
\end{table}

\begin{table}[ht]
\scriptsize
\tabcolsep=5pt
\caption{Estimation results of $\tau_{\alpha}$ under train datasets of Ex1 and Ex2 with Clayton copula structure.}
\label{table:alpha clayton}
\begin{center}
\begin{tabular}{lllllllllll}
\toprule
&&&&\multicolumn{3}{c}{$10\%$ censoring for $D$}&&\multicolumn{3}{c}{$35\%$ censoring for $D$}\tabularnewline
\cmidrule(lr){5-7}\cmidrule(lr){9-11}
\multicolumn{1}{c}{Example}&\multicolumn{1}{c}{$n_{\operatorname{train}}$}&\multicolumn{1}{c}{K}&\multicolumn{1}{c}{truth}&\multicolumn{1}{c}{RBias}&\multicolumn{1}{c}{SD}&\multicolumn{1}{c}{CP}&\multicolumn{1}{c}{}&\multicolumn{1}{c}{RBias}&\multicolumn{1}{c}{SD}&\multicolumn{1}{c}{CP}\tabularnewline
\hline
Ex1&100&3&0.2&0.047&0.141&0.925&&0.111&0.108&0.92\tabularnewline
&&&0.5&-0.121&0.119&0.94&&-0.103&0.132&0.925\tabularnewline
\cline{3-11}
&&7&0.2&-0.089&0.091&0.955&&0.085&0.115&0.92\tabularnewline
&&&0.5&-0.126&0.114&0.97&&-0.169&0.086&0.9\tabularnewline
\cline{2-11}
&200&3&0.2&0.057&0.115&0.93&&0.113&0.104&0.911\tabularnewline
&&&0.5&-0.095&0.085&0.955&&-0.022&0.12&0.94\tabularnewline
\cline{3-11}
&&7&0.2&-0.063&0.056&0.955&&0.108&0.097&0.91\tabularnewline
&&&0.5&-0.099&0.066&0.9&&-0.083&0.078&0.925\tabularnewline
\hline
Ex2&100&3&0.2&0.025&0.106&0.975&&0.11&0.117&0.97\tabularnewline
&&&0.5&-0.073&0.076&0.915&&-0.059&0.081&0.945\tabularnewline
\cline{3-11}
&&7&0.2&-0.09&0.054&0.98&&-0.008&0.057&0.975\tabularnewline
&&&0.5&-0.115&0.061&0.85&&-0.108&0.066&0.895\tabularnewline
\cline{2-11}
&200&3&0.2&-0.002&0.074&0.98&&0.07&0.063&0.98\tabularnewline
&&&0.5&-0.036&0.065&0.95&&-0.037&0.057&0.95\tabularnewline
\cline{3-11}
&&7&0.2&-0.064&0.035&0.98&&-0.012&0.034&0.95\tabularnewline
&&&0.5&-0.07&0.044&0.875&&-0.067&0.049&0.9\tabularnewline
\bottomrule
\end{tabular}\end{center}
\end{table}

\begin{table}[ht]
\scriptsize
\tabcolsep=3pt
\caption{Predictive accuracy of the proposed dynamic prediction (DP) method in comparison with the competing methods for training and testing 
datasets under Ex1 ($n_{\operatorname{train}}=100$, $n_{\operatorname{test}}=50$) with $35\%$ censoring for $D$. Values in each parenthesis are the empirical standard deviation and relative prediction accuracy.} 
\label{table:ex1_n100_cr5}
\begin{center}
\begin{tabular}{llllrlll}
\toprule
\multirow{2}*{}& \multicolumn{3}{c}{In-sample}&&\multicolumn{3}{c}{Out-of-sample}\tabularnewline
\cmidrule(lr){2-4}\cmidrule(lr){6-8}
{Method}&{MSPE}&{QPE}&{IBS}&&{MSPE}&{QPE}&{IBS}\tabularnewline
\hline
\multicolumn{8}{c}{$K=3$, $\tau_{\alpha} =0.2$}\tabularnewline
DP&1.493 (0.61,1)&0.241 (0.055,1)&0.055 (0.015,1)&$$&1.59 (0.899,1)&0.259 (0.076,1)&0.058 (0.017,1)\tabularnewline
P0&2.713 (0.803,0.55)&0.459 (0.073,0.53)&0.141 (0.026,0.39)&$$&2.872 (1.101,0.55)&0.478 (0.082,0.54)&0.141 (0.027,0.41)\tabularnewline
P1&1.906 (0.746,0.78)&0.247 (0.053,0.98)&0.053 (0.012,1.04)&$$&2.038 (1.106,0.78)&0.263 (0.074,0.98)&0.058 (0.016,1)\tabularnewline
P1m&1.892 (0.746,0.79)&0.24 (0.053,1)&0.05 (0.012,1.1)&$$&2.027 (1.106,0.78)&0.257 (0.074,1.01)&0.055 (0.016,1.05)\tabularnewline
PK&2.854 (0.839,0.52)&0.473 (0.069,0.51)&0.15 (0.026,0.37)&$$&3.035 (1.195,0.52)&0.492 (0.085,0.53)&0.151 (0.03,0.38)\tabularnewline
PKm&2.576 (0.791,0.58)&0.42 (0.071,0.57)&0.125 (0.024,0.44)&$$&2.73 (1.121,0.58)&0.438 (0.084,0.59)&0.126 (0.026,0.46)\tabularnewline
\hline
\multicolumn{8}{c}{$K=3$, $\tau_{\alpha} =0.5$}\tabularnewline
DP&1.479 (0.598,1)&0.221 (0.05,1)&0.059 (0.015,1)&$$&1.588 (0.866,1)&0.241 (0.075,1)&0.064 (0.02,1)\tabularnewline
P0&2.777 (0.845,0.53)&0.444 (0.073,0.5)&0.145 (0.025,0.41)&$$&2.918 (1.109,0.54)&0.461 (0.081,0.52)&0.147 (0.024,0.44)\tabularnewline
P1&1.992 (0.8,0.74)&0.228 (0.046,0.97)&0.051 (0.01,1.16)&$$&2.142 (1.11,0.74)&0.247 (0.074,0.98)&0.057 (0.015,1.12)\tabularnewline
P1m&1.99 (0.8,0.74)&0.226 (0.046,0.98)&0.05 (0.01,1.18)&$$&2.143 (1.11,0.74)&0.245 (0.074,0.98)&0.056 (0.015,1.14)\tabularnewline
PK&2.973 (0.907,0.5)&0.474 (0.071,0.47)&0.159 (0.026,0.37)&$$&3.139 (1.223,0.51)&0.488 (0.087,0.49)&0.164 (0.03,0.39)\tabularnewline
PKm&2.654 (0.847,0.56)&0.409 (0.071,0.54)&0.131 (0.024,0.45)&$$&2.794 (1.111,0.57)&0.424 (0.081,0.57)&0.133 (0.023,0.48)\tabularnewline
\hline
\multicolumn{8}{c}{$K=7$, $\tau_{\alpha} =0.2$}\tabularnewline
DP&1.353 (0.619,1)&0.233 (0.051,1)&0.049 (0.012,1)&$$&1.377 (0.78,1)&0.24 (0.062,1)&0.054 (0.014,1)\tabularnewline
P0&2.603 (0.755,0.52)&0.482 (0.078,0.48)&0.168 (0.032,0.29)&$$&2.665 (0.962,0.52)&0.499 (0.08,0.48)&0.169 (0.027,0.32)\tabularnewline
P1&1.707 (0.701,0.79)&0.258 (0.05,0.9)&0.062 (0.012,0.79)&$$&1.728 (0.953,0.8)&0.268 (0.061,0.9)&0.068 (0.018,0.79)\tabularnewline
P1m&1.679 (0.696,0.81)&0.246 (0.05,0.95)&0.056 (0.013,0.88)&$$&1.701 (0.953,0.81)&0.255 (0.061,0.94)&0.061 (0.017,0.89)\tabularnewline
PK&2.751 (0.809,0.49)&0.482 (0.069,0.48)&0.174 (0.029,0.28)&$$&2.83 (1.067,0.49)&0.495 (0.075,0.48)&0.178 (0.032,0.3)\tabularnewline
PKm&2.443 (0.75,0.55)&0.441 (0.074,0.53)&0.15 (0.029,0.33)&$$&2.496 (0.958,0.55)&0.455 (0.074,0.53)&0.151 (0.025,0.36)\tabularnewline
\hline
\multicolumn{8}{c}{$K=7$, $\tau_{\alpha} =0.5$}\tabularnewline
DP&1.316 (0.548,1)&0.217 (0.045,1)&0.058 (0.012,1)&$$&1.34 (0.705,1)&0.221 (0.063,1)&0.062 (0.015,1)\tabularnewline
P0&2.638 (0.747,0.5)&0.472 (0.074,0.46)&0.17 (0.03,0.34)&$$&2.741 (0.906,0.49)&0.485 (0.079,0.46)&0.172 (0.028,0.36)\tabularnewline
P1&1.745 (0.697,0.75)&0.237 (0.046,0.92)&0.06 (0.012,0.97)&$$&1.815 (0.898,0.74)&0.245 (0.064,0.9)&0.064 (0.016,0.97)\tabularnewline
P1m&1.74 (0.695,0.76)&0.235 (0.046,0.92)&0.059 (0.013,0.98)&$$&1.813 (0.888,0.74)&0.242 (0.063,0.91)&0.063 (0.016,0.98)\tabularnewline
PK&2.816 (0.798,0.47)&0.488 (0.072,0.44)&0.179 (0.03,0.32)&$$&2.929 (0.973,0.46)&0.498 (0.079,0.44)&0.184 (0.031,0.34)\tabularnewline
PKm&2.482 (0.738,0.53)&0.431 (0.07,0.5)&0.152 (0.029,0.38)&$$&2.579 (0.909,0.52)&0.442 (0.076,0.5)&0.154 (0.027,0.4)\tabularnewline
\bottomrule
\end{tabular}\end{center}
\end{table}

\begin{table}[ht]
\scriptsize
\tabcolsep=3pt
\caption{Predictive accuracy of the proposed dynamic prediction (DP) method in comparison with the competing methods for training and testing 
datasets under Ex2 ($n_{\operatorname{train}}=100$, $n_{\operatorname{test}}=50$) with $35\%$ censoring for $D$. Values in each parenthesis are the empirical standard deviation and relative prediction accuracy.} 
\label{table:ex2_n100_cr5}
\begin{center}
\begin{tabular}{llllrlll}
\toprule
\multirow{2}*{}& \multicolumn{3}{c}{In-sample}&&\multicolumn{3}{c}{Out-of-sample}\tabularnewline
\cmidrule(lr){2-4}\cmidrule(lr){6-8}
{Method}&{MSPE}&{QPE}&{IBS}&&{MSPE}&{QPE}&{IBS}\tabularnewline
\hline
\multicolumn{8}{c}{$K=3$, $\tau_{\alpha} =0.2$}\tabularnewline
DP&1.657 (0.608,1)&0.282 (0.057,1)&0.079 (0.017,1)&$$&1.736 (0.796,1)&0.295 (0.071,1)&0.08 (0.02,1)\tabularnewline
P0&2.829 (0.852,0.59)&0.444 (0.073,0.64)&0.145 (0.027,0.54)&$$&2.924 (1.007,0.59)&0.459 (0.085,0.64)&0.143 (0.029,0.56)\tabularnewline
P1&2.529 (0.832,0.66)&0.368 (0.062,0.77)&0.112 (0.021,0.71)&$$&2.617 (1.045,0.66)&0.379 (0.076,0.78)&0.114 (0.029,0.7)\tabularnewline
P1m&2.398 (0.81,0.69)&0.336 (0.061,0.84)&0.098 (0.02,0.81)&$$&2.483 (1,0.7)&0.347 (0.072,0.85)&0.099 (0.022,0.81)\tabularnewline
PK&2.524 (0.827,0.66)&0.366 (0.063,0.77)&0.113 (0.022,0.7)&$$&2.635 (1.049,0.66)&0.383 (0.075,0.77)&0.115 (0.028,0.7)\tabularnewline
PKm&2.4 (0.813,0.69)&0.336 (0.062,0.84)&0.098 (0.021,0.81)&$$&2.493 (1.011,0.7)&0.351 (0.074,0.84)&0.1 (0.024,0.8)\tabularnewline
\hline
\multicolumn{8}{c}{$K=3$, $\tau_{\alpha} =0.5$}\tabularnewline
DP&1.822 (0.614,1)&0.306 (0.06,1)&0.097 (0.02,1)&$$&1.849 (0.821,1)&0.32 (0.074,1)&0.102 (0.023,1)\tabularnewline
P0&2.952 (0.828,0.62)&0.442 (0.069,0.69)&0.154 (0.028,0.63)&$$&3.003 (1.066,0.62)&0.459 (0.08,0.7)&0.157 (0.03,0.65)\tabularnewline
P1&2.686 (0.824,0.68)&0.365 (0.064,0.84)&0.116 (0.023,0.84)&$$&2.72 (1.088,0.68)&0.381 (0.077,0.84)&0.123 (0.031,0.83)\tabularnewline
P1m&2.588 (0.82,0.7)&0.342 (0.062,0.89)&0.105 (0.022,0.92)&$$&2.629 (1.074,0.7)&0.357 (0.073,0.9)&0.112 (0.028,0.91)\tabularnewline
PK&2.672 (0.83,0.68)&0.364 (0.063,0.84)&0.116 (0.023,0.84)&$$&2.708 (1.078,0.68)&0.378 (0.073,0.85)&0.121 (0.03,0.84)\tabularnewline
PKm&2.576 (0.813,0.71)&0.339 (0.061,0.9)&0.105 (0.022,0.92)&$$&2.626 (1.069,0.7)&0.356 (0.07,0.9)&0.112 (0.027,0.91)\tabularnewline
\hline
\multicolumn{8}{c}{$K=7$, $\tau_{\alpha} =0.2$}\tabularnewline
DP&1.39 (0.533,1)&0.261 (0.05,1)&0.066 (0.013,1)&$$&1.517 (0.853,1)&0.271 (0.068,1)&0.069 (0.016,1)\tabularnewline
P0&2.555 (0.711,0.54)&0.47 (0.068,0.56)&0.168 (0.027,0.39)&$$&2.728 (1.058,0.56)&0.486 (0.08,0.56)&0.169 (0.029,0.41)\tabularnewline
P1&2.237 (0.705,0.62)&0.385 (0.059,0.68)&0.131 (0.024,0.5)&$$&2.393 (1.076,0.63)&0.397 (0.07,0.68)&0.132 (0.024,0.52)\tabularnewline
P1m&2.092 (0.68,0.66)&0.353 (0.059,0.74)&0.113 (0.023,0.58)&$$&2.236 (1.045,0.68)&0.365 (0.071,0.74)&0.114 (0.022,0.61)\tabularnewline
PK&2.229 (0.695,0.62)&0.384 (0.056,0.68)&0.131 (0.023,0.5)&$$&2.408 (1.102,0.63)&0.394 (0.072,0.69)&0.133 (0.025,0.52)\tabularnewline
PKm&2.089 (0.665,0.67)&0.353 (0.057,0.74)&0.114 (0.022,0.58)&$$&2.243 (1.051,0.68)&0.366 (0.071,0.74)&0.116 (0.023,0.59)\tabularnewline
\hline
\multicolumn{8}{c}{$K=7$, $\tau_{\alpha} =0.5$}\tabularnewline
DP&1.62 (0.608,1)&0.293 (0.054,1)&0.091 (0.018,1)&$$&1.542 (0.779,1)&0.298 (0.072,1)&0.093 (0.021,1)\tabularnewline
P0&2.709 (0.841,0.6)&0.458 (0.073,0.64)&0.169 (0.036,0.54)&$$&2.651 (0.971,0.58)&0.465 (0.081,0.64)&0.166 (0.03,0.56)\tabularnewline
P1&2.408 (0.791,0.67)&0.376 (0.057,0.78)&0.129 (0.025,0.71)&$$&2.324 (0.993,0.66)&0.38 (0.078,0.78)&0.129 (0.028,0.72)\tabularnewline
P1m&2.295 (0.78,0.71)&0.353 (0.058,0.83)&0.117 (0.025,0.78)&$$&2.215 (0.961,0.7)&0.357 (0.073,0.83)&0.117 (0.025,0.79)\tabularnewline
PK&2.407 (0.816,0.67)&0.375 (0.059,0.78)&0.127 (0.024,0.72)&$$&2.336 (1.014,0.66)&0.383 (0.077,0.78)&0.13 (0.028,0.72)\tabularnewline
PKm&2.293 (0.791,0.71)&0.352 (0.058,0.83)&0.116 (0.024,0.78)&$$&2.226 (0.974,0.69)&0.359 (0.073,0.83)&0.117 (0.026,0.79)\tabularnewline
\bottomrule
\end{tabular}\end{center}
\end{table}


\begin{table}
\scriptsize
\tabcolsep=3pt
\caption{Predictive accuracy of the proposed dynamic prediction (DP) method in comparison with the competing methods for training and testing datasets under Ex1 ($n_{train}=100$, $n_{test}=50$) with $10\%$ censoring for $D$. Values in each parenthesis are the empirical standard deviation and relative prediction accuracy.}
\label{table:ex1_n100_cr20}
\begin{center}
\begin{tabular}{llllrlll}
\toprule
\multirow{2}*{}& \multicolumn{3}{c}{In-sample}&&\multicolumn{3}{c}{Out-of-sample}\tabularnewline
\cmidrule(lr){2-4}\cmidrule(lr){6-8}
{Method}&{MSPE}&{QPE}&{IBS}&&{MSPE}&{QPE}&{IBS}\tabularnewline
\hline
\multicolumn{8}{c}{$K=3$, $\tau_{\alpha} =0.2$}\tabularnewline
DP&0.856 (0.337,1)&0.197 (0.037,1)&0.038 (0.006,1)&$$&0.9 (0.555,1)&0.207 (0.055,1)&0.039 (0.008,1)\tabularnewline
P0&2.055 (0.549,0.42)&0.426 (0.063,0.46)&0.085 (0.011,0.45)&$$&2.151 (0.704,0.42)&0.438 (0.081,0.47)&0.086 (0.013,0.45)\tabularnewline
P1&0.921 (0.404,0.93)&0.213 (0.036,0.92)&0.039 (0.006,0.97)&$$&0.973 (0.591,0.92)&0.221 (0.056,0.94)&0.041 (0.009,0.95)\tabularnewline
P1m&0.901 (0.398,0.95)&0.205 (0.036,0.96)&0.037 (0.006,1.03)&$$&0.957 (0.585,0.94)&0.214 (0.057,0.97)&0.039 (0.009,1)\tabularnewline
PK&2.348 (0.637,0.36)&0.474 (0.066,0.42)&0.096 (0.013,0.4)&$$&2.429 (0.929,0.37)&0.482 (0.083,0.43)&0.098 (0.016,0.4)\tabularnewline
PKm&1.842 (0.528,0.46)&0.387 (0.064,0.51)&0.077 (0.012,0.49)&$$&1.94 (0.703,0.46)&0.399 (0.077,0.52)&0.079 (0.013,0.49)\tabularnewline
\hline
\multicolumn{8}{c}{$K=3$, $\tau_{\alpha} =0.5$}\tabularnewline
DP&0.837 (0.349,1)&0.185 (0.036,1)&0.033 (0.006,1)&$$&0.888 (0.534,1)&0.2 (0.058,1)&0.035 (0.009,1)\tabularnewline
P0&2.035 (0.607,0.41)&0.42 (0.067,0.44)&0.073 (0.011,0.45)&$$&2.141 (0.708,0.41)&0.442 (0.073,0.45)&0.076 (0.011,0.46)\tabularnewline
P1&0.925 (0.447,0.9)&0.204 (0.036,0.91)&0.033 (0.007,1)&$$&0.958 (0.672,0.93)&0.214 (0.056,0.93)&0.034 (0.01,1.03)\tabularnewline
P1m&0.93 (0.448,0.9)&0.203 (0.037,0.91)&0.033 (0.007,1)&$$&0.965 (0.667,0.92)&0.214 (0.055,0.93)&0.034 (0.01,1.03)\tabularnewline
PK&2.338 (0.661,0.36)&0.481 (0.066,0.38)&0.084 (0.012,0.39)&$$&2.426 (0.883,0.37)&0.497 (0.079,0.4)&0.086 (0.015,0.41)\tabularnewline
PKm&1.836 (0.57,0.46)&0.385 (0.063,0.48)&0.067 (0.011,0.49)&$$&1.943 (0.702,0.46)&0.407 (0.068,0.49)&0.07 (0.011,0.5)\tabularnewline
\hline
\multicolumn{8}{c}{$K=7$, $\tau_{\alpha} =0.2$}\tabularnewline
DP&0.755 (0.369,1)&0.172 (0.035,1)&0.035 (0.006,1)&$$&0.751 (0.473,1)&0.182 (0.045,1)&0.038 (0.007,1)\tabularnewline
P0&2.195 (0.672,0.34)&0.446 (0.067,0.39)&0.099 (0.013,0.35)&$$&2.243 (0.718,0.33)&0.459 (0.074,0.4)&0.102 (0.014,0.37)\tabularnewline
P1&0.972 (0.434,0.78)&0.216 (0.039,0.8)&0.045 (0.007,0.78)&$$&0.962 (0.543,0.78)&0.223 (0.049,0.82)&0.048 (0.01,0.79)\tabularnewline
P1m&0.918 (0.415,0.82)&0.2 (0.038,0.86)&0.041 (0.007,0.85)&$$&0.919 (0.513,0.82)&0.208 (0.048,0.88)&0.044 (0.009,0.86)\tabularnewline
PK&2.427 (0.753,0.31)&0.469 (0.07,0.37)&0.109 (0.015,0.32)&$$&2.412 (0.803,0.31)&0.48 (0.071,0.38)&0.113 (0.016,0.34)\tabularnewline
PKm&1.965 (0.621,0.38)&0.404 (0.065,0.43)&0.09 (0.013,0.39)&$$&2.013 (0.694,0.37)&0.418 (0.072,0.44)&0.094 (0.014,0.4)\tabularnewline
\hline
\multicolumn{8}{c}{$K=7$, $\tau_{\alpha} =0.5$}\tabularnewline
DP&0.746 (0.349,1)&0.17 (0.033,1)&0.029 (0.005,1)&$$&0.776 (0.458,1)&0.189 (0.048,1)&0.032 (0.007,1)\tabularnewline
P0&2.119 (0.601,0.35)&0.441 (0.063,0.39)&0.075 (0.01,0.39)&$$&2.21 (0.658,0.35)&0.462 (0.069,0.41)&0.078 (0.011,0.41)\tabularnewline
P1&0.908 (0.43,0.82)&0.208 (0.038,0.82)&0.033 (0.006,0.88)&$$&0.969 (0.581,0.8)&0.226 (0.053,0.84)&0.037 (0.009,0.86)\tabularnewline
P1m&0.91 (0.434,0.82)&0.206 (0.039,0.83)&0.033 (0.007,0.88)&$$&0.971 (0.57,0.8)&0.224 (0.053,0.84)&0.036 (0.009,0.89)\tabularnewline
PK&2.359 (0.694,0.32)&0.485 (0.067,0.35)&0.084 (0.012,0.35)&$$&2.479 (0.81,0.31)&0.503 (0.068,0.38)&0.088 (0.014,0.36)\tabularnewline
PKm&1.9 (0.578,0.39)&0.403 (0.063,0.42)&0.068 (0.011,0.43)&$$&1.986 (0.642,0.39)&0.422 (0.064,0.45)&0.072 (0.011,0.44)\tabularnewline
\bottomrule
\end{tabular}\end{center}
\end{table}

\begin{table}
\scriptsize
\tabcolsep=3pt
\caption{Predictive accuracy of the proposed dynamic prediction (DP) method in comparison with the competing methods for training and testing datasets under Ex1 ($n_{train}=200$, $n_{test}=50$) with $10\%$ censoring for $D$. Values in each parenthesis are the empirical standard deviation and relative prediction accuracy.}
\label{table:ex1_n200_cr20}
\begin{center}
\begin{tabular}{llllrlll}
\toprule
\multirow{2}*{}& \multicolumn{3}{c}{In-sample}&&\multicolumn{3}{c}{Out-of-sample}\tabularnewline
\cmidrule(lr){2-4}\cmidrule(lr){6-8}
{Method}&{MSPE}&{QPE}&{IBS}&&{MSPE}&{QPE}&{IBS}\tabularnewline
\hline
\multicolumn{8}{c}{$K=3$, $\tau_{\alpha} =0.2$}\tabularnewline
DP&0.851 (0.253,1)&0.199 (0.027,1)&0.033 (0.004,1)&$$&0.834 (0.504,1)&0.206 (0.052,1)&0.034 (0.008,1)\tabularnewline
P0&2.085 (0.428,0.41)&0.431 (0.044,0.46)&0.074 (0.007,0.45)&$$&2.076 (0.636,0.4)&0.439 (0.065,0.47)&0.074 (0.01,0.46)\tabularnewline
P1&0.939 (0.301,0.91)&0.217 (0.027,0.92)&0.034 (0.004,0.97)&$$&0.903 (0.598,0.92)&0.223 (0.054,0.92)&0.035 (0.009,0.97)\tabularnewline
P1m&0.915 (0.298,0.93)&0.209 (0.027,0.95)&0.033 (0.004,1)&$$&0.888 (0.595,0.94)&0.215 (0.053,0.96)&0.033 (0.009,1.03)\tabularnewline
PK&2.381 (0.492,0.36)&0.484 (0.046,0.41)&0.084 (0.008,0.39)&$$&2.339 (0.843,0.36)&0.493 (0.076,0.42)&0.084 (0.015,0.4)\tabularnewline
PKm&1.858 (0.4,0.46)&0.39 (0.044,0.51)&0.067 (0.007,0.49)&$$&1.86 (0.624,0.45)&0.4 (0.062,0.51)&0.067 (0.01,0.51)\tabularnewline
\hline
\multicolumn{8}{c}{$K=3$, $\tau_{\alpha} =0.5$}\tabularnewline
DP&0.822 (0.242,1)&0.19 (0.024,1)&0.03 (0.003,1)&$$&0.8 (0.447,1)&0.194 (0.049,1)&0.03 (0.007,1)\tabularnewline
P0&2.072 (0.439,0.4)&0.435 (0.042,0.44)&0.067 (0.006,0.45)&$$&2.057 (0.608,0.39)&0.437 (0.059,0.44)&0.066 (0.008,0.45)\tabularnewline
P1&0.918 (0.311,0.9)&0.213 (0.025,0.89)&0.031 (0.004,0.97)&$$&0.888 (0.562,0.9)&0.213 (0.049,0.91)&0.03 (0.007,1)\tabularnewline
P1m&0.921 (0.312,0.89)&0.212 (0.025,0.9)&0.03 (0.004,1)&$$&0.892 (0.563,0.9)&0.213 (0.049,0.91)&0.03 (0.008,1)\tabularnewline
PK&2.401 (0.51,0.34)&0.498 (0.044,0.38)&0.077 (0.008,0.39)&$$&2.352 (0.777,0.34)&0.496 (0.072,0.39)&0.076 (0.012,0.39)\tabularnewline
PKm&1.871 (0.426,0.44)&0.399 (0.041,0.48)&0.061 (0.006,0.49)&$$&1.848 (0.597,0.43)&0.4 (0.057,0.48)&0.061 (0.008,0.49)\tabularnewline
\hline
\multicolumn{8}{c}{$K=7$, $\tau_{\alpha} =0.2$}\tabularnewline
DP&0.675 (0.232,1)&0.177 (0.025,1)&0.027 (0.003,1)&$$&0.702 (0.447,1)&0.182 (0.046,1)&0.027 (0.006,1)\tabularnewline
P0&2.184 (0.432,0.31)&0.457 (0.041,0.39)&0.077 (0.007,0.35)&$$&2.241 (0.533,0.31)&0.467 (0.059,0.39)&0.077 (0.009,0.35)\tabularnewline
P1&0.907 (0.282,0.74)&0.221 (0.027,0.8)&0.036 (0.004,0.75)&$$&0.932 (0.511,0.75)&0.222 (0.05,0.82)&0.036 (0.009,0.75)\tabularnewline
P1m&0.855 (0.269,0.79)&0.205 (0.026,0.86)&0.032 (0.004,0.84)&$$&0.878 (0.473,0.8)&0.208 (0.046,0.88)&0.033 (0.007,0.82)\tabularnewline
PK&2.374 (0.47,0.28)&0.484 (0.043,0.37)&0.084 (0.007,0.32)&$$&2.402 (0.835,0.29)&0.491 (0.082,0.37)&0.085 (0.015,0.32)\tabularnewline
PKm&1.937 (0.401,0.35)&0.413 (0.04,0.43)&0.07 (0.007,0.39)&$$&1.987 (0.526,0.35)&0.422 (0.059,0.43)&0.07 (0.009,0.39)\tabularnewline
\hline
\multicolumn{8}{c}{$K=7$, $\tau_{\alpha} =0.5$}\tabularnewline
DP&0.677 (0.227,1)&0.172 (0.024,1)&0.026 (0.003,1)&$$&0.745 (0.45,1)&0.187 (0.05,1)&0.027 (0.007,1)\tabularnewline
P0&2.103 (0.449,0.32)&0.448 (0.042,0.38)&0.068 (0.007,0.38)&$$&2.178 (0.634,0.34)&0.461 (0.063,0.41)&0.069 (0.009,0.39)\tabularnewline
P1&0.879 (0.301,0.77)&0.215 (0.028,0.8)&0.031 (0.004,0.84)&$$&0.951 (0.59,0.78)&0.228 (0.055,0.82)&0.033 (0.009,0.82)\tabularnewline
P1m&0.878 (0.299,0.77)&0.215 (0.028,0.8)&0.031 (0.004,0.84)&$$&0.945 (0.573,0.79)&0.228 (0.054,0.82)&0.032 (0.008,0.84)\tabularnewline
PK&2.368 (0.516,0.29)&0.497 (0.048,0.35)&0.077 (0.008,0.34)&$$&2.47 (0.849,0.3)&0.511 (0.076,0.37)&0.078 (0.013,0.35)\tabularnewline
PKm&1.889 (0.427,0.36)&0.41 (0.043,0.42)&0.063 (0.007,0.41)&$$&1.966 (0.636,0.38)&0.424 (0.063,0.44)&0.064 (0.009,0.42)\tabularnewline
\bottomrule
\end{tabular}\end{center}
\end{table}

\begin{table}
\scriptsize
\tabcolsep=3pt
\caption{Predictive accuracy of the proposed dynamic prediction (DP) method in comparison with the competing methods for training and testing datasets under Ex1 ($n_{train}=200$, $n_{test}=50$) with $35\%$ censoring for $D$. Values in each parenthesis are the empirical standard deviation and relative prediction accuracy.}
\label{table:ex1_n200_cr5}
\begin{center}
\begin{tabular}{llllrlll}
\toprule
\multirow{2}*{}& \multicolumn{3}{c}{In-sample}&&\multicolumn{3}{c}{Out-of-sample}\tabularnewline
\cmidrule(lr){2-4}\cmidrule(lr){6-8}
{Method}&{MSPE}&{QPE}&{IBS}&&{MSPE}&{QPE}&{IBS}\tabularnewline
\hline
\multicolumn{8}{c}{$K=3$, $\tau_{\alpha} =0.2$}\tabularnewline
DP&1.419 (0.37,1)&0.237 (0.034,1)&0.059 (0.009,1)&$$&1.48 (0.799,1)&0.254 (0.07,1)&0.057 (0.016,1)\tabularnewline
P0&2.717 (0.492,0.52)&0.458 (0.049,0.52)&0.153 (0.019,0.39)&$$&2.748 (0.999,0.54)&0.475 (0.07,0.53)&0.143 (0.021,0.4)\tabularnewline
P1&1.884 (0.462,0.75)&0.243 (0.034,0.98)&0.057 (0.009,1.04)&$$&1.906 (1.019,0.78)&0.26 (0.068,0.98)&0.057 (0.014,1)\tabularnewline
P1m&1.872 (0.46,0.76)&0.236 (0.033,1)&0.054 (0.009,1.09)&$$&1.89 (1.014,0.78)&0.253 (0.068,1)&0.054 (0.014,1.06)\tabularnewline
PK&2.863 (0.532,0.5)&0.477 (0.048,0.5)&0.164 (0.021,0.36)&$$&2.896 (1.078,0.51)&0.489 (0.08,0.52)&0.155 (0.029,0.37)\tabularnewline
PKm&2.575 (0.491,0.55)&0.42 (0.05,0.56)&0.136 (0.019,0.43)&$$&2.609 (1.019,0.57)&0.436 (0.072,0.58)&0.128 (0.02,0.45)\tabularnewline
\hline
\multicolumn{8}{c}{$K=3$, $\tau_{\alpha} =0.5$}\tabularnewline
DP&1.374 (0.39,1)&0.22 (0.032,1)&0.06 (0.011,1)&$$&1.405 (0.856,1)&0.231 (0.077,1)&0.059 (0.015,1)\tabularnewline
P0&2.738 (0.554,0.5)&0.452 (0.045,0.49)&0.155 (0.019,0.39)&$$&2.79 (1.088,0.5)&0.468 (0.07,0.49)&0.149 (0.02,0.4)\tabularnewline
P1&1.895 (0.519,0.73)&0.23 (0.03,0.96)&0.053 (0.008,1.13)&$$&1.924 (1.107,0.73)&0.24 (0.075,0.96)&0.053 (0.014,1.11)\tabularnewline
P1m&1.895 (0.518,0.73)&0.228 (0.031,0.96)&0.052 (0.008,1.15)&$$&1.924 (1.106,0.73)&0.239 (0.075,0.97)&0.052 (0.014,1.13)\tabularnewline
PK&2.935 (0.583,0.47)&0.485 (0.043,0.45)&0.17 (0.021,0.35)&$$&2.999 (1.169,0.47)&0.501 (0.086,0.46)&0.164 (0.029,0.36)\tabularnewline
PKm&2.598 (0.54,0.53)&0.414 (0.044,0.53)&0.138 (0.018,0.43)&$$&2.654 (1.109,0.53)&0.431 (0.073,0.54)&0.134 (0.021,0.44)\tabularnewline
\hline
\multicolumn{8}{c}{$K=7$, $\tau_{\alpha} =0.2$}\tabularnewline
DP&1.288 (0.381,1)&0.232 (0.033,1)&0.047 (0.009,1)&$$&1.192 (0.71,1)&0.233 (0.067,1)&0.047 (0.013,1)\tabularnewline
P0&2.567 (0.501,0.5)&0.484 (0.045,0.48)&0.159 (0.02,0.3)&$$&2.495 (0.87,0.48)&0.485 (0.063,0.48)&0.151 (0.022,0.31)\tabularnewline
P1&1.655 (0.451,0.78)&0.258 (0.031,0.9)&0.06 (0.008,0.78)&$$&1.581 (0.865,0.75)&0.258 (0.061,0.9)&0.058 (0.014,0.81)\tabularnewline
P1m&1.629 (0.45,0.79)&0.246 (0.031,0.94)&0.054 (0.008,0.87)&$$&1.549 (0.863,0.77)&0.246 (0.06,0.95)&0.052 (0.013,0.9)\tabularnewline
PK&2.726 (0.545,0.47)&0.492 (0.045,0.47)&0.168 (0.02,0.28)&$$&2.64 (0.945,0.45)&0.492 (0.072,0.47)&0.163 (0.027,0.29)\tabularnewline
PKm&2.412 (0.495,0.53)&0.445 (0.045,0.52)&0.143 (0.019,0.33)&$$&2.336 (0.88,0.51)&0.446 (0.065,0.52)&0.136 (0.021,0.35)\tabularnewline
\hline
\multicolumn{8}{c}{$K=7$, $\tau_{\alpha} =0.5$}\tabularnewline
DP&1.347 (0.389,1)&0.221 (0.031,1)&0.052 (0.01,1)&$$&1.295 (0.764,1)&0.228 (0.063,1)&0.05 (0.012,1)\tabularnewline
P0&2.711 (0.552,0.5)&0.475 (0.047,0.47)&0.155 (0.021,0.34)&$$&2.602 (0.987,0.5)&0.478 (0.068,0.48)&0.145 (0.023,0.34)\tabularnewline
P1&1.827 (0.508,0.74)&0.246 (0.031,0.9)&0.056 (0.008,0.93)&$$&1.739 (0.982,0.74)&0.255 (0.063,0.89)&0.056 (0.015,0.89)\tabularnewline
P1m&1.823 (0.507,0.74)&0.244 (0.031,0.91)&0.055 (0.008,0.95)&$$&1.731 (0.982,0.75)&0.253 (0.062,0.9)&0.055 (0.016,0.91)\tabularnewline
PK&2.916 (0.575,0.46)&0.498 (0.044,0.44)&0.166 (0.02,0.31)&$$&2.838 (1.106,0.46)&0.506 (0.077,0.45)&0.159 (0.029,0.31)\tabularnewline
PKm&2.568 (0.538,0.52)&0.438 (0.044,0.5)&0.14 (0.019,0.37)&$$&2.468 (0.991,0.52)&0.443 (0.065,0.51)&0.132 (0.021,0.38)\tabularnewline
\bottomrule
\end{tabular}\end{center}
\end{table}

\begin{table}
\scriptsize
\tabcolsep=3pt
\caption{Predictive accuracy of the proposed dynamic prediction (DP) method in comparison with the competing methods for training and testing datasets under Ex2 ($n_{train}=100$, $n_{test}=50$) with $10\%$ censoring for $D$. Values in each parenthesis are the empirical standard deviation and relative prediction accuracy.}
\label{table:ex2_n100_cr20}
\begin{center}
\begin{tabular}{llllrlll}
\toprule
\multirow{2}*{}& \multicolumn{3}{c}{In-sample}&&\multicolumn{3}{c}{Out-of-sample}\tabularnewline
\cmidrule(lr){2-4}\cmidrule(lr){6-8}
{Method}&{MSPE}&{QPE}&{IBS}&&{MSPE}&{QPE}&{IBS}\tabularnewline
\hline
\multicolumn{8}{c}{$K=3$, $\tau_{\alpha} =0.2$}\tabularnewline
DP&1.178 (0.452,1)&0.266 (0.05,1)&0.053 (0.009,1)&$$&1.282 (0.57,1)&0.282 (0.062,1)&0.056 (0.011,1)\tabularnewline
P0&2.089 (0.626,0.56)&0.432 (0.068,0.62)&0.086 (0.012,0.62)&$$&2.211 (0.689,0.58)&0.449 (0.079,0.63)&0.09 (0.015,0.62)\tabularnewline
P1&1.728 (0.577,0.68)&0.369 (0.057,0.72)&0.074 (0.011,0.72)&$$&1.82 (0.732,0.7)&0.378 (0.073,0.75)&0.077 (0.015,0.73)\tabularnewline
P1m&1.505 (0.534,0.78)&0.325 (0.055,0.82)&0.064 (0.011,0.83)&$$&1.611 (0.656,0.8)&0.34 (0.068,0.83)&0.067 (0.013,0.84)\tabularnewline
PK&1.724 (0.627,0.68)&0.37 (0.06,0.72)&0.074 (0.012,0.72)&$$&1.82 (0.722,0.7)&0.384 (0.072,0.73)&0.077 (0.016,0.73)\tabularnewline
PKm&1.497 (0.554,0.79)&0.325 (0.056,0.82)&0.063 (0.011,0.84)&$$&1.6 (0.638,0.8)&0.341 (0.064,0.83)&0.067 (0.013,0.84)\tabularnewline
\hline
\multicolumn{8}{c}{$K=3$, $\tau_{\alpha} =0.5$}\tabularnewline
DP&1.361 (0.499,1)&0.305 (0.06,1)&0.053 (0.009,1)&$$&1.447 (0.582,1)&0.326 (0.063,1)&0.056 (0.011,1)\tabularnewline
P0&2.057 (0.632,0.66)&0.433 (0.068,0.7)&0.075 (0.012,0.71)&$$&2.148 (0.709,0.67)&0.455 (0.069,0.72)&0.08 (0.013,0.7)\tabularnewline
P1&1.686 (0.565,0.81)&0.366 (0.062,0.83)&0.063 (0.011,0.84)&$$&1.744 (0.73,0.83)&0.382 (0.065,0.85)&0.067 (0.014,0.84)\tabularnewline
P1m&1.555 (0.552,0.88)&0.34 (0.061,0.9)&0.057 (0.011,0.93)&$$&1.628 (0.693,0.89)&0.357 (0.063,0.91)&0.062 (0.013,0.9)\tabularnewline
PK&1.662 (0.59,0.82)&0.365 (0.065,0.84)&0.062 (0.011,0.85)&$$&1.772 (0.761,0.82)&0.389 (0.065,0.84)&0.068 (0.014,0.82)\tabularnewline
PKm&1.548 (0.562,0.88)&0.338 (0.063,0.9)&0.057 (0.011,0.93)&$$&1.638 (0.71,0.88)&0.36 (0.061,0.91)&0.062 (0.013,0.9)\tabularnewline
\hline
\multicolumn{8}{c}{$K=7$, $\tau_{\alpha} =0.2$}\tabularnewline
DP&1.046 (0.416,1)&0.228 (0.047,1)&0.049 (0.009,1)&$$&1.027 (0.537,1)&0.241 (0.058,1)&0.052 (0.01,1)\tabularnewline
P0&2.329 (0.673,0.45)&0.458 (0.066,0.5)&0.103 (0.013,0.48)&$$&2.37 (0.709,0.43)&0.471 (0.069,0.51)&0.106 (0.013,0.49)\tabularnewline
P1&1.849 (0.623,0.57)&0.379 (0.058,0.6)&0.087 (0.012,0.56)&$$&1.81 (0.693,0.57)&0.385 (0.065,0.63)&0.089 (0.015,0.58)\tabularnewline
P1m&1.605 (0.555,0.65)&0.332 (0.056,0.69)&0.073 (0.012,0.67)&$$&1.614 (0.591,0.64)&0.343 (0.058,0.7)&0.076 (0.011,0.68)\tabularnewline
PK&1.823 (0.571,0.57)&0.377 (0.054,0.6)&0.086 (0.011,0.57)&$$&1.791 (0.692,0.57)&0.386 (0.063,0.62)&0.089 (0.014,0.58)\tabularnewline
PKm&1.583 (0.528,0.66)&0.329 (0.052,0.69)&0.073 (0.01,0.67)&$$&1.612 (0.596,0.64)&0.343 (0.059,0.7)&0.076 (0.011,0.68)\tabularnewline
\hline
\multicolumn{8}{c}{$K=7$, $\tau_{\alpha} =0.5$}\tabularnewline
DP&1.336 (0.498,1)&0.304 (0.054,1)&0.052 (0.009,1)&$$&1.438 (0.555,1)&0.321 (0.064,1)&0.055 (0.011,1)\tabularnewline
P0&2.268 (0.689,0.59)&0.464 (0.07,0.66)&0.08 (0.012,0.65)&$$&2.395 (0.743,0.6)&0.482 (0.073,0.67)&0.082 (0.012,0.67)\tabularnewline
P1&1.793 (0.6,0.75)&0.388 (0.058,0.78)&0.067 (0.011,0.78)&$$&1.89 (0.741,0.76)&0.398 (0.071,0.81)&0.069 (0.013,0.8)\tabularnewline
P1m&1.665 (0.591,0.8)&0.36 (0.059,0.84)&0.061 (0.011,0.85)&$$&1.767 (0.693,0.81)&0.373 (0.067,0.86)&0.063 (0.012,0.87)\tabularnewline
PK&1.801 (0.614,0.74)&0.388 (0.057,0.78)&0.067 (0.011,0.78)&$$&1.895 (0.76,0.76)&0.397 (0.068,0.81)&0.069 (0.013,0.8)\tabularnewline
PKm&1.663 (0.591,0.8)&0.36 (0.059,0.84)&0.061 (0.011,0.85)&$$&1.766 (0.698,0.81)&0.371 (0.063,0.87)&0.063 (0.012,0.87)\tabularnewline
\bottomrule
\end{tabular}\end{center}
\end{table}

\begin{table}
\scriptsize
\tabcolsep=3pt
\caption{Predictive accuracy of the proposed dynamic prediction (DP) method in comparison with the competing methods for training and testing datasets under Ex2 ($n_{train}=200$, $n_{test}=50$) with $10\%$ censoring for $D$. Values in each parenthesis are the empirical standard deviation and relative prediction accuracy.}
\label{table:ex2_n200_cr20}
\begin{center}
\begin{tabular}{llllrlll}
\toprule
\multirow{2}*{}& \multicolumn{3}{c}{In-sample}&&\multicolumn{3}{c}{Out-of-sample}\tabularnewline
\cmidrule(lr){2-4}\cmidrule(lr){6-8}
{Method}&{MSPE}&{QPE}&{IBS}&&{MSPE}&{QPE}&{IBS}\tabularnewline
\hline
\multicolumn{8}{c}{$K=3$, $\tau_{\alpha} =0.2$}\tabularnewline
DP&1.126 (0.303,1)&0.267 (0.031,1)&0.046 (0.005,1)&$$&1.129 (0.513,1)&0.276 (0.061,1)&0.047 (0.01,1)\tabularnewline
P0&2.046 (0.424,0.55)&0.438 (0.043,0.61)&0.075 (0.007,0.61)&$$&2.035 (0.57,0.55)&0.446 (0.061,0.62)&0.076 (0.011,0.62)\tabularnewline
P1&1.655 (0.374,0.68)&0.373 (0.035,0.72)&0.064 (0.006,0.72)&$$&1.645 (0.668,0.69)&0.381 (0.072,0.72)&0.065 (0.013,0.72)\tabularnewline
P1m&1.452 (0.351,0.78)&0.327 (0.035,0.82)&0.055 (0.006,0.84)&$$&1.433 (0.586,0.79)&0.334 (0.063,0.83)&0.056 (0.011,0.84)\tabularnewline
PK&1.648 (0.386,0.68)&0.373 (0.038,0.72)&0.064 (0.007,0.72)&$$&1.64 (0.685,0.69)&0.379 (0.072,0.73)&0.065 (0.014,0.72)\tabularnewline
PKm&1.449 (0.359,0.78)&0.328 (0.036,0.81)&0.055 (0.006,0.84)&$$&1.437 (0.593,0.79)&0.335 (0.063,0.82)&0.057 (0.011,0.82)\tabularnewline
\hline
\multicolumn{8}{c}{$K=3$, $\tau_{\alpha} =0.5$}\tabularnewline
DP&1.379 (0.337,1)&0.314 (0.038,1)&0.048 (0.005,1)&$$&1.417 (0.616,1)&0.322 (0.063,1)&0.049 (0.01,1)\tabularnewline
P0&2.145 (0.457,0.64)&0.448 (0.047,0.7)&0.069 (0.007,0.7)&$$&2.189 (0.774,0.65)&0.456 (0.069,0.71)&0.07 (0.012,0.7)\tabularnewline
P1&1.74 (0.416,0.79)&0.376 (0.041,0.84)&0.058 (0.007,0.83)&$$&1.779 (0.845,0.8)&0.384 (0.072,0.84)&0.058 (0.014,0.84)\tabularnewline
P1m&1.612 (0.397,0.86)&0.349 (0.04,0.9)&0.053 (0.006,0.91)&$$&1.645 (0.77,0.86)&0.356 (0.065,0.9)&0.053 (0.012,0.92)\tabularnewline
PK&1.735 (0.413,0.79)&0.378 (0.04,0.83)&0.058 (0.007,0.83)&$$&1.776 (0.816,0.8)&0.384 (0.071,0.84)&0.058 (0.013,0.84)\tabularnewline
PKm&1.612 (0.397,0.86)&0.35 (0.04,0.9)&0.053 (0.006,0.91)&$$&1.647 (0.767,0.86)&0.357 (0.064,0.9)&0.053 (0.012,0.92)\tabularnewline
\hline
\multicolumn{8}{c}{$K=7$, $\tau_{\alpha} =0.2$}\tabularnewline
DP&1.001 (0.32,1)&0.24 (0.034,1)&0.039 (0.005,1)&$$&1.043 (0.596,1)&0.245 (0.059,1)&0.039 (0.009,1)\tabularnewline
P0&2.3 (0.491,0.44)&0.471 (0.045,0.51)&0.079 (0.007,0.49)&$$&2.376 (0.703,0.44)&0.481 (0.064,0.51)&0.08 (0.009,0.49)\tabularnewline
P1&1.828 (0.445,0.55)&0.39 (0.039,0.62)&0.067 (0.007,0.58)&$$&1.91 (0.892,0.55)&0.395 (0.074,0.62)&0.068 (0.013,0.57)\tabularnewline
P1m&1.583 (0.413,0.63)&0.344 (0.039,0.7)&0.057 (0.006,0.68)&$$&1.649 (0.702,0.63)&0.351 (0.062,0.7)&0.058 (0.01,0.67)\tabularnewline
PK&1.836 (0.439,0.55)&0.392 (0.037,0.61)&0.067 (0.006,0.58)&$$&1.869 (0.874,0.56)&0.395 (0.07,0.62)&0.067 (0.012,0.58)\tabularnewline
PKm&1.594 (0.402,0.63)&0.345 (0.038,0.7)&0.058 (0.006,0.67)&$$&1.657 (0.704,0.63)&0.354 (0.061,0.69)&0.058 (0.009,0.67)\tabularnewline
\hline
\multicolumn{8}{c}{$K=7$, $\tau_{\alpha} =0.5$}\tabularnewline
DP&1.262 (0.295,1)&0.3 (0.032,1)&0.045 (0.005,1)&$$&1.296 (0.576,1)&0.313 (0.063,1)&0.046 (0.009,1)\tabularnewline
P0&2.17 (0.411,0.58)&0.459 (0.04,0.65)&0.07 (0.006,0.64)&$$&2.176 (0.668,0.6)&0.471 (0.066,0.66)&0.071 (0.01,0.65)\tabularnewline
P1&1.734 (0.379,0.73)&0.383 (0.035,0.78)&0.058 (0.006,0.78)&$$&1.734 (0.781,0.75)&0.395 (0.071,0.79)&0.06 (0.012,0.77)\tabularnewline
P1m&1.6 (0.365,0.79)&0.357 (0.035,0.84)&0.054 (0.006,0.83)&$$&1.593 (0.667,0.81)&0.368 (0.064,0.85)&0.054 (0.01,0.85)\tabularnewline
PK&1.743 (0.384,0.72)&0.384 (0.037,0.78)&0.059 (0.006,0.76)&$$&1.757 (0.777,0.74)&0.395 (0.07,0.79)&0.06 (0.012,0.77)\tabularnewline
PKm&1.603 (0.365,0.79)&0.358 (0.036,0.84)&0.054 (0.006,0.83)&$$&1.608 (0.681,0.81)&0.369 (0.064,0.85)&0.055 (0.011,0.84)\tabularnewline
\bottomrule
\end{tabular}\end{center}
\end{table}

\begin{table}
\scriptsize
\tabcolsep=3pt
\caption{Predictive accuracy of the proposed dynamic prediction (DP) method in comparison with the competing methods for training and testing datasets under Ex2 ($n_{train}=200$, $n_{test}=50$) with $35\%$ censoring for $D$. Values in each parenthesis are the empirical standard deviation and relative prediction accuracy.}
\label{table:ex2_n200_cr5}
\begin{center}
\begin{tabular}{llllrlll}
\toprule
\multirow{2}*{}& \multicolumn{3}{c}{In-sample}&&\multicolumn{3}{c}{Out-of-sample}\tabularnewline
\cmidrule(lr){2-4}\cmidrule(lr){6-8}
{Method}&{MSPE}&{QPE}&{IBS}&&{MSPE}&{QPE}&{IBS}\tabularnewline
\hline
\multicolumn{8}{c}{$K=3$, $\tau_{\alpha} =0.2$}\tabularnewline
DP&1.608 (0.375,1)&0.28 (0.037,1)&0.082 (0.012,1)&$$&1.663 (0.863,1)&0.287 (0.07,1)&0.08 (0.02,1)\tabularnewline
P0&2.806 (0.513,0.57)&0.447 (0.051,0.63)&0.156 (0.021,0.53)&$$&2.886 (1.12,0.58)&0.455 (0.074,0.63)&0.146 (0.026,0.55)\tabularnewline
P1&2.49 (0.495,0.65)&0.368 (0.041,0.76)&0.119 (0.017,0.69)&$$&2.589 (1.147,0.64)&0.378 (0.074,0.76)&0.116 (0.03,0.69)\tabularnewline
P1m&2.365 (0.481,0.68)&0.336 (0.041,0.83)&0.104 (0.016,0.79)&$$&2.451 (1.121,0.68)&0.346 (0.069,0.83)&0.101 (0.024,0.79)\tabularnewline
PK&2.506 (0.513,0.64)&0.37 (0.043,0.76)&0.121 (0.016,0.68)&$$&2.555 (1.155,0.65)&0.371 (0.071,0.77)&0.114 (0.027,0.7)\tabularnewline
PKm&2.374 (0.494,0.68)&0.337 (0.043,0.83)&0.105 (0.015,0.78)&$$&2.435 (1.129,0.68)&0.342 (0.069,0.84)&0.1 (0.023,0.8)\tabularnewline
\hline
\multicolumn{8}{c}{$K=3$, $\tau_{\alpha} =0.5$}\tabularnewline
DP&1.739 (0.402,1)&0.312 (0.038,1)&0.1 (0.014,1)&$$&1.737 (0.853,1)&0.315 (0.085,1)&0.094 (0.022,1)\tabularnewline
P0&2.886 (0.564,0.6)&0.454 (0.048,0.69)&0.161 (0.022,0.62)&$$&2.932 (1.085,0.59)&0.462 (0.076,0.68)&0.152 (0.026,0.62)\tabularnewline
P1&2.587 (0.56,0.67)&0.371 (0.041,0.84)&0.121 (0.017,0.83)&$$&2.604 (1.128,0.67)&0.372 (0.085,0.85)&0.114 (0.028,0.82)\tabularnewline
P1m&2.496 (0.552,0.7)&0.349 (0.04,0.89)&0.111 (0.017,0.9)&$$&2.519 (1.1,0.69)&0.351 (0.08,0.9)&0.105 (0.025,0.9)\tabularnewline
PK&2.588 (0.56,0.67)&0.372 (0.041,0.84)&0.121 (0.017,0.83)&$$&2.625 (1.15,0.66)&0.376 (0.09,0.84)&0.117 (0.029,0.8)\tabularnewline
PKm&2.496 (0.553,0.7)&0.349 (0.04,0.89)&0.111 (0.016,0.9)&$$&2.53 (1.103,0.69)&0.354 (0.081,0.89)&0.106 (0.025,0.89)\tabularnewline
\hline
\multicolumn{8}{c}{$K=7$, $\tau_{\alpha} =0.2$}\tabularnewline
DP&1.502 (0.448,1)&0.267 (0.037,1)&0.063 (0.01,1)&$$&1.547 (0.854,1)&0.275 (0.073,1)&0.062 (0.015,1)\tabularnewline
P0&2.747 (0.583,0.55)&0.482 (0.052,0.55)&0.163 (0.021,0.39)&$$&2.784 (1.017,0.56)&0.492 (0.068,0.56)&0.153 (0.022,0.41)\tabularnewline
P1&2.414 (0.576,0.62)&0.394 (0.043,0.68)&0.126 (0.015,0.5)&$$&2.497 (1.119,0.62)&0.409 (0.077,0.67)&0.123 (0.025,0.5)\tabularnewline
P1m&2.25 (0.555,0.67)&0.361 (0.043,0.74)&0.11 (0.015,0.57)&$$&2.29 (1.043,0.68)&0.373 (0.07,0.74)&0.105 (0.02,0.59)\tabularnewline
PK&2.41 (0.572,0.62)&0.394 (0.042,0.68)&0.126 (0.016,0.5)&$$&2.45 (1.093,0.63)&0.403 (0.076,0.68)&0.121 (0.026,0.51)\tabularnewline
PKm&2.249 (0.551,0.67)&0.362 (0.042,0.74)&0.11 (0.016,0.57)&$$&2.281 (1.032,0.68)&0.37 (0.071,0.74)&0.104 (0.02,0.6)\tabularnewline
\hline
\multicolumn{8}{c}{$K=7$, $\tau_{\alpha} =0.5$}\tabularnewline
DP&1.69 (0.423,1)&0.3 (0.038,1)&0.086 (0.013,1)&$$&1.633 (0.871,1)&0.303 (0.067,1)&0.083 (0.021,1)\tabularnewline
P0&2.902 (0.591,0.58)&0.47 (0.054,0.64)&0.159 (0.025,0.54)&$$&2.835 (1.167,0.58)&0.473 (0.071,0.64)&0.15 (0.026,0.55)\tabularnewline
P1&2.576 (0.569,0.66)&0.385 (0.044,0.78)&0.12 (0.018,0.72)&$$&2.505 (1.19,0.65)&0.387 (0.069,0.78)&0.115 (0.026,0.72)\tabularnewline
P1m&2.463 (0.555,0.69)&0.362 (0.043,0.83)&0.11 (0.018,0.78)&$$&2.387 (1.166,0.68)&0.365 (0.065,0.83)&0.106 (0.023,0.78)\tabularnewline
PK&2.578 (0.571,0.66)&0.384 (0.043,0.78)&0.121 (0.018,0.71)&$$&2.518 (1.199,0.65)&0.39 (0.07,0.78)&0.116 (0.026,0.72)\tabularnewline
PKm&2.464 (0.558,0.69)&0.362 (0.043,0.83)&0.111 (0.018,0.77)&$$&2.393 (1.169,0.68)&0.366 (0.064,0.83)&0.105 (0.024,0.79)\tabularnewline
\bottomrule
\end{tabular}\end{center}
\end{table}

\begin{table}[ht]
\scriptsize
\tabcolsep=3pt
\caption{Predictive accuracy of the proposed dynamic prediction (DP) method in comparison with the competing methods for training and testing 
datasets under Ex1 and Ex2 ( $K=3$, $\tau_{\alpha} =0.5$) with $10\%$ censoring for $D$, based on the training set of size 2000. Values in each parenthesis are the empirical standard deviation and relative prediction accuracy.} 
\label{table:ex12_n2000_cr20}
\begin{center}
\begin{tabular}{llllrlll}
\toprule
\multirow{2}*{}& \multicolumn{3}{c}{In-sample}&&\multicolumn{3}{c}{Out-of-sample}\tabularnewline
\cmidrule(lr){2-4}\cmidrule(lr){6-8}
{Method}&{MSPE}&{QPE}&{IBS}&&{MSPE}&{QPE}&{IBS}\tabularnewline
\hline
\multicolumn{8}{c}{Ex1,$n_{\operatorname{train}}=2000$, $n_{\operatorname{test}}=50$}\tabularnewline
DP&0.788 (0.081,1)&0.188 (0.009,1)&0.027 (0.001,1)&$$&0.783 (0.453,1)&0.187 (0.049,1)&0.026 (0.006,1)\tabularnewline
P0&2.052 (0.149,0.38)&0.435 (0.016,0.43)&0.061 (0.002,0.44)&$$&2.072 (0.564,0.38)&0.436 (0.05,0.43)&0.06 (0.008,0.43)\tabularnewline
P1&0.928 (0.106,0.85)&0.214 (0.01,0.88)&0.028 (0.001,0.96)&$$&0.923 (0.561,0.85)&0.211 (0.052,0.89)&0.027 (0.008,0.96)\tabularnewline
P1m&0.93 (0.106,0.85)&0.213 (0.01,0.88)&0.028 (0.001,0.96)&$$&0.925 (0.562,0.85)&0.21 (0.052,0.89)&0.027 (0.008,0.96)\tabularnewline
PK&2.371 (0.17,0.33)&0.5 (0.016,0.38)&0.069 (0.002,0.39)&$$&2.392 (0.83,0.33)&0.497 (0.072,0.38)&0.069 (0.012,0.38)\tabularnewline
PKm&1.849 (0.142,0.43)&0.397 (0.015,0.47)&0.056 (0.002,0.48)&$$&1.869 (0.588,0.42)&0.397 (0.052,0.47)&0.055 (0.008,0.47)\tabularnewline
\hline
\multicolumn{8}{c}{Ex2,$n_{\operatorname{train}}=2000$, $n_{\operatorname{test}}=50$}\tabularnewline
DP&1.342 (0.114,1)&0.315 (0.012,1)&0.044 (0.002,1)&$$&1.439 (0.694,1)&0.321 (0.075,1)&0.044 (0.01,1)\tabularnewline
P0&2.109 (0.15,0.64)&0.45 (0.014,0.7)&0.063 (0.002,0.7)&$$&2.223 (0.754,0.65)&0.458 (0.068,0.7)&0.064 (0.011,0.69)\tabularnewline
P1&1.715 (0.14,0.78)&0.38 (0.012,0.83)&0.053 (0.002,0.83)&$$&1.833 (0.862,0.79)&0.384 (0.083,0.84)&0.053 (0.013,0.83)\tabularnewline
P1m&1.587 (0.135,0.85)&0.352 (0.012,0.89)&0.049 (0.002,0.9)&$$&1.691 (0.803,0.85)&0.355 (0.076,0.9)&0.049 (0.012,0.9)\tabularnewline
PK&1.711 (0.138,0.78)&0.379 (0.013,0.83)&0.053 (0.002,0.83)&$$&1.832 (0.871,0.79)&0.387 (0.084,0.83)&0.053 (0.013,0.83)\tabularnewline
PKm&1.585 (0.134,0.85)&0.351 (0.013,0.9)&0.049 (0.002,0.9)&$$&1.694 (0.804,0.85)&0.357 (0.076,0.9)&0.049 (0.012,0.9)\tabularnewline
\hline
\end{tabular}\end{center}
\end{table}

\begin{table}[ht]
\scriptsize
\tabcolsep=3pt
\caption{Predictive accuracy of the proposed dynamic prediction (DP) method and its variations for training and testing 
datasets under Ex1 and Ex2 ( $n_{\operatorname{train}}=100$, $n_{\operatorname{test}}=50$, Frank copula structure) with $35\%$ censoring for $D$. \emph{DP} uses 
all intermediate events, whereas \emph{DP2:6} and \emph{DP3:5} use only intermediate events 2-6 and 3-5, respectively. 
Values in each parenthesis are the empirical standard deviation and relative prediction accuracy. 
}
\label{table:ex12_group}
\begin{center}
\begin{tabular}{lllllrllll}
\toprule
\multirow{2}*{}& \multicolumn{3}{c}{In-sample}&&\multicolumn{3}{c}{Out-of-sample}\tabularnewline
\cmidrule(lr){2-5}\cmidrule(lr){7-10}
{Method}&{MSPE}&{QPE}&{IBS}&{AUC}&&{MSPE}&{QPE}&{IBS}&{AUC}\tabularnewline
\hline
\multicolumn{10}{c}{Ex1, $K=7$, $\tau_{\alpha}=0.2$}\tabularnewline
DP&1.335 (0.539)&0.231 (0.046)&0.05 (0.012)&0.967 (0.034)&&1.361 (0.728)&0.235 (0.071)&0.052 (0.014)&0.959 (0.05)\tabularnewline
DP2:6&1.461 (0.548)&0.266 (0.048)&0.065 (0.014)&0.954 (0.041)&&1.49 (0.767)&0.272 (0.076)&0.067 (0.016)&0.946 (0.063)\tabularnewline
DP3:5&1.635 (0.572)&0.307 (0.051)&0.084 (0.015)&0.928 (0.061)&&1.676 (0.784)&0.314 (0.078)&0.088 (0.02)&0.928 (0.071)\tabularnewline
P4&2.11 (0.66)&0.361 (0.056)&0.114 (0.023)&0.949 (0.052)&&2.166 (0.829)&0.372 (0.074)&0.117 (0.021)&0.952 (0.059)\tabularnewline
\hline
\multicolumn{10}{c}{Ex1, $K=7$, $\tau_{\alpha}=0.5$}\tabularnewline
DP&1.334 (0.528)&0.212 (0.045)&0.056 (0.014)&0.974 (0.024)&&1.394 (0.84)&0.226 (0.064)&0.06 (0.015)&0.974 (0.041)\tabularnewline
DP2:6&1.507 (0.541)&0.265 (0.049)&0.077 (0.017)&0.949 (0.041)&&1.579 (0.858)&0.281 (0.069)&0.08 (0.018)&0.941 (0.059)\tabularnewline
DP3:5&1.723 (0.572)&0.318 (0.057)&0.097 (0.02)&0.912 (0.064)&&1.791 (0.9)&0.333 (0.074)&0.1 (0.02)&0.906 (0.094)\tabularnewline
P4&2.204 (0.656)&0.354 (0.056)&0.114 (0.021)&0.925 (0.055)&&2.245 (1.005)&0.366 (0.071)&0.116 (0.021)&0.932 (0.073)\tabularnewline
\hline
\multicolumn{10}{c}{Ex2, $K=7$, $\tau_{\alpha}=0.2$}\tabularnewline
DP&1.486 (0.492)&0.262 (0.049)&0.067 (0.014)&0.942 (0.044)&&1.45 (0.693)&0.266 (0.065)&0.07 (0.016)&0.931 (0.082)\tabularnewline
DP2:6&1.546 (0.5)&0.274 (0.05)&0.073 (0.014)&0.942 (0.042)&&1.51 (0.713)&0.279 (0.066)&0.077 (0.017)&0.928 (0.076)\tabularnewline
DP3:5&1.672 (0.522)&0.304 (0.052)&0.086 (0.016)&0.932 (0.045)&&1.639 (0.733)&0.308 (0.069)&0.091 (0.02)&0.916 (0.087)\tabularnewline
P4&2.217 (0.642)&0.356 (0.062)&0.113 (0.023)&0.94 (0.048)&&2.181 (0.865)&0.364 (0.068)&0.117 (0.024)&0.932 (0.09)\tabularnewline
\hline
\multicolumn{10}{c}{Ex2, $K=7$, $\tau_{\alpha}=0.5$}\tabularnewline
DP&1.659 (0.642)&0.296 (0.051)&0.092 (0.017)&0.909 (0.061)&&1.77 (0.813)&0.308 (0.069)&0.097 (0.019)&0.892 (0.097)\tabularnewline
DP2:6&1.704 (0.653)&0.304 (0.053)&0.095 (0.017)&0.908 (0.06)&&1.827 (0.816)&0.318 (0.069)&0.1 (0.019)&0.89 (0.107)\tabularnewline
DP3:5&1.786 (0.665)&0.319 (0.053)&0.102 (0.017)&0.902 (0.061)&&1.923 (0.832)&0.336 (0.074)&0.106 (0.021)&0.877 (0.122)\tabularnewline
P4&2.407 (0.828)&0.359 (0.058)&0.119 (0.023)&0.896 (0.07)&&2.52 (1.017)&0.377 (0.075)&0.126 (0.023)&0.888 (0.1)\tabularnewline
\hline
\end{tabular}\end{center}
\end{table}

\begin{landscape}
\begin{table}[t]
\tiny
\tabcolsep=1pt
\caption{Comparison of different prediction methods under Ex1–2 with $10\%$ censoring for $D$. Reported values are the means of the performance metrics, with empirical standard deviations shown in parentheses.\emph{Apparent}: error on full sample; \emph{3-fold CV}: 3-fold cross-validation; \emph{random CV$_{1/3}$} and \emph{random CV$_{1/5}$}: random cross-validation with 1/3 or 1/5 test proportion.}
\label{table:ex12_split}
\begin{center}
\begin{tabular}{lllllllllllllll}
\toprule
\multirow{2}*{}& \multicolumn{4}{l}{MSPE}&&\multicolumn{4}{l}{IBS}&&\multicolumn{4}{l}{AUC}\tabularnewline
\cmidrule(lr){2-5}\cmidrule(lr){6-10}\cmidrule(lr){11-15}
Method&Apparent&3-fold CV&random CV$_{1/3}$&random CV$_{1/5}$&&Apparent&3-fold CV&random CV$_{1/3}$&random CV$_{1/5}$&&Apparent&3-fold CV&random CV$_{1/3}$&random CV$_{1/5}$\tabularnewline
\hline
\multicolumn{15}{c}{Ex1, $n=150$, $K=3$, $\tau_{\alpha} =0.2$}\tabularnewline
DP&0.794 (0.264)&0.86 (0.278)&0.824 (0.26)&0.812 (0.277)&&0.037 (0.005)&0.04 (0.005)&0.039 (0.005)&0.038 (0.005)&&0.964 (0.017)&0.962 (0.017)&0.962 (0.017)&0.963 (0.018)\tabularnewline
P0&1.996 (0.446)&2.068 (0.46)&2.041 (0.449)&2.027 (0.454)&&0.084 (0.009)&0.086 (0.009)&0.086 (0.009)&0.084 (0.009)&&0.955 (0.022)&0.955 (0.022)&0.955 (0.022)&0.955 (0.024)\tabularnewline
P1&0.884 (0.328)&0.937 (0.335)&0.882 (0.31)&0.874 (0.324)&&0.039 (0.005)&0.042 (0.005)&0.041 (0.005)&0.04 (0.005)&&0.959 (0.02)&0.959 (0.021)&0.959 (0.021)&0.959 (0.024)\tabularnewline
P1m&0.86 (0.326)&0.914 (0.334)&0.86 (0.307)&0.85 (0.321)&&0.037 (0.005)&0.04 (0.005)&0.039 (0.005)&0.038 (0.005)&&0.962 (0.019)&0.961 (0.019)&0.961 (0.02)&0.961 (0.022)\tabularnewline
PK&2.273 (0.504)&2.348 (0.519)&2.298 (0.504)&2.29 (0.523)&&0.096 (0.01)&0.099 (0.01)&0.098 (0.01)&0.096 (0.011)&&0.717 (0.079)&0.715 (0.079)&0.717 (0.082)&0.718 (0.083)\tabularnewline
PKm&1.79 (0.436)&1.863 (0.451)&1.83 (0.437)&1.815 (0.442)&&0.076 (0.009)&0.079 (0.009)&0.079 (0.009)&0.077 (0.009)&&0.948 (0.024)&0.946 (0.025)&0.946 (0.025)&0.947 (0.026)\tabularnewline
\hline
\multicolumn{15}{c}{Ex1, $n=150$, $K=3$, $\tau_{\alpha} =0.5$}\tabularnewline
DP&0.814 (0.256)&0.889 (0.266)&0.856 (0.25)&0.845 (0.27)&&0.032 (0.005)&0.035 (0.005)&0.035 (0.005)&0.034 (0.005)&&0.974 (0.014)&0.972 (0.015)&0.972 (0.014)&0.973 (0.015)\tabularnewline
P0&2.078 (0.448)&2.152 (0.459)&2.123 (0.452)&2.104 (0.464)&&0.075 (0.008)&0.077 (0.008)&0.077 (0.008)&0.075 (0.008)&&0.953 (0.021)&0.953 (0.022)&0.954 (0.021)&0.954 (0.022)\tabularnewline
P1&0.923 (0.325)&0.979 (0.33)&0.92 (0.306)&0.913 (0.325)&&0.033 (0.005)&0.036 (0.005)&0.035 (0.005)&0.034 (0.005)&&0.962 (0.021)&0.962 (0.021)&0.963 (0.02)&0.963 (0.021)\tabularnewline
P1m&0.927 (0.327)&0.984 (0.332)&0.925 (0.309)&0.918 (0.327)&&0.033 (0.005)&0.036 (0.005)&0.035 (0.005)&0.034 (0.005)&&0.962 (0.02)&0.961 (0.021)&0.962 (0.02)&0.962 (0.021)\tabularnewline
PK&2.418 (0.515)&2.496 (0.526)&2.437 (0.515)&2.422 (0.532)&&0.086 (0.009)&0.089 (0.01)&0.087 (0.01)&0.086 (0.01)&&0.693 (0.07)&0.69 (0.072)&0.692 (0.071)&0.692 (0.076)\tabularnewline
PKm&1.88 (0.434)&1.955 (0.444)&1.92 (0.434)&1.903 (0.449)&&0.069 (0.008)&0.071 (0.009)&0.071 (0.009)&0.069 (0.008)&&0.94 (0.025)&0.938 (0.026)&0.939 (0.026)&0.94 (0.027)\tabularnewline
\hline
\multicolumn{15}{c}{Ex2, $n=150$, $K=3$, $\tau_{\alpha} =0.2$}\tabularnewline
DP&1.156 (0.348)&1.228 (0.362)&1.201 (0.357)&1.175 (0.343)&&0.053 (0.007)&0.056 (0.007)&0.055 (0.007)&0.054 (0.007)&&0.921 (0.033)&0.918 (0.034)&0.918 (0.033)&0.919 (0.036)\tabularnewline
P0&2.056 (0.488)&2.129 (0.501)&2.104 (0.498)&2.072 (0.482)&&0.087 (0.009)&0.089 (0.009)&0.089 (0.009)&0.086 (0.009)&&0.928 (0.031)&0.927 (0.031)&0.927 (0.031)&0.928 (0.033)\tabularnewline
P1&1.678 (0.446)&1.742 (0.457)&1.694 (0.449)&1.666 (0.429)&&0.074 (0.009)&0.077 (0.009)&0.076 (0.009)&0.074 (0.008)&&0.852 (0.048)&0.851 (0.047)&0.852 (0.047)&0.851 (0.054)\tabularnewline
P1m&1.464 (0.414)&1.529 (0.425)&1.49 (0.418)&1.46 (0.406)&&0.064 (0.008)&0.066 (0.008)&0.066 (0.008)&0.064 (0.008)&&0.916 (0.035)&0.914 (0.036)&0.914 (0.035)&0.914 (0.039)\tabularnewline
PK&1.681 (0.464)&1.747 (0.476)&1.699 (0.471)&1.674 (0.457)&&0.074 (0.009)&0.077 (0.009)&0.076 (0.009)&0.074 (0.009)&&0.85 (0.055)&0.848 (0.056)&0.849 (0.055)&0.85 (0.056)\tabularnewline
PKm&1.465 (0.424)&1.53 (0.436)&1.492 (0.431)&1.463 (0.413)&&0.064 (0.008)&0.066 (0.008)&0.066 (0.008)&0.064 (0.008)&&0.913 (0.037)&0.911 (0.037)&0.912 (0.037)&0.912 (0.038)\tabularnewline
\hline
\multicolumn{15}{c}{Ex2, $n=150$, $K=3$, $\tau_{\alpha} =0.5$}\tabularnewline
DP&1.342 (0.371)&1.42 (0.385)&1.395 (0.371)&1.376 (0.393)&&0.053 (0.006)&0.055 (0.006)&0.055 (0.006)&0.053 (0.006)&&0.911 (0.036)&0.906 (0.038)&0.907 (0.037)&0.909 (0.039)\tabularnewline
P0&2.074 (0.493)&2.141 (0.502)&2.105 (0.485)&2.081 (0.501)&&0.076 (0.008)&0.078 (0.008)&0.078 (0.008)&0.076 (0.008)&&0.914 (0.035)&0.913 (0.035)&0.914 (0.034)&0.914 (0.036)\tabularnewline
P1&1.682 (0.457)&1.743 (0.464)&1.695 (0.444)&1.66 (0.466)&&0.063 (0.007)&0.066 (0.007)&0.065 (0.007)&0.063 (0.007)&&0.866 (0.043)&0.865 (0.044)&0.865 (0.043)&0.867 (0.047)\tabularnewline
P1m&1.559 (0.441)&1.621 (0.448)&1.574 (0.427)&1.547 (0.451)&&0.058 (0.007)&0.06 (0.007)&0.06 (0.007)&0.058 (0.007)&&0.899 (0.037)&0.898 (0.038)&0.899 (0.037)&0.9 (0.04)\tabularnewline
PK&1.693 (0.475)&1.754 (0.484)&1.701 (0.46)&1.677 (0.487)&&0.064 (0.008)&0.066 (0.008)&0.065 (0.008)&0.063 (0.008)&&0.865 (0.044)&0.863 (0.045)&0.864 (0.045)&0.865 (0.047)\tabularnewline
PKm&1.565 (0.453)&1.627 (0.46)&1.579 (0.436)&1.556 (0.463)&&0.058 (0.007)&0.06 (0.007)&0.06 (0.007)&0.058 (0.007)&&0.899 (0.038)&0.898 (0.039)&0.899 (0.039)&0.9 (0.04)\tabularnewline
\hline
\end{tabular}\end{center}
\end{table}
\end{landscape}

\begin{table}[t]
\scriptsize
\tabcolsep=5pt
\begin{center}
\caption{Coverage probability (CP) and median interval width (MID) of individual-level prediction intervals under training and testing datasets ($n_{test}=50$) of Examples Ex1-2.} 
\label{tab:predictionInterval}
\begin{tabular}{lllllllllllll}
\toprule
\multicolumn{1}{l}{Example}&\multicolumn{1}{c}{$K$}&\multicolumn{1}{c}{$\tau_{\alpha}$}&\multicolumn{1}{c}{$n_{train}$}&\multicolumn{4}{c}{$10\%$ censoring for $D$}&\multicolumn{4}{c}{no censoring for $D$}\tabularnewline
\cmidrule(lr){5-8}\cmidrule(lr){9-12}
&&&&\multicolumn{2}{c}{In-sample}&\multicolumn{2}{c}{Out-of-sample}&\multicolumn{2}{c}{In-sample}&\multicolumn{2}{c}{Out-of-sample}\tabularnewline
\cmidrule(lr){5-6}\cmidrule(lr){7-8}\cmidrule(lr){9-10}\cmidrule(lr){11-12}
&&&&CP&MID&CP&MID&CP&MID&CP&MID\tabularnewline
\hline
Ex1&3&0.2&100&0.937&1.402&0.884&1.456&0.949&1.435&0.897&1.446\tabularnewline
&&&200&0.934&1.469&0.905&1.533&0.942&1.442&0.907&1.495\tabularnewline
&&&400&0.935&1.524&0.921&1.613&0.938&1.473&0.92&1.516\tabularnewline
&&0.5&100&0.923&1.332&0.872&1.399&0.933&1.265&0.868&1.298\tabularnewline
&&&200&0.925&1.383&0.892&1.423&0.932&1.362&0.898&1.379\tabularnewline
&&&400&0.923&1.394&0.906&1.484&0.926&1.359&0.911&1.406\tabularnewline
&7&0.2&100&0.924&1.078&0.821&1.071&0.942&1.019&0.833&1.07\tabularnewline
&&&200&0.921&1.085&0.863&1.153&0.933&1.054&0.872&1.082\tabularnewline
&&&400&0.917&1.141&0.887&1.173&0.921&1.088&0.892&1.116\tabularnewline
&&0.5&100&0.89&0.989&0.789&1.028&0.903&0.975&0.794&1.023\tabularnewline
&&&200&0.885&1.028&0.82&1.095&0.889&0.993&0.83&1.01\tabularnewline
&&&400&0.88&1.045&0.851&1.091&0.88&1.001&0.851&1.06\tabularnewline
\hline
Ex2&3&0.2&100&0.951&2.868&0.891&2.874&0.966&2.984&0.905&2.985\tabularnewline
&&&200&0.942&3.211&0.916&3.25&0.949&3.249&0.911&3.224\tabularnewline
&&&400&0.937&3.331&0.921&3.406&0.941&3.328&0.923&3.331\tabularnewline
&&0.5&100&0.942&3.564&0.893&3.484&0.959&3.675&0.902&3.645\tabularnewline
&&&200&0.934&3.886&0.899&3.813&0.94&3.898&0.908&3.894\tabularnewline
&&&400&0.93&4.098&0.91&3.999&0.931&4.157&0.913&4.089\tabularnewline
&7&0.2&100&0.946&2.02&0.845&2.008&0.966&2.078&0.855&2.026\tabularnewline
&&&200&0.934&2.225&0.878&2.244&0.944&2.252&0.883&2.295\tabularnewline
&&&400&0.922&2.341&0.886&2.347&0.926&2.353&0.891&2.44\tabularnewline
&&0.5&100&0.928&2.863&0.823&2.844&0.939&2.946&0.831&2.941\tabularnewline
&&&200&0.911&3.275&0.854&3.245&0.911&3.285&0.847&3.284\tabularnewline
&&&400&0.894&3.42&0.865&3.435&0.896&3.423&0.868&3.351\tabularnewline
\bottomrule
\end{tabular}\end{center}
\end{table}

\begin{table}
\scriptsize
\tabcolsep=3pt
\caption{Performance comparison under model mis-specification for the proposed dynamic prediction (DP) method in comparison with the competing methods. Training and testing datasets ($n_{train}=100$, $n_{test}=50$) for Ex1 and Ex2 are generated under the \textbf{\emph{Frank}} copula structure with $35\%$ censoring for $D$, but fitted using \textbf{\emph{Gumbel}} models. Values in each parenthesis are the empirical standard deviation and relative prediction accuracy.}
\label{table:mis gumbel}
\begin{center}
\begin{tabular}{llllrlll}
\toprule
\multirow{2}*{}& \multicolumn{3}{c}{In-sample}&&\multicolumn{3}{c}{Out-of-sample}\tabularnewline
\cmidrule(lr){2-4}\cmidrule(lr){6-8}
{Method}&{MSPE}&{QPE}&{IBS}&&{MSPE}&{QPE}&{IBS}\tabularnewline
\hline
\multicolumn{8}{c}{Ex1, $K=7$, $\tau_{\alpha} =0.2$}\tabularnewline
DP&1.338 (0.55,1)&0.232 (0.052,1)&0.048 (0.011,1)&$$&1.36 (0.761,1)&0.246 (0.071,1)&0.059 (0.017,1)\tabularnewline
P0&2.567 (0.68,0.52)&0.478 (0.073,0.49)&0.17 (0.028,0.28)&$$&2.573 (0.856,0.53)&0.489 (0.078,0.5)&0.168 (0.027,0.35)\tabularnewline
P1&1.658 (0.622,0.81)&0.258 (0.052,0.9)&0.064 (0.013,0.75)&$$&1.64 (0.842,0.83)&0.263 (0.064,0.94)&0.068 (0.016,0.87)\tabularnewline
P1m&1.634 (0.616,0.82)&0.246 (0.052,0.94)&0.058 (0.012,0.83)&$$&1.618 (0.84,0.84)&0.251 (0.066,0.98)&0.062 (0.016,0.95)\tabularnewline
PK&2.725 (0.727,0.49)&0.48 (0.069,0.48)&0.177 (0.028,0.27)&$$&2.73 (0.967,0.5)&0.485 (0.078,0.51)&0.176 (0.031,0.34)\tabularnewline
PKm&2.419 (0.679,0.55)&0.438 (0.07,0.53)&0.152 (0.028,0.32)&$$&2.414 (0.861,0.56)&0.45 (0.075,0.55)&0.151 (0.025,0.39)\tabularnewline
\hline
\multicolumn{8}{c}{Ex1, $K=7$, $\tau_{\alpha} =0.5$}\tabularnewline
DP&1.281 (0.579,1)&0.219 (0.048,1)&0.057 (0.011,1)&$$&1.386 (0.82,1)&0.228 (0.069,1)&0.068 (0.019,1)\tabularnewline
P0&2.572 (0.764,0.5)&0.46 (0.075,0.48)&0.167 (0.027,0.34)&$$&2.629 (0.969,0.53)&0.468 (0.082,0.49)&0.166 (0.025,0.41)\tabularnewline
P1&1.714 (0.716,0.75)&0.233 (0.047,0.94)&0.06 (0.012,0.95)&$$&1.752 (0.982,0.79)&0.242 (0.067,0.94)&0.065 (0.018,1.05)\tabularnewline
P1m&1.711 (0.716,0.75)&0.232 (0.048,0.94)&0.059 (0.013,0.97)&$$&1.751 (0.98,0.79)&0.239 (0.067,0.95)&0.064 (0.018,1.06)\tabularnewline
PK&2.77 (0.819,0.46)&0.483 (0.071,0.45)&0.179 (0.027,0.32)&$$&2.878 (1.13,0.48)&0.493 (0.086,0.46)&0.181 (0.031,0.38)\tabularnewline
PKm&2.431 (0.761,0.53)&0.421 (0.072,0.52)&0.15 (0.026,0.38)&$$&2.486 (0.981,0.56)&0.43 (0.079,0.53)&0.149 (0.026,0.46)\tabularnewline
\hline
\multicolumn{8}{c}{Ex2, $K=7$, $\tau_{\alpha} =0.2$}\tabularnewline
DP&1.461 (0.574,1)&0.256 (0.05,1)&0.067 (0.017,1)&$$&1.612 (0.874,1)&0.282 (0.072,1)&0.082 (0.022,1)\tabularnewline
P0&2.644 (0.745,0.55)&0.474 (0.068,0.54)&0.171 (0.028,0.39)&$$&2.688 (1.054,0.6)&0.486 (0.075,0.58)&0.171 (0.029,0.48)\tabularnewline
P1&2.325 (0.731,0.63)&0.388 (0.056,0.66)&0.132 (0.023,0.51)&$$&2.387 (1.115,0.68)&0.399 (0.074,0.71)&0.139 (0.026,0.59)\tabularnewline
P1m&2.163 (0.704,0.68)&0.353 (0.054,0.73)&0.115 (0.022,0.58)&$$&2.198 (1.055,0.73)&0.365 (0.07,0.77)&0.118 (0.023,0.69)\tabularnewline
PK&2.332 (0.756,0.63)&0.388 (0.056,0.66)&0.132 (0.02,0.51)&$$&2.397 (1.103,0.67)&0.401 (0.07,0.7)&0.138 (0.027,0.59)\tabularnewline
PKm&2.168 (0.717,0.67)&0.354 (0.056,0.72)&0.114 (0.02,0.59)&$$&2.201 (1.062,0.73)&0.366 (0.069,0.77)&0.118 (0.023,0.69)\tabularnewline
\hline
\multicolumn{8}{c}{Ex2, $K=7$, $\tau_{\alpha} =0.5$}\tabularnewline
DP&1.727 (0.602,1)&0.305 (0.056,1)&0.094 (0.019,1)&$$&1.884 (0.885,1)&0.336 (0.074,1)&0.11 (0.027,1)\tabularnewline
P0&2.775 (0.79,0.62)&0.457 (0.07,0.67)&0.166 (0.026,0.57)&$$&2.871 (1.059,0.66)&0.472 (0.08,0.71)&0.17 (0.03,0.65)\tabularnewline
P1&2.484 (0.774,0.7)&0.379 (0.06,0.8)&0.127 (0.021,0.74)&$$&2.598 (1.116,0.73)&0.394 (0.074,0.85)&0.133 (0.027,0.83)\tabularnewline
P1m&2.365 (0.756,0.73)&0.353 (0.059,0.86)&0.115 (0.021,0.82)&$$&2.465 (1.069,0.76)&0.369 (0.069,0.91)&0.12 (0.025,0.92)\tabularnewline
PK&2.485 (0.762,0.69)&0.378 (0.059,0.81)&0.126 (0.021,0.75)&$$&2.58 (1.116,0.73)&0.393 (0.077,0.85)&0.132 (0.027,0.83)\tabularnewline
PKm&2.367 (0.749,0.73)&0.353 (0.058,0.86)&0.114 (0.02,0.82)&$$&2.458 (1.077,0.77)&0.367 (0.071,0.92)&0.12 (0.025,0.92)\tabularnewline
\bottomrule
\end{tabular}\end{center}
\end{table}

\begin{table}
\scriptsize
\tabcolsep=3pt
\caption{Performance comparison under model mis-specification for the proposed dynamic prediction (DP) method in comparison with the competing methods. Training and testing datasets ($n_{train}=100$, $n_{test}=50$) for Ex1 and Ex2 are generated under the \textbf{\emph{Frank}} copula structure with $35\%$ censoring for $D$, but fitted using \textbf{\emph {Clayton}} models. Values in each parenthesis are the empirical standard deviation and relative prediction accuracy.}
\label{table: mis clayton}
\begin{center}
\begin{tabular}{llllrlll}
\toprule
\multirow{2}*{}& \multicolumn{3}{c}{In-sample}&&\multicolumn{3}{c}{Out-of-sample}\tabularnewline
\cmidrule(lr){2-4}\cmidrule(lr){6-8}
{Method}&{MSPE}&{QPE}&{IBS}&&{MSPE}&{QPE}&{IBS}\tabularnewline
\hline
\multicolumn{8}{c}{Ex1, $K=7$, $\tau_{\alpha} =0.2$}\tabularnewline
DP&1.395 (0.623,1)&0.239 (0.057,1)&0.057 (0.018,1)&$$&1.509 (0.885,1)&0.251 (0.072,1)&0.061 (0.02,1)\tabularnewline
P0&2.585 (0.739,0.54)&0.471 (0.073,0.51)&0.164 (0.028,0.35)&$$&2.738 (0.978,0.55)&0.486 (0.078,0.52)&0.166 (0.028,0.37)\tabularnewline
P1&1.731 (0.694,0.81)&0.255 (0.053,0.94)&0.066 (0.014,0.86)&$$&1.88 (0.993,0.8)&0.268 (0.066,0.94)&0.072 (0.02,0.85)\tabularnewline
P1m&1.709 (0.688,0.82)&0.249 (0.052,0.96)&0.062 (0.014,0.92)&$$&1.851 (0.983,0.82)&0.261 (0.065,0.96)&0.067 (0.018,0.91)\tabularnewline
PK&2.734 (0.793,0.51)&0.469 (0.067,0.51)&0.173 (0.027,0.33)&$$&2.975 (1.141,0.51)&0.489 (0.079,0.51)&0.178 (0.031,0.34)\tabularnewline
PKm&2.429 (0.729,0.57)&0.434 (0.07,0.55)&0.147 (0.026,0.39)&$$&2.601 (0.995,0.58)&0.452 (0.076,0.56)&0.151 (0.027,0.4)\tabularnewline
\hline
\multicolumn{8}{c}{Ex1, $K=7$, $\tau_{\alpha} =0.5$}\tabularnewline
DP&1.354 (0.526,1)&0.231 (0.049,1)&0.066 (0.019,1)&$$&1.427 (0.856,1)&0.246 (0.074,1)&0.07 (0.018,1)\tabularnewline
P0&2.647 (0.714,0.51)&0.465 (0.073,0.5)&0.169 (0.028,0.39)&$$&2.711 (1.011,0.53)&0.481 (0.093,0.51)&0.167 (0.028,0.42)\tabularnewline
P1&1.773 (0.665,0.76)&0.24 (0.047,0.96)&0.063 (0.014,1.05)&$$&1.819 (1.002,0.78)&0.254 (0.07,0.97)&0.069 (0.018,1.01)\tabularnewline
P1m&1.768 (0.665,0.77)&0.242 (0.048,0.95)&0.063 (0.015,1.05)&$$&1.814 (1.001,0.79)&0.256 (0.07,0.96)&0.07 (0.018,1)\tabularnewline
PK&2.862 (0.768,0.47)&0.486 (0.071,0.48)&0.183 (0.029,0.36)&$$&2.959 (1.131,0.48)&0.501 (0.088,0.49)&0.184 (0.03,0.38)\tabularnewline
PKm&2.523 (0.716,0.54)&0.436 (0.071,0.53)&0.155 (0.028,0.43)&$$&2.593 (1.024,0.55)&0.452 (0.087,0.54)&0.155 (0.028,0.45)\tabularnewline
\hline
\multicolumn{8}{c}{Ex2, $K=7$, $\tau_{\alpha} =0.2$}\tabularnewline
DP&1.452 (0.556,1)&0.272 (0.055,1)&0.072 (0.018,1)&$$&1.595 (0.834,1)&0.283 (0.07,1)&0.074 (0.02,1)\tabularnewline
P0&2.634 (0.719,0.55)&0.471 (0.068,0.58)&0.17 (0.028,0.42)&$$&2.79 (0.918,0.57)&0.487 (0.08,0.58)&0.169 (0.029,0.44)\tabularnewline
P1&2.333 (0.716,0.62)&0.39 (0.063,0.7)&0.136 (0.024,0.53)&$$&2.496 (0.997,0.64)&0.399 (0.072,0.71)&0.137 (0.027,0.54)\tabularnewline
P1m&2.168 (0.676,0.67)&0.367 (0.062,0.74)&0.12 (0.023,0.6)&$$&2.321 (0.945,0.69)&0.379 (0.07,0.75)&0.122 (0.022,0.61)\tabularnewline
PK&2.334 (0.727,0.62)&0.391 (0.061,0.7)&0.135 (0.023,0.53)&$$&2.474 (0.984,0.64)&0.399 (0.073,0.71)&0.136 (0.027,0.54)\tabularnewline
PKm&2.165 (0.687,0.67)&0.367 (0.061,0.74)&0.119 (0.022,0.61)&$$&2.319 (0.936,0.69)&0.38 (0.071,0.74)&0.122 (0.025,0.61)\tabularnewline
\hline
\multicolumn{8}{c}{Ex2, $K=7$, $\tau_{\alpha} =0.5$}\tabularnewline
DP&1.644 (0.647,1)&0.298 (0.055,1)&0.094 (0.018,1)&$$&1.716 (0.829,1)&0.315 (0.067,1)&0.1 (0.022,1)\tabularnewline
P0&2.766 (0.895,0.59)&0.452 (0.068,0.66)&0.166 (0.028,0.57)&$$&2.768 (1.006,0.62)&0.472 (0.075,0.67)&0.167 (0.027,0.6)\tabularnewline
P1&2.493 (0.882,0.66)&0.379 (0.063,0.79)&0.131 (0.023,0.72)&$$&2.511 (1.052,0.68)&0.397 (0.073,0.79)&0.138 (0.029,0.72)\tabularnewline
P1m&2.387 (0.873,0.69)&0.363 (0.061,0.82)&0.122 (0.023,0.77)&$$&2.389 (1.015,0.72)&0.379 (0.068,0.83)&0.128 (0.026,0.78)\tabularnewline
PK&2.493 (0.882,0.66)&0.377 (0.057,0.79)&0.13 (0.022,0.72)&$$&2.503 (1.051,0.69)&0.394 (0.069,0.8)&0.138 (0.027,0.72)\tabularnewline
PKm&2.385 (0.869,0.69)&0.36 (0.058,0.83)&0.121 (0.022,0.78)&$$&2.39 (1.02,0.72)&0.379 (0.067,0.83)&0.128 (0.025,0.78)\tabularnewline
\bottomrule
\end{tabular}\end{center}
\end{table}

\begin{table}
\scriptsize
\tabcolsep=3pt
\caption{Performance comparison under model misspecification for the proposed dynamic prediction (DP) method with \textbf{\emph{Frank}}, \textbf{\emph{Gumbel}} or \textbf{\emph{Clayton}} models. Training and testing datasets ($n_{train}=100$, $n_{test}=50$) for Ex1 and Ex2 are generated under the \textbf{\emph{Frank}} copula structure with $35\%$ censoring for $D$. Values in each parenthesis are the empirical standard deviation and ranking of prediction accuracy. The overall rating is the average ranking across all predictive accuracy measures. This table summarizes the results of the proposed DP method from Tables \ref{table:ex1_n100_cr5}-\ref{table:ex2_n100_cr5} and Tables \ref{table: mis clayton}-\ref{table:mis gumbel} in this appendix. }
\label{table:misspecification comparison}
\begin{center}
\begin{tabular}{llllrllll}
\toprule
\multirow{2}*{}& \multicolumn{3}{c}{In-sample}&&\multicolumn{3}{c}{Out-of-sample}&\multicolumn{1}{c}{Overall}\tabularnewline
\cmidrule(lr){2-4}\cmidrule(lr){6-8}
{Method}&{MSPE}&{QPE}&{IBS}&&{MSPE}&{QPE}&{IBS}&{Rating}\tabularnewline
\hline
\multicolumn{8}{c}{Ex1, $K=7$, $\tau_{\alpha} =0.2$}\tabularnewline
DP(Frank)&1.353 (0.619,2)& 0.233 (0.051,1)& 0.049 (0.012,1)&&1.377 (0.78,2)& 0.24 (0.062,1)& 0.054 (0.014,1)&1.33\tabularnewline
DP(Gumbel)& 1.338 (0.55,1)& 0.232 (0.052,1)& 0.048 (0.011,1)&& 1.36 (0.761,1)&0.246 (0.071,2)&0.059 (0.017,2)&1.33\tabularnewline
DP(Clayton)&1.395 (0.623,3)&0.239 (0.057,3)&0.057 (0.018,3)&&1.509 (0.885,3)&0.251 (0.072,3)&0.061 (0.02,3)& 3.00 \tabularnewline
\hline
\multicolumn{8}{c}{Ex1, $K=7$, $\tau_{\alpha} =0.5$}\tabularnewline
DP(Frank)&1.316 (0.548,2)& 0.217 (0.045,1)& 0.058 (0.012,1)&& 1.34 (0.705,1)& 0.221 (0.063,1)& 0.062 (0.015,1)&1.17 \tabularnewline
DP(Gumbel)& 1.281 (0.579,1)& 0.219 (0.048,1)& 0.057 (0.011,1)&&1.386 (0.82,2)&0.228 (0.069,2)&0.068 (0.019,2)&1.50 \tabularnewline
DP(Clayton)&1.354 (0.526,3)&0.231 (0.049,3)&0.066 (0.019,3)&&1.427 (0.856,3)&0.246 (0.074,3)&0.07 (0.018,3)&3.00 \tabularnewline
\hline
\multicolumn{8}{c}{Ex2, $K=7$, $\tau_{\alpha} =0.2$}\tabularnewline
DP(Frank)& 1.39 (0.533,1)&0.261 (0.05,2)& 0.066 (0.013,1)&& 1.517 (0.853,1)& 0.271 (0.068,1)& 0.069 (0.016,1)&1.17\tabularnewline
DP(Gumbel)&1.461 (0.574,3)& 0.256 (0.05,1)& 0.067 (0.017,1)&&1.612 (0.874,3)&0.282 (0.072,2)&0.082 (0.022,3)&2.17 \tabularnewline
DP(Clayton)&1.452 (0.556,2)&0.272 (0.055,3)&0.072 (0.018,3)&&1.595 (0.834,2)&0.283 (0.07,2)&0.074 (0.02,2)&2.33 \tabularnewline
\hline
\multicolumn{8}{c}{Ex2, $K=7$, $\tau_{\alpha} =0.5$}\tabularnewline
DP(Frank)& 1.62 (0.608,1)& 0.293 (0.054,1)& 0.091 (0.018,1)&& 1.542 (0.779,1)& 0.298 (0.072,1)& 0.093 (0.021,1)&1.00\tabularnewline
DP(Gumbel)&1.727 (0.602,3)&0.305 (0.056,3)&0.094 (0.019,2)&&1.884 (0.885,3)&0.336 (0.074,3)&0.11 (0.027,2)& 2.67 \tabularnewline
DP(Clayton)&1.644 (0.647,2)&0.298 (0.055,2)&0.094 (0.018,2)&&1.716 (0.829,2)&0.315 (0.067,2)&0.1 (0.022,2)&2.00\tabularnewline
\bottomrule
\end{tabular}\end{center}
\end{table}

\begin{table}
\scriptsize
\tabcolsep=2pt
\caption{Estimation results of association parameters within the joint distribution function on the observable data region. Training datasets ($n_{train}=100$) are generated under the mixed wedge structure in Ex3. Kendall's tau values representing the association between each $\tildeT_k$ and $D$ for $k=1,\cdots,7$, are specified as  $\tau^{(u)}=0.5$ for the upper wedge and $\tau^{(l)} =0.3,0.5,0.7$ for the lower wedges.}
\label{table:ex3mix_theta}
\begin{center}
\begin{tabular}{llrrrrrrrrrrrrrrrrrrrrrrrr}
\toprule
\multicolumn{2}{c}{Truth}&&\multicolumn{2}{c}{Estimation of $\tau_{\alpha}$}&&\multicolumn{19}{c}{Estimation of $\tau^{(u)}=0.5$}\tabularnewline
\cmidrule(lr){1-2}\cmidrule(lr){4-5}\cmidrule(lr){7-26}
&&&&&&\multicolumn{2}{c}{$\tau_{\theta_1}$}&&\multicolumn{2}{c}{$\tau_{\theta_2}$}&&\multicolumn{2}{c}{$\tau_{\theta_3}$}&&\multicolumn{2}{c}{$\tau_{\theta_4}$}&&\multicolumn{2}{c}{$\tau_{\theta_5}$}&&\multicolumn{2}{c}{$\tau_{\theta_6}$}&&\multicolumn{2}{c}{$\tau_{\theta_7}$}\tabularnewline
\cmidrule(lr){7-8}\cmidrule(lr){10-11}\cmidrule(lr){13-14}
\cmidrule(lr){16-17}\cmidrule(lr){19-20} \cmidrule(lr){22-23}\cmidrule(lr){24-26}
\multicolumn{1}{c}{$\tau_{\alpha}$}&\multicolumn{1}{c}{$\tau^{(l)}$}&&\multicolumn{1}{c}{Bias}&\multicolumn{1}{c}{SD}&&\multicolumn{1}{c}{Bias}&\multicolumn{1}{c}{SD}&&\multicolumn{1}{c}{Bias}&\multicolumn{1}{c}{SD}&&\multicolumn{1}{c}{Bias}&\multicolumn{1}{c}{SD}&&\multicolumn{1}{c}{Bias}&\multicolumn{1}{c}{SD}&&\multicolumn{1}{c}{Bias}&\multicolumn{1}{c}{SD}&&\multicolumn{1}{c}{Bias}&\multicolumn{1}{c}{SD}&&\multicolumn{1}{c}{Bias}&\multicolumn{1}{c}{SD}\tabularnewline
\hline
$0.2$&$0.3$&&$-0.024$&$0.039$&&$0.022$&$0.095$&&$0.017$&$0.097$&&$0.019$&$0.093$&&$0.022$&$0.095$&&$0.022$&$0.095$&&$0.028$&$0.094$&&$0.022$&$0.093$\tabularnewline&$0.5$&&$0.003$&$0.04$&&$0.002$&$0.084$&&$-0.006$&$0.088$&&$-0.009$&$0.087$&&$-0.001$&$0.09$&&$0.009$&$0.087$&&$0.001$&$0.082$&&$0.007$&$0.082$\tabularnewline&$0.7$&&$-0.013$&$0.046$&&$0.064$&$0.084$&&$0.064$&$0.074$&&$0.07$&$0.077$&&$0.07$&$0.074$&&$0.064$&$0.077$&&$0.068$&$0.074$&&$0.069$&$0.069$\tabularnewline
\hline$0.5$&$0.3$&&$-0.058$&$0.049$&&$0.051$&$0.083$&&$0.059$&$0.088$&&$0.061$&$0.087$&&$0.06$&$0.088$&&$0.054$&$0.083$&&$0.058$&$0.086$&&$0.052$&$0.083$\tabularnewline&$0.5$&&$-0.016$&$0.048$&&$0.012$&$0.086$&&$0.013$&$0.083$&&$0.019$&$0.088$&&$0.022$&$0.083$&&$0.006$&$0.081$&&$0.01$&$0.086$&&$0.016$&$0.088$\tabularnewline&$0.7$&&$-0.038$&$0.054$&&$0.059$&$0.077$&&$0.059$&$0.08$&&$0.062$&$0.078$&&$0.059$&$0.08$&&$0.061$&$0.082$&&$0.062$&$0.078$&&$0.064$&$0.086$\tabularnewline
\bottomrule
\end{tabular}\end{center}
\end{table}

\begin{table}
\scriptsize
\tabcolsep=3pt
\caption{Predictive performances. Training ($n_{train}=100$) and testing ($n_{test}=50$) datasets are generated under the mixed wedge structure in Ex3. Kendall's tau values representing the association between each $\tildeT_k$ and $D$ for $k=1,\cdots,7$, are specified as  $\tau^{(u)}=0.5$ for the upper wedge and $\tau^{(l)} =0.3,0.5,0.7$ for the lower wedges.}
\label{table:ex3mix_DP}
\begin{center}
\begin{tabular}{llllrllll}
\toprule
\multirow{2}*{}& \multicolumn{3}{c}{In-sample}&&\multicolumn{3}{c}{Out-of-sample}\tabularnewline
\cmidrule(lr){2-4}\cmidrule(lr){6-8}
{Method}&{MSPE}&{QPE}&{IBS}&&{MSPE}&{QPE}&{IBS}\tabularnewline
\hline
\multicolumn{8}{c}{$\tau_{\alpha} =0.2$, $\tau^{(l)}=0.3$}\tabularnewline
DP&0.342 (0.161,1)&0.104 (0.024,1)&0.023 (0.004,1)&&0.357 (0.209,1)&0.114 (0.028,1)&0.024 (0.005,1)\tabularnewline
P0&0.831 (0.269,0.41)&0.266 (0.037,0.39)&0.054 (0.008,0.43)&&0.856 (0.322,0.42)&0.28 (0.04,0.41)&0.056 (0.008,0.43)\tabularnewline
P1m&0.642 (0.23,0.53)&0.198 (0.035,0.53)&0.04 (0.007,0.58)&&0.665 (0.289,0.54)&0.214 (0.035,0.53)&0.043 (0.007,0.56)\tabularnewline
\hline
\multicolumn{8}{c}{$\tau_{\alpha} =0.2$, $\tau^{(l)}=0.5$}\tabularnewline
DP&0.317 (0.147,1)&0.109 (0.024,1)&0.032 (0.006,1)&&0.369 (0.229,1)&0.12 (0.034,1)&0.035 (0.008,1)\tabularnewline
P0&0.825 (0.259,0.38)&0.267 (0.041,0.41)&0.075 (0.011,0.43)&&0.889 (0.322,0.42)&0.283 (0.046,0.42)&0.079 (0.013,0.44)\tabularnewline
P1m&0.626 (0.22,0.51)&0.198 (0.036,0.55)&0.056 (0.009,0.57)&&0.683 (0.29,0.54)&0.213 (0.041,0.56)&0.059 (0.011,0.59)\tabularnewline
\hline
\multicolumn{8}{c}{$\tau_{\alpha} =0.2$,  $\tau^{(l)}=0.7$}\tabularnewline
DP&0.301 (0.165,1)&0.091 (0.021,1)&0.039 (0.006,1)&&0.335 (0.208,1)&0.1 (0.026,1)&0.042 (0.008,1)\tabularnewline
P0&0.885 (0.292,0.34)&0.263 (0.038,0.35)&0.105 (0.014,0.37)&&0.958 (0.353,0.35)&0.275 (0.043,0.36)&0.109 (0.015,0.39)\tabularnewline
P1m&0.669 (0.253,0.45)&0.189 (0.034,0.48)&0.075 (0.012,0.52)&&0.732 (0.315,0.46)&0.2 (0.035,0.5)&0.079 (0.012,0.53)\tabularnewline
\hline
\multicolumn{8}{c}{$\tau_{\alpha} =0.5$, $\tau^{(l)}=0.3$}\tabularnewline
DP&0.427 (0.172,1)&0.122 (0.034,1)&0.037 (0.007,1)&&0.461 (0.231,1)&0.131 (0.037,1)&0.039 (0.008,1)\tabularnewline
P0&0.878 (0.284,0.49)&0.269 (0.047,0.45)&0.069 (0.012,0.54)&&0.921 (0.347,0.5)&0.28 (0.048,0.47)&0.072 (0.013,0.54)\tabularnewline
P1m&0.686 (0.238,0.62)&0.191 (0.041,0.64)&0.049 (0.011,0.76)&&0.728 (0.318,0.63)&0.201 (0.041,0.65)&0.052 (0.011,0.75)\tabularnewline
\hline
\multicolumn{8}{c}{$\tau_{\alpha} =0.5$, $\tau^{(l)}=0.5$}\tabularnewline
DP&0.409 (0.144,1)&0.131 (0.027,1)&0.041 (0.007,1)&&0.469 (0.272,1)&0.148 (0.033,1)&0.044 (0.009,1)\tabularnewline
P0&0.86 (0.244,0.48)&0.265 (0.035,0.49)&0.072 (0.01,0.57)&&0.925 (0.371,0.51)&0.281 (0.041,0.53)&0.076 (0.011,0.58)\tabularnewline
P1m&0.673 (0.215,0.61)&0.191 (0.032,0.69)&0.052 (0.01,0.79)&&0.743 (0.365,0.63)&0.209 (0.04,0.71)&0.057 (0.012,0.77)\tabularnewline
\hline
\multicolumn{8}{c}{$\tau_{\alpha} =0.5$, $\tau^{(l)}=0.7$}\tabularnewline
DP&0.415 (0.173,1)&0.115 (0.029,1)&0.038 (0.006,1)&&0.459 (0.274,1)&0.127 (0.034,1)&0.04 (0.008,1)\tabularnewline
P0&0.942 (0.272,0.44)&0.266 (0.041,0.43)&0.069 (0.011,0.55)&&0.992 (0.369,0.46)&0.28 (0.043,0.45)&0.072 (0.011,0.56)\tabularnewline
P1m&0.753 (0.249,0.55)&0.188 (0.037,0.61)&0.049 (0.01,0.78)&&0.797 (0.366,0.58)&0.202 (0.037,0.63)&0.053 (0.011,0.75)\tabularnewline
\bottomrule
\end{tabular}\end{center}
\end{table}

\begin{figure}[b]
\centering
\begin{subfigure}{0.7\textwidth}
\centering
\includegraphics[width=0.7\textwidth]{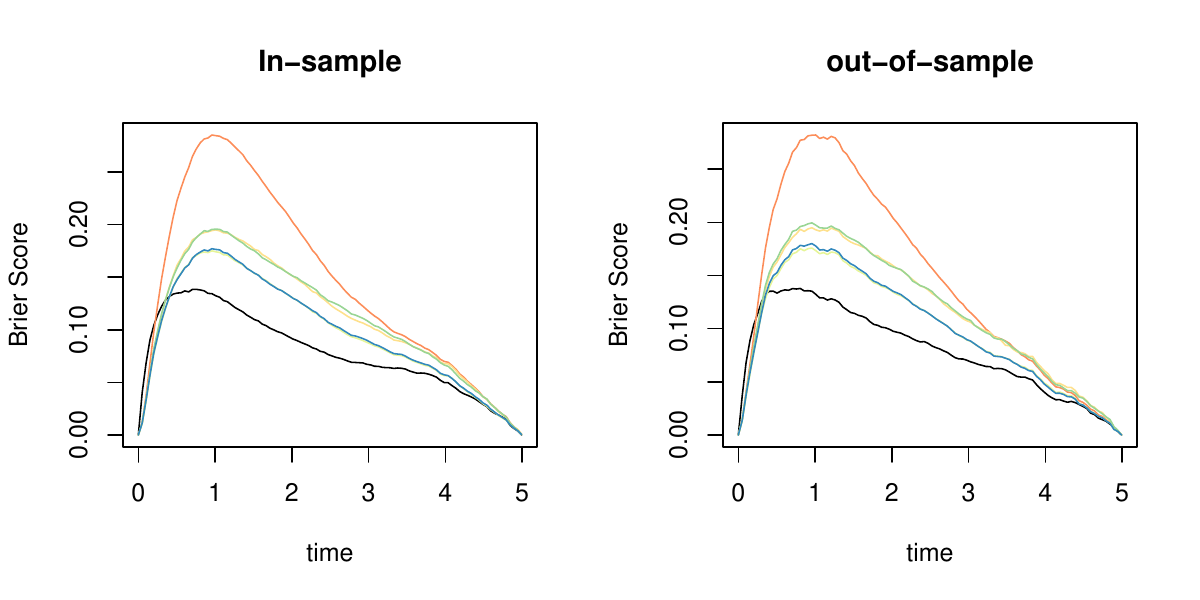}
\caption{$K=3,\tau_{\alpha}=0.2$}
\end{subfigure}
\begin{subfigure}{0.7\textwidth}
\centering
\includegraphics[width=0.7\textwidth]{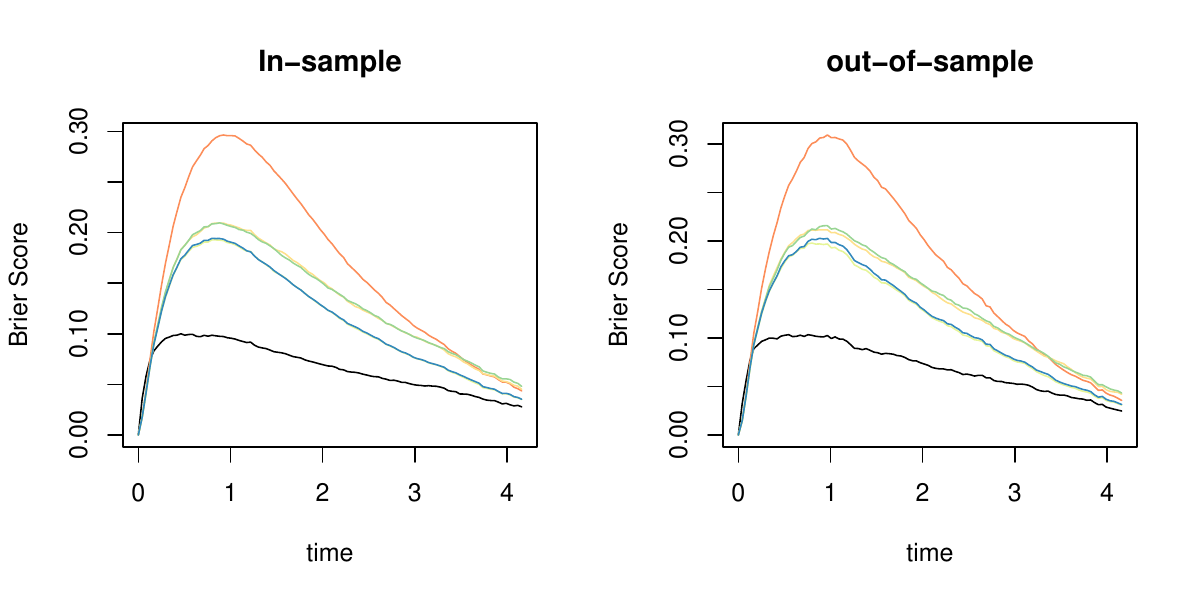}
\caption{$K=7,\tau_{\alpha}=0.2$}
\end{subfigure}

\begin{subfigure}{0.7\textwidth}
\centering
\includegraphics[width=0.7\textwidth]{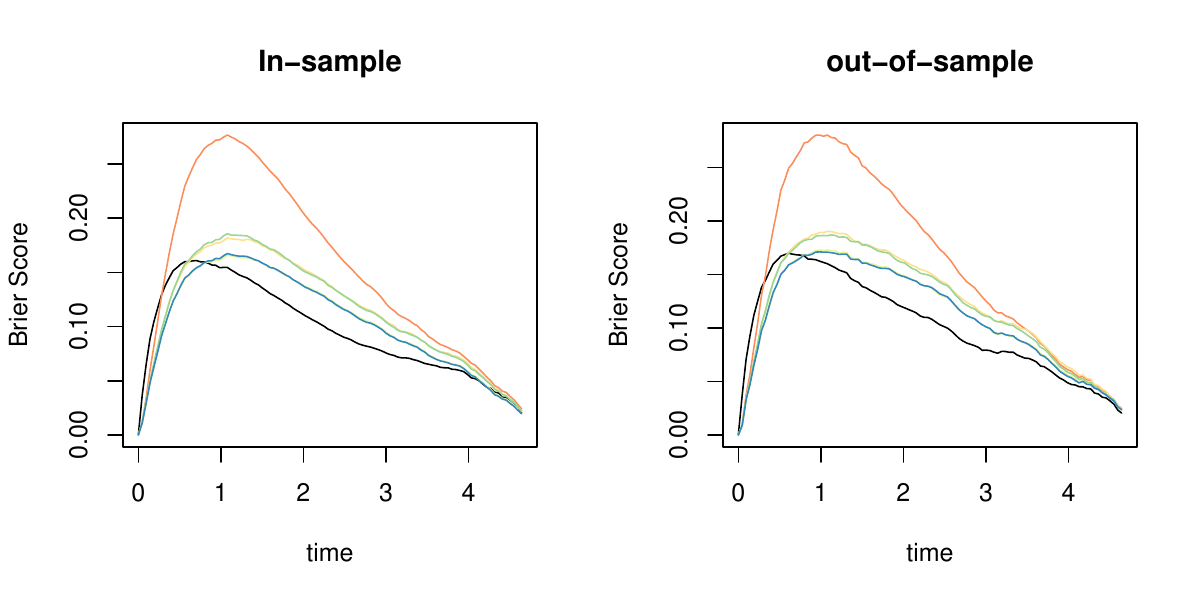}
\caption{$K=3,\tau_{\alpha}=0.5$}
\end{subfigure}

\begin{subfigure}{0.7\textwidth}
\centering
\includegraphics[width=0.7\textwidth]{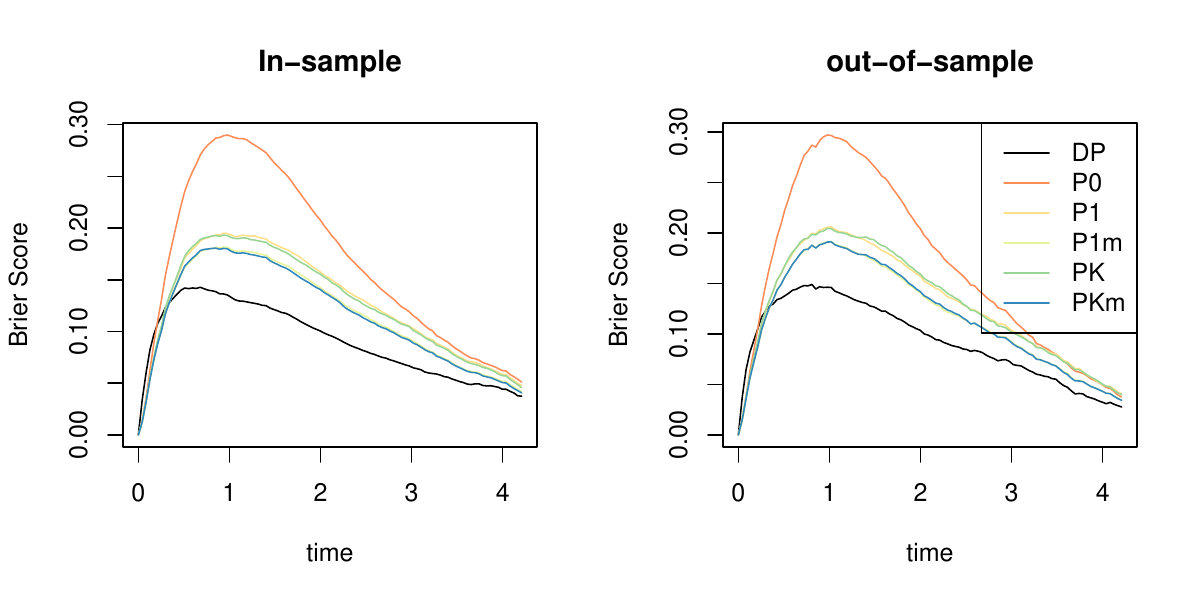}
\caption{$K=7,\tau_{\alpha}=0.5$}
\end{subfigure}
\caption{Averaged in-sample and out-of-sample Brier Score under 200 training and testing datasets of Ex2 ($n_{\operatorname{train}}=100$, $n_{\operatorname{test}}=50$) with $35\%$ censoring for $D$.}
\label{fig:EX2 bs}
\end{figure}

\begin{table}[h]
\scriptsize
\tabcolsep=4pt
\begin{center}
\caption{Average runtime (measured in minutes) per replicate  under EX1.}
\label{tab:runtime}
\begin{tabular}{rrrrrrrr}
\toprule
\multicolumn{4}{c}{$10\%$ censoring for $D$}&\multicolumn{4}{c}{$35\%$ censoring for $D$}\tabularnewline
\cmidrule(lr){1-4}\cmidrule(lr){5-8}
$n_{\operatorname{train}}$&$K$&$\tau_{\alpha}$&Run Time&$n_{\operatorname{train}}$&$K$&$\tau_{\alpha}$&Run Time\tabularnewline
\hline
$100$&$3$&$0.2$&$0.12$&$100$&$3$&$0.2$&$ 4.42$\tabularnewline
    & &$0.5$&$0.10$& &  &$0.5$&$ 4.87$\tabularnewline
    &$7$&$0.2$&$2.29$& &$7$&$0.2$&$ 8.80$\tabularnewline
    & &$0.5$&$2.40$& & &$0.5$&$ 8.35$\tabularnewline
$200$&$3$&$0.2$&$0.38$&$200$&$3$&$0.2$&$10.08$\tabularnewline
    & &$0.5$&$0.41$& & &$0.5$&$ 9.99$\tabularnewline
    &$7$&$0.2$&$8.00$& &$7$&$0.2$&$19.81$\tabularnewline
    & &$0.5$&$6.28$& & &$0.5$&$21.00$\tabularnewline
\bottomrule
\end{tabular}\end{center}
\end{table}

\clearpage
\section{Additional results from the Framingham heart data analysis 
}

An anonymous reviewer pointed out the potential for systematic bias arising from using the first 2,500 individuals as the training set. Accordingly, we randomly selected 2,500 individuals for training and used the remaining individuals for testing. This procedure was repeated 150 times to ensure stability and reduce sensitivity to any single random partition. 
Table \ref{table:realdata alpha split} summarizes estimation results of association parameters $\tau_{\alpha}$ and $\tau_{\theta_k}$ for the Framingham heart data using both random and fixed training–test splits.  The fixed-split results are consistent with those from the random splits, indicating that no evidence of systematic bias induced by the fixed partition used in the earlier version of the paper, and the pattern in the first 2500 patients does not appear to differ substantially from that in the remaining cohort. 

Additionally, to reduce possible variability arising from a single data partition, we considered the general random cross-validation scheme, using random training–test splits with 2,500 individuals randomly assigned to the training set. We have also included a comparison of the proposed dynamic prediction (DP) method and its variations (DP1 and DP2) with competing algorithms (P0,P3,P3m,P6,P6m). DP uses all seven intermediate events; DP1 excludes HYP, and DP2 excludes AP and HYP. P0 is the landmark prediction method using the Kaplan-Meier (KM) estimator of $S_D$; P$k$ and P$k$m are survival prediction algorithms incorporating the $k$-th intermediate outcome only. 
Since MSPE and QPE metrics are based on true values $D_i$ which are known in simulation studies but unknown in practice. These two metrics are adjusted with inverse probability of censoring weights (IPCW) to evaluate agreement between predicted and observed survival times. 
We evaluated the performance of different methods using metrics including MSPE and QPE adjusted with IPCW, IBS and time-dependent AUC (at 15 years). 

Figure \ref{fig:framingham split} displays boxplots of these performance metrics for the testing datasets across different methods considered. 
As shown in the figure, the proposed DP/DP1/DP2 methods achieve lower mean squared and quantile prediction errors than the competing methods 
and also exhibit clearly higher discrimination ability, as indicated by greater AUC values.  All methods considered yield low IBS values, while the proposed DP/DP1/DP2 methods demonstrate slightly higher IBS, 
likely due to insufficient follow-up, as 75\% of patients were lost to follow-up and fixed censored at 8766\emph{th} day (24.02 years) in the study.

To further evaluate agreement between predicted and observed survival probabilities, we conducted a calibration analysis. Calibration slopes,  constructed using 100 bootstrap replicates of the training samples under the Frank copula dependence structure, are $0.939$ (95$\%$ CI $[0.913, 0.964])$ for the training data and $0.944$ (95$\%$ CI $[0.876, 1.012])$ for the test data.
Both slopes are close to the ideal value of 1.0 and consistent between the training and test datasets, demonstrating that our model is well-calibrated and yields reliable predicted survival probabilities.

\begin{table}[htbp]
\tiny
\tabcolsep=4pt
\begin{center}
\caption{Estimation results of association parameters $\tau_{\alpha}$ and $\tau_{\theta_k}$ in Frank copula model for the Framingham heart data, comparing a random training–test split (2,500 individuals for training) with the split based on the first 2,500 individuals.  Values in each bracket are 95$\%$ confidence intervals based on 200 bootstrapping replicates.}
\label{table:realdata alpha split}
\begin{tabular}{lccccccccccccccc}
\toprule
\multicolumn{1}{l}{}&\multicolumn{1}{l}{}&\multicolumn{1}{c}{$\tau_{\alpha}$}& \multicolumn{7}{c}{$\tau_{\theta_k}$}\tabularnewline
\cmidrule(lr){4-10}
&&&\multicolumn{1}{c}{AP}&\multicolumn{1}{c}{CHD}&\multicolumn{1}{c}{MIFC}&\multicolumn{1}{c}{CVD}&\multicolumn{1}{c}{STRK}&\multicolumn{1}{c}{HYP}&\multicolumn{1}{c}{MI}\tabularnewline
\hline
Fixed split&Estimate  &0.371  &0.300& 0.482& 0.694& 0.608& 0.680& 0.113& 0.584\tabularnewline
&confidence intervals& [0.314,0.426]  & [0.238,0.36] & [0.434,0.522] & [0.641,0.746] & [0.568,0.652] & [0.615,0.75] & [0.07,0.15] & [0.514,0.651]\tabularnewline
\hline
Random split
&1st Qu.&$0.357$&$0.307$&$0.480$&$0.686$&$0.595$&$0.660$&$0.122$&$0.565$\tabularnewline
&Median&$0.371$&$0.315$&$0.485$&$0.692$&$0.600$&$0.668$&$0.127$&$0.572$\tabularnewline
&Mean&$0.373$&$0.314$&$0.485$&$0.692$&$0.600$&$0.668$&$0.126$&$0.572$\tabularnewline
&3rd Qu.&$0.384$&$0.322$&$0.489$&$0.697$&$0.604$&$0.678$&$0.130$&$0.581$\tabularnewline
\bottomrule
\end{tabular}\end{center}
\end{table}

\begin{figure}[t]
\centering
\centering
\includegraphics[width=.9\textwidth]{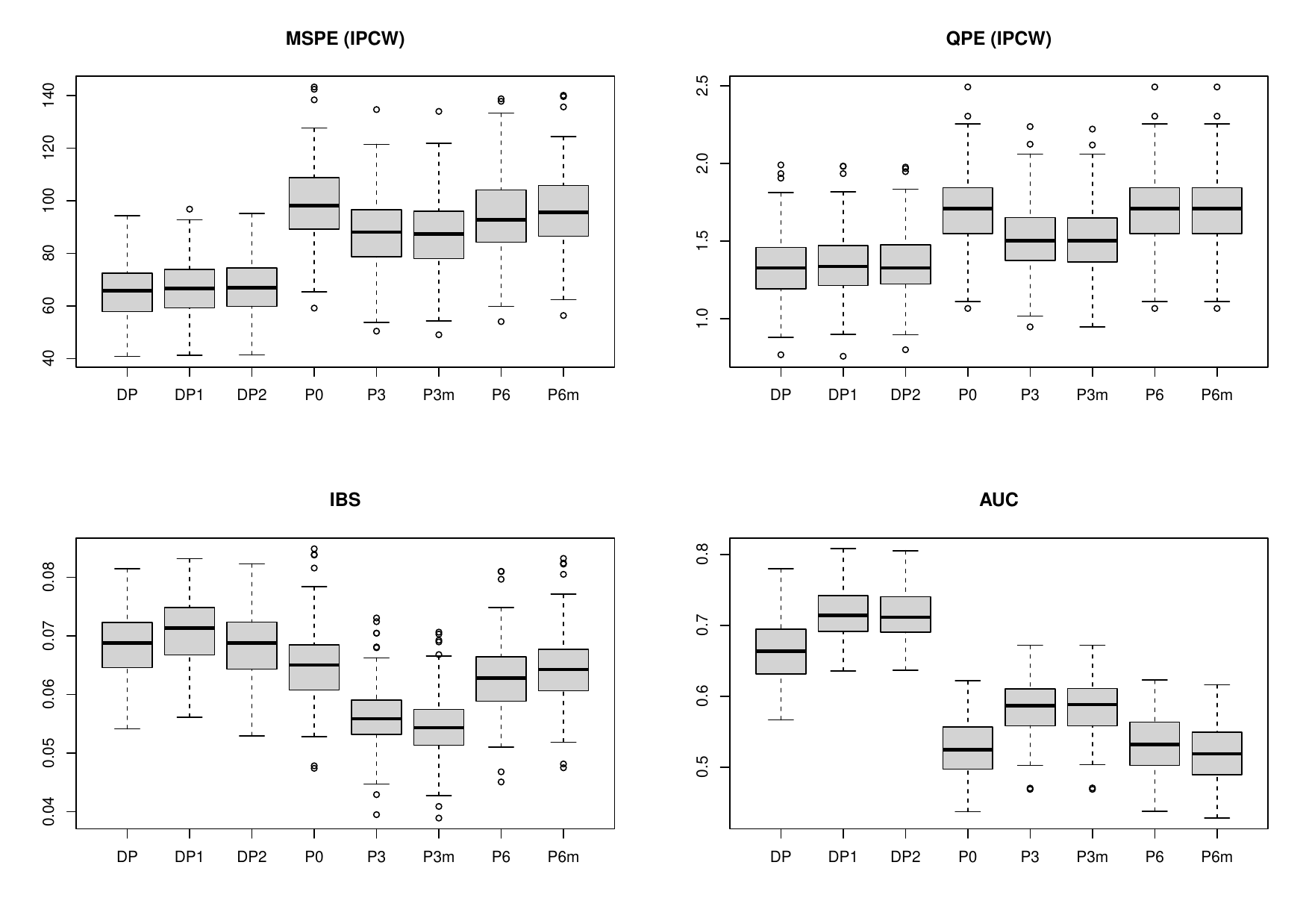}
\caption{Predictive performance of different methods (the proposed DP/DP1/DP2 and competing algorithms P0/P3/P3m/P6/P6m) evaluated on the testing datasets across repeated random splits for the Framingham Heart Study. Top panels: mean squared prediction error (MSPE, left) and quantile prediction error (QPE, right) adjusted by IPCW. Bottom panels: integrated Brier score (IBS, left) and time-dependent AUC at 15 years (right).}
\label{fig:framingham split}
\end{figure}

\begin{figure}[htbp]
\centering
\begin{subfigure}{0.32\textwidth}
\centering
\includegraphics[width=.9\textwidth]{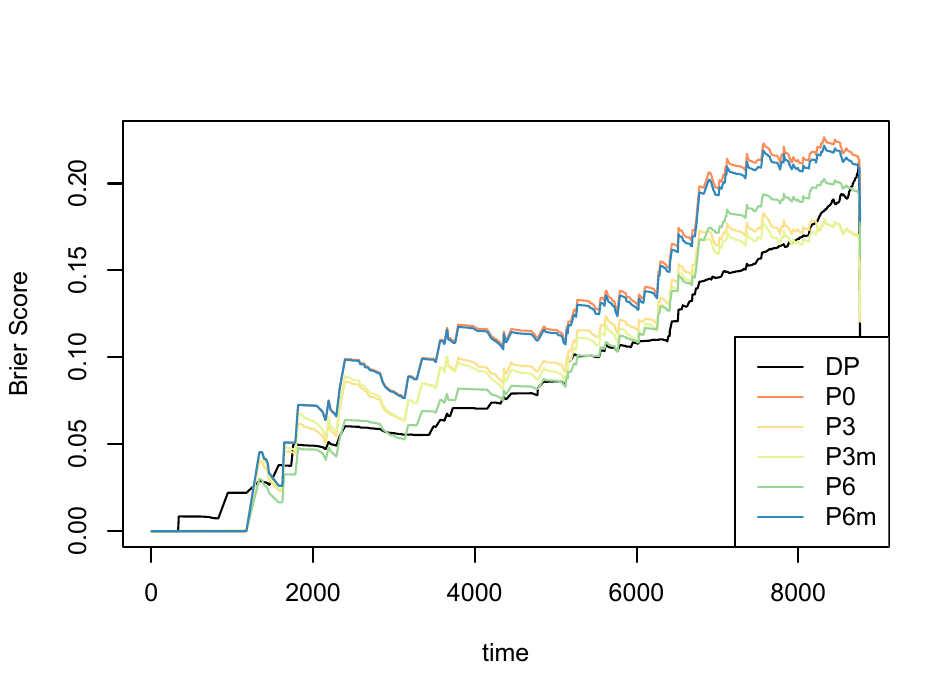}
\caption{Frank copula}
\end{subfigure}
\begin{subfigure}{0.32\textwidth}
\centering
\includegraphics[width=.9\textwidth]{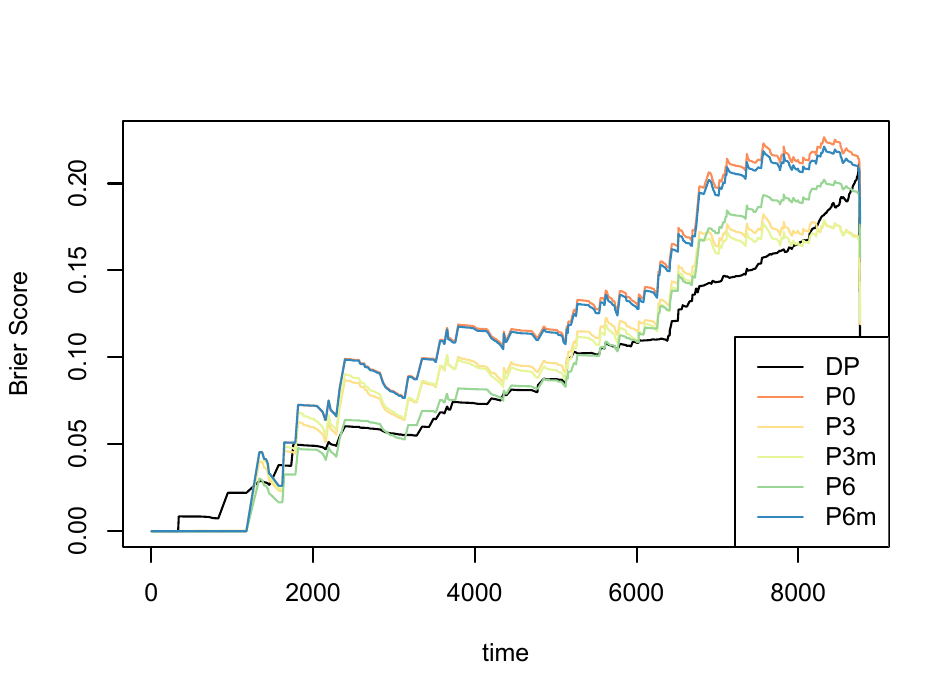}
\caption{Clayton copula}
\end{subfigure}
\begin{subfigure}{0.32\textwidth}
\centering
\includegraphics[width=.9\textwidth]{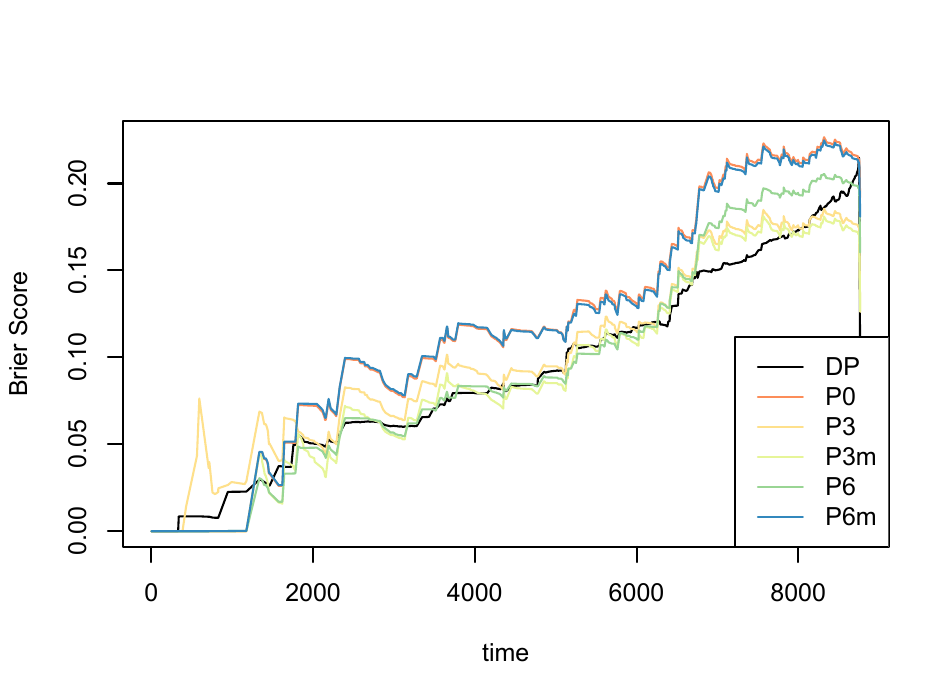}
\caption{Gumbel copula}
\end{subfigure}
\caption{Out-of-sample Brier Score for the testing data (the last 333 individuals) of the  Framingham heart data and under different copula structures.}
\label{fig:framingham bs}
\end{figure}


\end{document}